\documentclass[11pt,twoside]{article}

\usepackage{authblk}

\usepackage{fullpage}

\usepackage{epsf}
\usepackage{fancyhdr}
\usepackage{graphics}
\usepackage{graphicx}
\usepackage{psfrag}

\usepackage[linesnumbered,ruled]{algorithm2e}
\DontPrintSemicolon

\usepackage{color}

\usepackage{amsthm}
\usepackage{amsfonts}
\usepackage{amsmath}
\usepackage{amssymb,bbm}
\usepackage{mathrsfs}

\usepackage{tikz-cd} 

\usepackage[colorlinks]{hyperref}

\usepackage{chngpage}
\usepackage{tabularx}%

\usepackage{enumitem}
\usepackage{booktabs}
\usepackage{caption,subcaption}

\usepackage{mathtools}

\DeclareMathOperator{\Diag}{Diag}

\newcommand{\prob}{\mathbb{P}}

\newcommand{\RN}[1]{%
  \textup{\uppercase\expandafter{\romannumeral#1}}%
}

\newcommand{\barT}{\widebar \cT}
\newcommand{\barPhi}{\widebar\Phi}
\newcommand{\barG}{\overline{G}}
\newcommand{\dikinwalk}{{Dikin walk}\xspace}
\newcommand{\ballwalk}{{ball walk}\xspace}
\newcommand{\hitandrun}{{hit-and-run}\xspace}
\newcommand{\geodesicwalk}{{geodesic walk}\xspace}
\newcommand{\RHMC}{{RHMC}\xspace}

\newcommand{\couplingprob}{\gamma} 
\newcommand{\acceptbeforemin}[2]{\Lambda(#1, #2)}
\newcommand{\distdiff}[2]{\Gamma(#1,#2)}

\newcommand{\Glog}{G}

\newcommand{\rinf}{R_o}
\newcommand{\rinfbd}{\widehat{R_o}}

\newcommand{\tcT}{\widetilde\cT}
\newcommand{\crossbd}{C_K}  
\newcommand{\euclidbd}{C_E}

\newcommand{\inr}{\widetilde{r}}
\newcommand{\outr}{\widetilde{R}}
\newcommand{\rasc}{r_{\text{ASC}}}

\newcommand{\dH}{d^{\cH}}
\newcommand{\dn}{\mathfrak{d}}
\newcommand{\spectr}{\eta}

\newcommand{\Hlogmetric}{unregularized logarithmic metric\xspace}

\newcommand{\soft}{soft-threshold\xspace}
\newcommand{\softdw}{soft-threshold Dikin walk\xspace}

\newcommand{\regdw}{regularized Dikin walk\xspace}
\newcommand{\softlog}{soft-threshold  metric\xspace}
\newcommand{\reglewis}{regularized Lewis metric\xspace}

\newcommand{\Ggamma}{G^{(\gamma)}(x)}
\newcommand{\Hgamma}{H^{(\gamma)}(x)}

\newcommand{\cMR}{\mathcal{M}_R^{\delta}}
\newcommand{\deri}{\mathcal{D}}

\newcommand{\cB}{\mathcal{B}}

\newcommand{\cH}{\mathcal{H}}
\newcommand{\cT}{\mathcal{T}}
\newcommand{\cX}{\mathcal{X}}

\newcommand{\cE}{\mathcal{E}}
\newcommand{\cA}{\mathcal{A}}
\newcommand{\cM}{\mathcal{M}}

\newcommand{\cP}{\mathcal{P}}

\newcommand{\TT}{\mathcal{T}}

\newcommand{\defn}{:=}

\newtheoremstyle{named}{}{}{\itshape}{}{\bfseries}{.}{.5em}{\thmnote{#3's }#1}
\theoremstyle{named}

\theoremstyle{plain}

\newtheorem{theorem}{Theorem}

\newtheorem{lemma}{Lemma}
\newtheorem{corollary}{Corollary}
\newtheorem{definition}{Definition}

\newtheorem{fact}{Fact}

\newlength{\widebarargwidth}
\newlength{\widebarargheight}
\newlength{\widebarargdepth}
\DeclareRobustCommand{\widebar}[1]{%
  \settowidth{\widebarargwidth}{\ensuremath{#1}}%
  \settoheight{\widebarargheight}{\ensuremath{#1}}%
  \settodepth{\widebarargdepth}{\ensuremath{#1}}%
  \addtolength{\widebarargwidth}{-0.3\widebarargheight}%
  \addtolength{\widebarargwidth}{-0.3\widebarargdepth}%
  \makebox[0pt][l]{\hspace{0.3\widebarargheight}%
    \hspace{0.3\widebarargdepth}%
    \addtolength{\widebarargheight}{0.3ex}%
    \rule[\widebarargheight]{0.95\widebarargwidth}{0.1ex}}%
  {#1}}

\makeatletter
\long\def\@makecaption#1#2{
        \vskip 0.8ex
        \setbox\@tempboxa\hbox{\small {\bf #1:} #2}
        \parindent 1.5em  
        \dimen0=\hsize
        \advance\dimen0 by -3em
        \ifdim \wd\@tempboxa >\dimen0
                \hbox to \hsize{
                        \parindent 0em
                        \hfil
                        \parbox{\dimen0}{\def\baselinestretch{0.96}\small
                                {\bf #1.} #2
                                }
                        \hfil}
        \else \hbox to \hsize{\hfil \box\@tempboxa \hfil}
        \fi
        }
\makeatother

\long\def\comment#1{}
\definecolor{battleshipgrey}{rgb}{0.52, 0.52, 0.51}
\definecolor{darkgray}{rgb}{0.66, 0.66, 0.66}
\definecolor{darkgreen}{rgb}{0.0, 0.2, 0.13}
\definecolor{darkspringgreen}{rgb}{0.09, 0.45, 0.27}
\definecolor{dukeblue}{rgb}{0.0, 0.0, 0.61}
\definecolor{olivedrab7}{rgb}{0.24, 0.2, 0.12}
\definecolor{darkblue}{rgb}{0.0, 0.0, 0.55}
\definecolor{darkscarlet}{rgb}{0.34, 0.01, 0.1}
\definecolor{candyapplered}{rgb}{1.0, 0.03, 0.0}
\definecolor{ao(english)}{rgb}{0.0, 0.5, 0.0}
\definecolor{applegreen}{rgb}{0.55, 0.71, 0.0}

\newcommand{\todo}[1]{{\bf{{\textcolor{green}{{TODO}}}}}}

\DeclareMathOperator{\diag}{diag}
\DeclareMathOperator{\var}{Var}

\DeclareMathOperator{\trace}{Tr}

\DeclareMathOperator{\vol}{vol}

\newcommand{\real}{\ensuremath{\mathbb{R}}}

\newcommand{\Ind}[1]{\ensuremath{\mathbb{I}_{\left\{ #1 \right \}}}} 

\newcommand{\ball}{\ensuremath{\mathbb{B}}}
\newcommand{\Exs}{\ensuremath{{\mathbb{E}}}}

\newcommand{\Normal}{\ensuremath{\mathcal{N}}}

\DeclareMathOperator{\Var}{Var}

\DeclarePairedDelimiterX{\infdivx}[2]{(}{)}{%
  #1\;\delimsize\|\;#2%
}
\newcommand{\kldiv}{\text{KL}\infdivx}

\newcommand{\Ot}{\widetilde{O}}

\newcommand{\brackets}[1]{\left[ #1 \right]}
\newcommand{\parenth}[1]{\left( #1 \right)}

\newcommand{\braces}[1]{\left\{ #1 \right \}}
\newcommand{\abss}[1]{\left| #1 \right |}

\newcommand{\angles}[1]{\left\langle #1 \right \rangle}

\newcommand{\floors}[1]{\left\lfloor #1 \right \rfloor}
\newcommand{\tp}{^\top}

\newcommand{\vecnorm}[2]{\left\| #1\right\|_{#2}}

\title{Regularized Dikin Walks for Sampling Truncated Logconcave Measures, Mixed Isoperimetry and Beyond Worst-Case Analysis}
\author{Minhui Jiang}
\author{Yuansi Chen}
\affil{ETH Z\"urich}
\date{}

\begin{document}

\maketitle

\begin{abstract}

We study the problem of drawing samples from a logconcave distribution truncated on a polytope, motivated by computational challenges in Bayesian statistical models with indicator variables, such as probit regression. Building on interior point methods and the Dikin walk for sampling from uniform distributions, we analyze the mixing time of regularized Dikin walks. Our contributions are threefold. First, for a logconcave and log-smooth distribution with condition number $\kappa$, truncated on a polytope in $\real^n$ defined with $m$ linear constraints, we prove that the soft-threshold Dikin walk mixes in $\Ot((m+\kappa)n)$ iterations from a warm initialization. It improves upon prior work which required the polytope to be bounded and involved a bound dependent on the radius of the bounded region. Moreover, we introduce the \regdw using Lewis weights for approximating the John ellipsoid. We show that it mixes in $\Ot((n^{2.5}+\kappa n)$. Second, we extend the mixing time guarantees mentioned above to weakly log-concave distributions truncated on polytopes, provided that they have a finite covariance matrix. Third, going beyond worst-case mixing time analysis, we demonstrate that \softdw can mix significantly faster when only a limited number of constraints intersect the high-probability mass of the distribution, improving the $\Ot((m+\kappa)n)$ upper bound to $\Ot(m + \kappa n)$. Additionally, per-iteration complexity of \regdw and ways to generate a warm initialization are discussed to facilitate practical implementation. 

\end{abstract}

\tableofcontents

\section{Introduction}\label{sec_intro}
Sampling from high-dimensional distributions under constraints is an important computational challenge in fields such as Bayesian statistics, computer science, and systems biology. In Bayesian statistics, posterior sampling is crucial for estimation and uncertainty quantification. There, posterior distributions are often truncated on a subset of $\real^n$, rendering standard sampling algorithms for smooth distributions such as Langevin algorithms or Hamiltonian Monte Carlo ineffective. Examples of such posteriors include ordered exponential family models, censored data models, and ordered linear models~\cite{gelfand_bayesian_1992}. In addition, Bayesian probit regression models~\cite{albert_bayesian_1993}, isotonic regression~\cite{neelon_bayesian_2004}, spatial probit models~\cite{lesage_new_2011}, and tobit models~\cite{anceschi2023bayesian} all require sampling from truncated normal distributions as a subroutine. In computer science, efficient volume computation of convex bodies relies heavily on sampling from distributions truncated on these bodies \cite{lovasz1993random,Lov06hrcorner,lee_geodesic_2017,cousins_gaussian_2018,pmlr-v247-kook24b}. In systems biology, constraint-based modeling of human metabolic networks uses polytopes to define solution spaces, making high-dimensional polytope-constrained sampling essential for these biological tasks~\cite{wiback_monte_2004,lewis_constraining_2012,saa_ll-achrb_2016,haraldsdottir_chrr_2017,heirendt_creation_2019}. To address these needs, many sampling algorithms are developed and analyzed. Although uniform sampling on convex bodies or polytopes is well understood for algorithms such as Ball walk, hit-and-run \cite{Vempala2005Survey}, and Dikin walks \cite{sachdeva2016mixing}, significant challenges remain to develop and understand new algorithms that leverage additional structure in non-uniform distributions.

One important motivating family of distributions in Bayesian statistics is the SUN (unified skew normal) distribution~\cite{arellanovalle_unification_2006,anceschi2023bayesian}. It can serve as conjugate priors/posteriors for various generalized linear models, including probit models~\cite{durante_conjugate_2019}, dynamic multivariate probit~\cite{fasano_closed-form_2021}, and tobit models \cite{anceschi2023bayesian}. Efficient sampling from a SUN distribution accelerates Bayesian inference for these models. Since any SUN distribution can be represented as the sum of a Gaussian vector and a truncated Gaussian vector \cite{anceschi2023bayesian}, and sampling Gaussian is easy, sampling from SUN distributions reduces to the problem of sampling from truncated Gaussian vectors. This way of sampling SUN distributions is called ``i.i.d. sampling'' in~\cite{anceschi2023bayesian}, which has gained increasing popularity when compared to the data augmentation method combined with Gibbs sampling in~\cite{albert_bayesian_1993}. While the latter enjoys a small per iteration cost, the convergence of Gibbs sampling can be slow and is in general difficult to quantify~\cite{johndrow_mcmc_2019,qin_convergence_2019}. However, the convergence of i.i.d. sampling is fully determined by the convergence in sampling from truncated Gaussian vectors, which leads us to study the problem of efficient sampling from truncated Gaussian distributions.

Previous experiments show that sampling from SUN distributions through the ``i.i.d. sampling'' significantly outperforms directly sampling via Hamiltonian Monte Carlo~\cite{anceschi2023bayesian}. However, while the proposed algorithm works well for small sample size $N$ and relatively large feature dimension $p$, it becomes computationally prohibitive for datasets with large sample size $N$, even with state-of-the-art methods~\cite{anceschi2023bayesian}. In these scenarios, deterministic approximations of the exact posterior were considered, such as variational Bayes (VB)~\cite{blei2017variational,fasano_class_2022},  expectation-propagation~\cite{minka2013expectation}, and Laplace approximation. However, these methods can struggle with approximation quality. For example, VB in \cite{fasano_class_2022} only performs well when $p > N$. This leaves a gap for determining efficient sampling algorithms in the moderate or large $N$ and moderate $p$ regime.

The challenge in sampling from SUN distributions with moderate or large $N$ stems from the fact that the truncated Gaussian vector has both its dimension and the number of inequality constraints linearly scales in $N$. Specifically, for a Gaussian with mean vector $\mu$ and covariance matrix $\Sigma$, the truncated distribution in SUN is given by
\begin{equation}\label{eq_motivation_truncated}
    \pi(x)\propto \Ind{x\geq 0}\exp\brackets{-\frac12{(x-\mu)\tp\Sigma^{-1}(x-\mu)}}.
\end{equation}
When $N$ reaches the scale of thousands, given that each iteration costs at least $N^2$ additions or multiplications, it is an open question whether sampling from a truncated Gaussian distribution in $\real^N$ can be done in less than $N^3$ total number of additions or multiplications. 

With the above motivation in mind, in this paper, we study mixing time guarantees of regularized Dikin walks for sampling strongly logconcave and log-smooth distributions truncated on polytopes. Dikin walks form a family of Markov chain sampling algorithms inspired by interior point methods in optimization, which has already shown promising performance for sampling truncated uniform distributions~\cite{kannan_random_2012}. Regularized Dikin walks are their variants to tackle non-uniform distributions. However, the existing mixing time guarantees such as those in~\cite{NEURIPS2023_mangoubi} match neither those in truncated uniform sampling nor those in unconstrained sampling when these scenarios are treated as special cases. As we explain at the end of this section after introducing appropriate notations, we improve the existing mixing time analysis and compare its computational performance to previous works. Through theoretical understanding of the mixing time, we also develop a new isoperimetric inequality that combines the Euclidean metric and the Hilbert metric on a convex set, which enables us with a better mixing time analysis.

\subsection{Problem Setup and Notations}\label{sec_setup_and_notations}

We start by defining a target distribution of interest. Let $K\subseteq \real^n$ denote a (possibly unbounded) convex set. Let $\Pi$ denote a target distribution over $\real^n$, we assume that $\Pi$ has a density $\pi(x)\propto \mathbf{1}_K(x) e^{-f(x)}$ with respect to the Lebesgue measure. Throughout this paper, we assume the function $f:\real^n\to\real$ to be twice differentiable. 
\paragraph{Truncated strongly logconcave log-smooth target distribution} Specifically, the target distribution is expressed as
\begin{equation}\label{eq_distri}
\begin{aligned}
    &\pi(x)\propto \mathbf{1}_{K}(x)\exp(- f(x)),\\ 
    \text{where }&K \text{ is open and convex}, \alpha I\preceq \nabla^2 f(x) \preceq \beta I, \ \forall x \in \real^n.
\end{aligned}
\end{equation}
We use \textit{$\beta$-log-smooth} to refer to the condition $\nabla^2 f(x) \preceq \beta I, \forall x \in K$. Similarly, \textit{$\alpha$-strongly logconcave} refers to the condition $\nabla^2 f(x) \succeq \alpha I, \forall x \in K$ and $\alpha > 0$. Furthermore, we define the \textit{condition number} $\kappa$ to be $\kappa\defn \beta/\alpha$. 

When $K$ is a polytope with $m$ linear constraints, we let 
\begin{equation*}
K\defn\braces{x\bigg|a_i\tp x-b_i>0 \text{ for }i\in [m]},
\end{equation*}
where $a_i\in\real^n$ and $b_i\in\real$ for $i\in [m]$, we let $A\in\real^{m\times n}$ and $b\in\real^m$ with $A\tp\defn (a_1\ldots a_m)$ and $b=(b_1\ldots b_m)\tp$. As a shorthand, we write $K$ as $K = \braces{x|Ax>b}$. 

In the special case of sampling a Gaussian distribution truncated on a polytope $K = \braces{x|Ax>b}$, the negative log-density $f(x) \defn (x-\mu) \tp \Sigma (x-\mu)$, where $\mu$, $\Sigma$ are the mean and covariance of the untruncated Gaussian. Note that without loss of generality, we can assume that the condition number $\kappa$ is $1$. This is because one can sample an affine transformed distribution and recover samples of the original distribution by an inverse affine transform. More precisely, consider the affine transformation $\cA :x \mapsto \Sigma^{-\frac12}(x-\mu)$, which transforms a random variable $X \sim \pi$ to $Y \defn \cA(X)$. The density of $Y$ satisfies
\begin{equation}\label{eq_affine}
    \pi_Y(y)\propto \pi_X(\cA^{-1}(y))\propto \exp\parenth{-\frac12 y\tp y}\Ind{A\Sigma^{\frac12}y>b-A\mu}.
\end{equation}
Let $\widetilde{A}=A\Sigma^{\frac12}$ and $\widetilde{b}=b-A\mu$, we can first sample $\exp\parenth{-\frac12 y\tp y}$ truncated on $\braces{\widetilde{A} x > \widetilde{b}}$ and apply an inverse affine transform to obtain samples from $\pi$. 

\paragraph{Truncated weakly logconcave log-smooth target distribution} Later, we loose the assumption of $\alpha$-strongly logconcaveness, and instead assume that the target distribution $\Pi$ has a bounded covariance matrix. To be precise, we assume 
\begin{equation}\label{eq_distri_weakly}
\begin{aligned}
&\pi(x)\propto \mathbf{1}_K(x)\exp(-f(x)),\\
\text{where }&K\subseteq \real^n \text{ is open \& convex, }0\preceq \nabla^2 f(x)\preceq \beta I, \forall x\in \real^n,\\
\text{ and } \Pi& \text{ has a bounded covariance matrix }\Sigma_{\pi}.
\end{aligned}
\end{equation}


We say a function $q:\real^n\to [0,\infty)$ is \textit{logconcave} if  $\log q(x)$ is a concave function over $\real^n$ (note that we allow $\log q(x)=-\infty$ when $q(x)=0$). We say a probability distribution $\Pi$ on $\real^n$ is \textit{logconcave} if it admits a density function $\pi$ with respect to the Lebesgue measure and $\pi$ is a logconcave function.

A probability distribution $\Pi$ on $\real^n$ is said to be \textit{more logconcave than Gaussian with covariance $\frac1\alpha I_n$} if its density $\pi$ satisfies
\begin{equation}\label{eq_more_log_concave}
    \pi(x)\propto \exp\parenth{-\frac\alpha2\vecnorm{x}{2}^2}\cdot q(x),
\end{equation}
where $q$ is a logconcave function over $\real^n$. Note that $q$ is not required to be continuous. As a special case, our target distribution in Eq.~\eqref{eq_distri} is more logconcave than Gaussian with covariance $\frac{1}{\alpha}I_n$ where we take $q(x)\defn \mathbf{1}_K(x)\exp(-f(x)+\frac{\alpha}{2}\vecnorm{x}{2}^2)$.

\paragraph{Markov chain basics}
Next, we introduce Markov chains notations. Assume $\cX$ is a Borel measurable subset of $\real^n$, and $\mathcal{B}(\cX)$ denotes the Borel $\sigma$-algebra over $\cX$. Then a Markov chain on $\cX$ is characterized by a \textit{transition kernel} $\cT$, which is a function $\cT:\cX\times \mathcal{B}(\cX)\to \real^+$ satisfying
\begin{itemize}
    \item For each $x\in \cX$, $\cT(x,\cdot)$ is a probability measure over $(\cX,\mathcal{B}(\cX))$.
    \item For each $B\in\mathcal{B}(\cX)$, $\cT(\cdot,B)$ is a $\mathcal{B}(\cX)$-measurable function over $\cX$. 
\end{itemize}
In this paper, we write $\cT_x(B)\defn \cT(x,B)$ and use $\cT_x$ to denote the probability measure at $x$.  A transition kernel $\cT$ can also be seen as an operator on probability measures. For a probability measure $\mu_0$ over $\cX$, we define a new probability measure $\cT(\mu_0)$ after applying one step of Markov chain by:
\begin{equation*}
\cT(\mu_0) (B)\defn \int_{x\in\cX}\cT_x(B) \mu_0(dx),
\end{equation*}
for any $B \in \mathcal{B}(\cX)$, and it is easy to verify that $\cT(\mu_0)$ is a probability measure. Applying $\cT$ recursively on $\mu_0$ gives us the measure after applying $k$ steps of Markov chain $\cT^{k}(\mu_0)\defn \cT\parenth{\cT^{k-1}(\mu_0)}$, denoted as $\mu_k$ for short. $\Pi$ is called a \textit{stationary distribution} of $\cT$ if $\cT(\Pi) = \Pi$. 


To quantify how fast $\cT^k(\mu_0)$ converges to the target distribution $\Pi$, we need a distance between two probability measures. Given two probability measures $\prob$ and $\mathbb{Q}$ over $\real^n$, the \textit{total variation distance}, or TV-distance for short, between them is defined as:
\begin{equation*}
\vecnorm{\prob-\mathbb{Q}}{TV}\defn \underset{B\in \cB(\real^n)}\sup \abss{\prob(B)-\mathbb{Q}(B)}.
\end{equation*}
For an error tolerance $\epsilon>0$, given a Markov chain with transition kernel $\cT$ and stationary distribution $\Pi$, we define its \textit{mixing time} as the number of iterations it takes to be $\epsilon$-close to its stationary distribution. More precisely, $T_\text{mix}(\epsilon;\mu_0)$ is the smallest integer $k$ such that $\vecnorm{\cT^{k}(\mu_0)-\Pi}{TV}\leq \epsilon$:
\begin{equation*}
T_\text{mix}(\epsilon;\mu_0)\defn \inf\braces{k\in \mathbb{Z}^+\big|\vecnorm{\cT^{k}(\mu_0)-\Pi}{TV}\leq \epsilon}.
\end{equation*}
We wish to prove an upper bound of $T_\text{mix}(\epsilon;\mu_0)$, as a function of inputs to our algorithm, such as the dimension of the distribution $n$, the number of constraints of the polytope $m$, the initial distribution $\mu_0$ and the error tolerance $\epsilon$.

We say an initial distribution $\mu_0$ is $M$-\textit{warm} with respect to the target distribution $\Pi$ if
\begin{equation*}
\underset{B\in\cB(\real^n)}{\sup}\brackets{\frac{\mu_0(B)}{\Pi(B)}}\leq M.
\end{equation*}

\paragraph{Other notation} We use $I$ to denote the identity matrix, and $I_m$ or $I_n$ to put an emphasis on the dimension of the identity matrix. Given a vector $v\in\real^k$, we use $\Diag(v)$ to denote the diagonal matrix in $\real^{k\times k}$ that has $v$ as its diagonal elements. Similarly, given a matrix $P\in\real^{k\times k}$, we use $\diag(P)$ to denote the vector in $\real^k$ that contains diagonal elements of $P$.  

Given a differentiable matrix function $F:\real^n \to \real^{k\times l}$, fix $x\in \real^n$ and $h\in\real^n$, we use $\deri F(x)[h]$ to denote the derivative of $F$ at $x$ in the direction $h$:
\begin{equation*}
    \deri F(x)[h]\defn \lim_{t\to 0}\frac{F(x+th)-F(x)}{t}.
\end{equation*}
Given a point $x\in K = \braces{x \mid A x > b}$, we define the slackness at $x$ to be $s_x\defn Ax-b$, and $S_x\defn \Diag(s_x)$. We also define $A_x\in\real^{m\times n}$ to be $A_x\defn S_x^{-1}A$. We  use $\ball(x,r)$ as a shorthand for a ball centered at $x\in\real^n$ with radius $r>0$. In other words,
\begin{equation*}
    \ball(x,r)\defn \braces{y\in\real^n\bigg|(y-x)\tp (y-x)\leq r^2}.
\end{equation*}
Without ambiguity, $\ball_r$ denotes the ball centered at $0$ in $\real^n$ with radius $r$. For $\mu\in\real^n$ and $\Sigma \in \real^{n \times n}$ positive definite (PD) matrix, we use the notation $\Normal(\mu,\Sigma)$ to denote the multivariate Gaussian distribution in $\real^n$ with mean $\mu$ and covariance $\Sigma$. We also use $\Normal(z;\mu,\Sigma)$ to denote the probability density of $\Normal(\mu,\Sigma)$ computed at $z$. In other words,
\begin{equation*}
    \Normal(z;\mu,\Sigma)\defn \frac{1}{(2\pi)^{n/2}\sqrt{\det\parenth{\Sigma}}}\exp\brackets{-\frac12(z-\mu)\tp\Sigma^{-1}(z-\mu)}
\end{equation*}

\subsection{Regularized Dikin Walks}\label{sec_algorithm}

We introduce {\regdw}s, as our Markov chain of interest, which is an extension of the \softdw proposed in \cite{NEURIPS2023_mangoubi}. Algorithm~\ref{algo_main} shows its pseudocode. After formally introducing the transition kernels before and after the Metropolis step in the \regdw respectively, we discuss two choices local metrics in the setting where the truncation happens on a polytope.

\begin{algorithm}
\SetKwData{Left}{left}\SetKwData{This}{this}\SetKwData{Up}{up}
\SetKwFunction{Union}{Union}\SetKwFunction{FindCompress}{FindCompress}
\SetKwInOut{Input}{input}\SetKwInOut{Output}{output}
\SetKwComment{Comment}{ //}{ }
\Input{local metric $H$, step-size $r$, regularization size $\lambda>0$,\\
initial distribution $\mu_0$, number of iterations $T$} 
\Output{$x_T$}
draw $x_0\sim \mu_0$\;
\For{$t\leftarrow 1$ \KwTo $T$}{
    $x_t\leftarrow x_{t-1}$\;
    \text{draw }$U\sim \mathrm{Unif}(0,1)$ \Comment{lazification step} 
    \If{$U<\frac12$}{
        \text{draw }$z\sim\Normal(x_t,\frac{r^2}{n}\parenth{H(x_t)+\lambda I}^{-1})$\;
        $x_t\leftarrow z$ \text{ with probability }$\min\braces{1,\frac{\mathbf{1}_K(z)\exp(-f(z))}{\exp(-f(x_t))}\cdot 
        \frac{\Normal\parenth{x_t;z,\frac{r^2}{n} \parenth{H(z)+\lambda I}^{-1}}}{\Normal\parenth{z;x_t,\frac{r^2}{n} \parenth{H(x_t)+\lambda I}^{-1}}}}$\;
    }
}
\caption{\regdw for Logconcave Distributions Truncated on $K$}
\label{algo_main}
\end{algorithm}

We define a local metric as a matrix function $H:K\to \mathbb{S}_{+}^{n}$, which maps any $x\in K$ to a positive semi-definite (PSD) matrix $H(x)$. For convenience, we use $G$ to denote the regularized version of $H$ for some fixed regularization size $\lambda$,  $G(x)\defn H(x)+\lambda I$. Moreover, we let $E(x,G(x),r)$ denote the following ellipsoid:
\begin{equation}
E(x,G(x),r)\defn \{z| (z-x)\tp G(x)(z-x)\leq r^2\}.
\end{equation}

Instead of drawing proposals from a fixed Gaussian distribution as in Random Walk Metropolis, we draw proposals from the Gaussian distribution $\Normal\parenth{x,\frac{r^2}{n}G(x)^{-1}}$ in \dikinwalk. Thus in the Metropolis-Hasting step in Algorithm \ref{algo_main}, we need to multiply the ratio of proposal densities to ensure the stationary distribution of our Markov chain is the target distribution specified in Eq.~\eqref{eq_distri}.

In order to upper bound the mixing times of \regdw, we formally introduce the transition kernels before and after the Metropolis step respectively. We use $\mathcal{P}_x$ to denote the Gaussian distribution $\Normal(x,\frac{r^2}{n}G(x)^{-1})$ centered at $x$ with covariance inverse proportional to $G$, and $\TT_x$ to denote the probability distribution after one step of Metropolis filter. Specifically, for any Borel measurable set $B\subseteq \real^n$, 

\begin{equation}\label{eq_Transition_kernel_def}
    \TT_x(B)= \brackets{1-\Exs_{z\sim \cP_x}\parenth{\alpha(x,z)} }\delta_{x}(B)+\int_{z\in B}{\alpha(x,z)}\Normal\parenth{z;x,\frac{r^2}{n}G(x)^{-1}}dz,
\end{equation}
where $\delta_x$ is the Dirac measure at $x$, and $\alpha(x,z)$ is the acceptance rate defined in Algorithm \ref{algo_main}
\begin{equation}\label{eq_accept_rate_def}
\alpha(x,z)\defn\min\braces{1,\frac{\mathbf{1}_K(z)\exp(-f(z))}{\exp(-f(x_t))}\cdot 
        \frac{\Normal\parenth{x_t;z,\frac{r^2}{n} G(z)^{-1}}}{\Normal\parenth{z;x_t,\frac{r^2}{n} G(x_t)^{-1}}}}.
\end{equation}

It is worth noting that the actual transition kernel for each step in Algorithm \ref{algo_main} is the lazification of $\TT_x$, and we denote the lazification of $\cT$ to be $\barT$. In other words, for any Borel measurable set $B\subseteq \real^n$, we define:
\begin{equation*}
\widebar \TT_x(B)\defn\frac12 \delta_x(B)+\frac12\TT_x(B),
\end{equation*}

For sampling logconcave distributions truncated on polytopes specifically, we introduce two tailored local metrics for \regdw. The first is the \softdw introduced in \cite{NEURIPS2023_mangoubi}, and the second is the \reglewis designed by us, which is directly motivated by the properties of Lewis weights, see \cite{lee2019solving,laddha2020strong,pmlr-v247-kook24b}. 

\begin{definition}[\soft \cite{NEURIPS2023_mangoubi}]\label{def_logarithmic}
Given a polytope $K\defn \braces{x|Ax>b}$, for $x\in K$,  the soft-threshold metric $G(x)$ is defined as
\begin{equation*}
    G(x)\defn  H(x)+\lambda I_n, \text{ with }H(x)\defn A_x\tp A_x.
\end{equation*}
\end{definition}
\begin{definition}[\reglewis]\label{def_Lewis}
Given a polytope $K$, for $x\in K$, we define the \reglewis $G(x)$ to be
\begin{equation*}
G(x)\defn H(x)+\lambda I_n, \text{ with } H(x)\defn c_1\sqrt{n} (\log m)^{c_2} A_x\tp W_x A_x,
\end{equation*}    
where $c_1, c_2$ are some absolute constant, $W_x\defn\Diag(w_x)$ is the Lewis weights of the matrix $A_x$, which is defined by the following optimization problem:
\begin{equation}\label{eq_Lewis_weights_def}
w_x\defn \underset{w\in\real^m_{+}}{\arg\max}\brackets{\log\det\parenth{A_x\tp W^{c_q}A_x}-c_q\sum_{i=1}^m w_i},
\end{equation}
where $c_q$ is a shorthand for $c_q\defn 1-\frac2q$, and we choose $q=\Theta(\log m)$.
\end{definition}

\subsection{Main Contributions}
We investigate \softdw and our extension to more general \regdw{s} for sampling from a logconcave and log-smooth distribution truncated on a convex set.  Our contributions are three-fold. 

First, we improve the analysis on the \softdw proposed in \cite{NEURIPS2023_mangoubi}, upper bounding the mixing times for a wider range of shape constraints under the strongly logconcave and strongly log-smooth regime. For polytope constraints specifically, we show it mixes in $\Ot\parenth{(m+\kappa)n}$ iterations from a warm initialization. One improvement to~\cite{NEURIPS2023_mangoubi} is that we no longer require the polytope to be bounded. For truncated Gaussian sampling, where an affine transformation ensures $\kappa=1$, we reduce the mixing time upper bound from $\Ot(n^2)$ in \cite{pmlr-v247-kook24b} to $\Ot(n)$ when $m=o(n)$, matching the state-of-the-art mixing time analysis in unconstrained logconcave sampling using Random Walk Metropolis. To reduce the dependency on $m$ when $m\gg n$, we  use \regdw with \reglewis and show the mixing time is $\Ot(n^{2.5}+\kappa n)$.  

Second, we extend our mixing time analysis on the \regdw to the weakly logconcave distributions, where we only require the target distribution $\Pi$ to have a finite covariance matrix $\Sigma_{\pi}$. For polytope constraints specifically, we show that \softdw mixes in $\Ot\parenth{mn+\spectr\beta n }$, where $\spectr$ is the largest eigenvalue of the covariance matrix $\Sigma_{\pi}$. Using \regdw with \reglewis, we achieve a mixing time of $\Ot(n^{2.5}+\spectr\beta n)$ when $m > n$. 

Meanwhile, our regularized Dikin walk runs with only zeroth-order information of the negative log-density $f$. This is because we add a fixed regularization term to the local metric, which does not depend on first- or second-order information of $f$. In contrast, \cite{pmlr-v247-kook24b} relies on constructing a barrier function depending on $f$ and add the corresponding Hessian to the local metric, which restricts the target distributions that their algorithm can be applied~to. 

Third, for truncated distributions on a polytope, going beyond worst-case analysis, we demonstrate that \softdw achieves a mixing time dependent on the number of constraints that a high-probability ball intersects, rather than the total number of constraints $m$, improving over the above $O(mn)$ upper bound. 

\subsection{Paper Organization}

The rest of the paper is structured as follows. Section \ref{sec_related_works} reviews the literature on Markov chains for sampling distributions truncated on a convex set. In Section \ref{sec_main_results}, we summarize our main results, stated as theorems and corollaries, as well as main lemmas and the proof ideas. In Section \ref{sec_new_iso}, we prove the a new isoperimetry under a combination of Euclidean and Hilbert metric, which is a core innovation behind our mixing time results. In Section \ref{sec_conductance_proofs}, we give the specific proofs of our results based on conductance/$s$-conductance. Finally, in Section \ref{sec_disc}, we explore potential directions to extend our work.

\section{Related Works}\label{sec_related_works}
The problem of sampling from distributions truncated on convex sets has been widely studied. The simplest case of uniform sampling on a convex body has seen significant progress in algorithm design and mixing time analysis in the last few decades. In this section, we first discuss this progress on uniform sampling in Section~\ref{sec_unif_sampling}. Then we dive into extensions of these algorithms to non-uniform sampling in Section~\ref{sec_nonunif_sampling}.

\subsection{Uniform Sampling over Convex Bodies}\label{sec_unif_sampling}


Markov chain Monte Carlo (MCMC) algorithms are by far the most popular sampling algorithms for uniform sampling on convex bodies. We first introduce the so-called ``general-purpose'' samplers, using the terminology from \cite{pmlr-v247-kook24b}, which are designed for sampling from a general convex body $K\subseteq \real^n$. It assumes having access to the convex body only through a membership oracle which checks whether $x \in K$. Thus, the convergence of general-purpose samplers is measured by how many times the oracle is called to achieve certain error tolerance. Two famous examples of general-purpose samplers are \ballwalk and \hitandrun. 

For uniform sampling over a convex body $K\subseteq \real^n$, assume that the convex body is in the isotropic position, i.e., $\Exs_K(x)=0$ and $\Exs_K(xx\tp )=I$, then \cite{Kannan1997RandomWA} shows that the mixing time of \ballwalk  to be $\Ot(n^2/\psi_n^2)$ under warmness, where $\psi_n$ is the smallest isoperimetric constant for any isotropic logconcave distributions over $\real^n$, also known as Kannan-Lovasz-Simonovits (KLS) constant. Hit-and-run has also been studied for decades (see \cite{lovasz_hit-and-run_1999},\cite{Lov06hrcorner},\cite{chen2022hitandrun}). For \hitandrun on isotropic $K$, \cite{chen2022hitandrun} proved a mixing time of $\Ot(n^2/\psi_n^2)$ under warmness. Using the current lower bound for the KLS constant $\psi_n\gtrsim \log^{-\frac12}(n)$ (see \cite{chen_almost_2021}, \cite{klartag_bourgains_2022}, \cite{jambulapati_slightly_2022}, \cite{klartag_logarithmic_2023}), the mixing time for both \ballwalk and \hitandrun is $\Ot(n^2)$ for isotropic convex bodies.

One potential drawback for general-purpose samplers like \ballwalk or \hitandrun is that, when the convex body is not isotropic, an affine transformation needs to be computed to bring the convex body to near-isotropic positions. This preprocessing step is called rounding. For a convex body $K\subseteq \real^n$, the state-of-art rounding algorithm from \cite{jia_reducing_2021} has an oracle complexity of $\Ot(n^3)$.  So for a general non-isotropic convex body, the oracle complexity to get the first approximate sample is $\Ot(n^3)$. 

General-purpose samplers like \ballwalk and \hitandrun only use the membership oracle to access $K$. In practice, we often have richer information of the target distribution. For example, when sampling uniformly from a polytope $K=\braces{x|Ax-b>0}$. To exploit the structures of the linear constraints, a class of structured samplers called \dikinwalk was proposed. Motivated by interior point methods in convex optimization, in \dikinwalk, a local ellipsoid (often defined according to the Hessian of a barrier function) is used as the proposal distribution. The local ellipsoid can automatically adjust its radius based on its closeness to the boundary of $K$, making it possible to circumvent the rounding procedure for non-isotropic target distributions. 

Using \dikinwalk to sample uniformly from a convex body, \cite{kannan_random_2012} proved a mixing time of $\Ot(mn)$ using logarithmic barrier. In the regime where $m\gg n$, the linear dependency of mixing times upon $m$ is often undesirable. A solution to this problem is to put weights on different constraints, thus tuning the mixing time dependency upon $m$ and $n$: \cite{chen2018fast} proved mixing times of $\Ot(m^{0.5}n^{1.5})$ and $\Ot(n^{2.5})$ using Vaidya weights and approximate John weights respectively. 


To better exploit the geometry induced by the polytope structure, instead of \dikinwalk where each proposal is based on an ellipsoid in the Euclidean space, \cite{lee_geodesic_2017} designed \geodesicwalk, where a Riemannian metric is defined by the Hessian of the logarithmic barrier, and the proposal is determined by solving a differential equation to follow the geodesic on the manifold. This allows the sampler to take larger step while maintaining a low rejection rate. \cite{lee_geodesic_2017} achieved a mixing time of $\Ot\parenth{mn^{3/4}}$, which improved the $\Ot(mn)$ mixing time of plain \dikinwalk in \cite{kannan_random_2012}. Later, \cite{lee_RHMC_2018} further improved the mixing time to $\Ot(mn^{2/3})$ by using \RHMC (Riemannian Hamiltonian Monte Carlo), where a larger step can be taken since the Metropolis filter is eliminated while the stationary distribution is maintained throughout following the differential equations. To deal with the case $m\gg n$, \cite{gatmiry_sampling_2024} recently redesigned \RHMC where the Riemannian metric is based on a hybrid barrier  (of Lewis-weights barrier and logarithmic barrier), and the mixing time is proved to be $\Ot(m^{1/3}n^{4/3})$ for uniform sampling.

\begin{table}
\centering
\begin{tabular}{|c|c|c|}
\hline
Weights in Dikin Walks &  Assumptions on $f$$^\S$       & Mixing Time$^{\#}$  \\
 \hline
 logarithmic \cite{kannan_random_2012}  &  $f\equiv 0$ & $mn$\\ 
 Vaidya \cite{chen2018fast} & $f\equiv 0$ & $m^{1/2}n^{3/2}$\\
 logarithmic+$\ell_2$-reg \cite{NEURIPS2023_mangoubi} & $\beta$-smooth $f$  & $mn+n\beta R^{2} \,^\dag$ \\
 logarithmic+Gaussian \cite{pmlr-v247-kook24b} & quadratic $f$ & $(m+n)n$\\
 Vaidya+Gaussian\cite{pmlr-v247-kook24b} & quadratic $f$ & $m^{1/2}n^{3/2}$\\
 Lewis+Gaussian \cite{pmlr-v247-kook24b} & quadratic $f$ & $n^{2.5}$\\
 \textcolor{red}{logarithmic+$\ell_2$-reg (this paper)}& $\alpha$-convex \& $\beta$-smooth $f$ & $\textcolor{red}{(m+\kappa)n}$ \\
 \textcolor{red}{ logarithmic+$\ell_2$-reg (this paper)}& quadratic $f$ & $\textcolor{red}{mn} \,^\ddag$ \\ 
\textcolor{red}{logarithmic+$\ell_2$-reg (this paper)}& convex \& $\beta$-smooth $f$ & $\textcolor{red}{(m+\spectr\beta)n}^\star$ \\
 \textcolor{red}{Lewis+$\ell_2$-reg (this paper)}& $\alpha$-convex \& $\beta$-smooth $f$ & $\textcolor{red}{(n^{3/2}+\kappa)n}$ \\
\textcolor{red}{Lewis+$\ell_2$-reg (this paper)}& convex \& $\beta$-smooth $f$ & $\textcolor{red}{(n^{3/2}+\spectr\beta)n}$ \\
 \hline
\end{tabular}
\caption{ Mixing time upper bounds of Dikin walks with different weight choices in their local metrics for sampling a logconcave distribution truncated on a polytope from a warm start. $^\S f$ is the negative log-density of  the target distribution $\pi(x)\propto e^{-f(x)}\mathbf{1}_K(x)dx$. If $f$ is both $\alpha$-convex and $\beta$-smooth, we define $\kappa\defn \beta/\alpha$ to be its condition number.
$^\#$Logarithmic factors are omitted  in the mixing time upper bounds.
$^\dag$\cite{NEURIPS2023_mangoubi} assumed the polytope $K$ is bounded in a ball of radius $R$.
$^\ddag$For the special case where $f$ is quadratic, one can always do an affine transformation so that $\kappa=1$.
$^\star$$\spectr$ denotes the spectral norm of the covariance matrix of the target distribution}
\label{tab_related_works}
\end{table}

\subsection{Non-Uniform Sampling over Convex Sets}\label{sec_nonunif_sampling}

Besides uniform sampling on convex bodies, non-uniform sampling truncated on convex sets has also attracted a lot of attention as we discussed in Section \ref{sec_intro}. Moreover, volume computation of an arbitrary convex body has been an important topic in computer science community. \cite{cousins_gaussian_2018} showed that sampling from a Gaussian distribution truncated on a convex body is an important subroutine for Gaussian cooling procedure in volume computation.

General-purpose samplers like \ballwalk and \hitandrun have been extended to sampling from  general logconcave distributions. For well-rounded\footnote{The mixing time is $\Ot(n^2\frac{R^2}{r^2})$ where $\Exs_\pi(\vecnorm{x-z_f}2^2)\leq R^2$ and the level set of $\Pi$ with measure $1/8$ contains a ball of radius $r$. The well-rounded logconcave distribution is defined when $\frac{R}{r}\approx \sqrt{n}$.} logconcave distributions, both \ballwalk \cite{lovasz2007geometry} and \hitandrun \cite{lovasz_fast_2006, Lov06hrcorner} mixes in time $\Ot(n^3)$ given a warm start. 

Similar to the issue of non-isotropic target distributions in uniform sampling, for a general logconcave distribution, one needs to compute an affine transformation to bring the distribution to near-isotropic position.  The state-of-art rounding algorithm in \cite{lovasz_fast_2006} has a cost of $\Ot(n^4)$ assuming the knowledge of maximum of the function $f$.

The \dikinwalk has been extended to non-uniform distributions truncated on polytopes as well. \cite{narayanan2017efficient} defined the Dikin ellipsoid by rescaling the logarithmic metric by some constant according to the non-uniform distribution. However, their result applied to the uniform case implies a mixing time of $\Ot(m^2n^3)$, not matching the $\Ot(mn)$ bound proved in \cite{kannan_random_2012}. Later \cite{NEURIPS2023_mangoubi} proposed a soft-threshold version of \dikinwalk by adding a Euclidean metric (as a regularizer) to the local metric, so that the local ellipsoid is automatically shaped by both the polytope $K$ and the non-uniform distribution $e^{-f}$, and the mixing time is $\Ot\parenth{(mn+n\beta R^2)}$ for a polytope bounded in a ball of radius $R$ and $f$ to be $\beta$-smooth. However, the polytopes are often unbounded for some applications, e.g., the truncated Gaussian sampling problem in application in SUN distributions as in Eq.~\eqref{eq_motivation_truncated}. Later \cite{pmlr-v247-kook24b} provided a general framework for combining different penalties induced by various constraints for designing the local metric in \dikinwalk. By lifting up the state-space, adding the local metric induced by the polytope and the Hessian of a Gaussian barrier function, \cite{pmlr-v247-kook24b} proved a mixing time of $\Ot((m+n)n)$.

Apart from imposing an extra barrier function induced by the constraints, other approaches to deal with the truncation structures are brought up. For instance, reflected Hamiltonian Monte Carlo where the trajectory is reflected after the sampler hits the constraints (see \cite{pakman_exact_2014}, \cite{afshar_reflection_2015} for algorithmic designs and numerical experiments), and \cite{chalkis_truncated_2023} proved a mixing time of $\Ot(\kappa n^2l^2)$ under a warm start for a bounded polytope, where $\kappa$ is the condition number of the negative logdensity and $l$ is an upper bound on the number of reflections. Another approach to deal with the non-smoothness introduced by the truncation is to use proximal-based proposal distributions in the Metropolis-Hasting algorithms, where \cite{lee_structured_2021} proved a mixing time of $\Ot(\kappa n)$. However, each step of proximal-sampling requires calling a proximal sampling oracle, an efficient subroutine to draw from the proximal distributions. The designs of proximal sampling oracles are often highly problem-specific (see \cite{mou_efficient_2022} for examples of designs). Moreover, Langevin Monte Carlo has been extended to truncated distributions via projection operations  \cite{bubeck_sampling_2018,brosse_sampling_2017}. For example, \cite{bubeck_sampling_2018} proved a mixing time of $\Ot\parenth{\frac{R^6\max(n,R\beta)^{12}}{\epsilon^{12}}}$ for convex body $K$ bounded in a ball of radius $R$ and $\beta$-smooth negative logdensity $f$. 


To facilitate algorithm comparison, we summarize related existing mixing time analysis on \dikinwalk with various weight choices in Table \ref{tab_related_works}.

\section{Main Results}\label{sec_main_results}
We are ready to present the main results of this paper. Section \ref{sec_preliminary} introduces some preliminary properties that are needed to state our results. Section \ref{sec_results_strongly} establishes mixing time upper bounds for regularized Dikin walks applied to $\alpha$-strongly logconcave and $\beta$-log-smooth distributions truncated on a convex set $K$, as detailed in Theorem~\ref{th_main} and Corollaries \ref{cor_logarithmic}, \ref{cor_m<n}, and \ref{cor_Lewis} with quantitative rates when $K$ is a polytope. Section \ref{sec_results_weakly} extends these results to weakly logconcave distributions, as summarized in Theorem \ref{th_weakly} and Corollaries \ref{cor_logarithmic_weakly} and \ref{cor_Lewis_weakly}. In Section~\ref{sec_beyond_worst_result}, when $K$ is a polytope, we improve the above worst-case analysis, showing that the mixing time upper bound of \softdw is determined by the number of constraints that intersect the high-probability mass of the distribution, rather than by the total number of constraints. Section \ref{sec_practical_matters} provides an analysis of per-step complexity and discusses strategies for obtaining warm starts. Section~\ref{sec_main_lemmas} concludes this section by outlining the main lemmas and proof ideas.

\subsection{Properties of local metrics}\label{sec_preliminary}
We define a few properties that a local metric $G$ can satisfy. Since we do not restrict the convex set upon which the target distribution is truncated to be bounded/compact, we first need to extend the cross-ratio distances defined on convex bodies to unbounded convex sets. 

\begin{definition}\label{def_cross_ratio_unbounded}
Assume $K\subseteq \real^n$ is an open and convex set. Fix $x,y\in K$, we define the cross-ratio distance $d_K(x,y)$ under three scenarios:
\begin{itemize}
    \item If the line $\widebar{xy}$ intersects $\partial K$ (the boundary of $K$) with $p,q$, and in the order $p,x,y,q$, then
    \begin{equation*}
    d_K(x,y)\defn \frac{\vecnorm{p-q}2\vecnorm{x-y}2}{\vecnorm{p-x}2\vecnorm{q-y}2}.
    \end{equation*}
    \item If one of $\{p,q\}$ is at infinity, then we delete the two distances involving that point. In other words, if $p=\infty$, then $d_K(x,y)\defn \frac{\vecnorm{x-y}2}{\vecnorm{q-y}2}$; If $q=\infty$, then $d_K(x,y)\defn \frac{\vecnorm{x-y}2}{\vecnorm{p-x}2}$.
    \item If both $\{p,q\}$ are at infinity, then $d_K(x,y)\defn 0$.
\end{itemize}
\end{definition}
It is worth noting that the above three scenarios can be summarized in an equivalent definition: For (possibly unbounded) convex set $K$, fix $x,y\in K$, the cross-ratio distance is defined by $d_K(x,y)=\underset{R\to\infty}{\lim}d_{K\cap \ball(0,R)}(x,y)$. It is worth noting that for all $x\in \partial K$, $y\in K$ and $y\neq x$, we have $d_K(x,y)=+\infty$. We just avoid this technical issue by requiring $K$ to be an open subset of $\real^n$ in this paper.  Our mixing time result can be extended to non-open $K$ easily, because the Lebesgue measure of the boundary $\partial K$ of a convex set $K$ is $0$. 

In theorem statements and proofs, we introduce a few properties of a local metric $G$, following~\cite{pmlr-v247-kook24b}.
\begin{definition}[\cite{laddha2020strong, pmlr-v247-kook24b}]\label{def_sc}
A matrix function $G:K\to\mathbb{S}_+^n$ is said to satisfy
\begin{itemize}
\item Strong self-concordance (SSC) if $G$ is positive-definite on $K$ and 
\begin{equation*}
    \vecnorm{G(x)^{-\frac12}\deri G(x)[h]G(x)^{-\frac12}}F\leq 2\vecnorm{h}{G(x)}, \quad \forall x \in K, h\in\real^n.
\end{equation*}
\item Lower trace self-concordance(LTSC) if $G$ is positive-definite on $K$ and 
\begin{align*}
    \trace\braces{G(x)^{-1}\deri^2 G(x)[h,h]}\geq -\vecnorm{h}{G(x)}^2, \quad \forall x \in K, h\in\real^n.
\end{align*}
We say it satisfies strongly lower trace self-concordant (SLTSC) if for any PSD matrix function $\barG$ on $K$ it holds that
\begin{equation*}
\trace\braces{\parenth{\barG(x)+G(x)}^{-1}\deri^2 G(x)[h,h]}\geq -\vecnorm{h}{G(x)}^2, \quad \forall x \in K, h\in\real^n.
\end{equation*}
\item Average self-concordance (ASC) if for any $\epsilon>0$ there exists $r_\epsilon>0$ such that for $r\leq r_{\epsilon}$,
\begin{equation*}
    \prob_{z\sim \Normal(x,\frac{r^2}{d}G(x)^{-1})}\parenth{\abss{\vecnorm{z-x}{G(z)}^2-\vecnorm{z-x}{G(x)}^2}\leq \frac{2\epsilon r^2}{n}}\geq 1-\epsilon.
\end{equation*}
We say it satisfies strongly average self-concordant (SASC) if for any $\epsilon>0$ and any PSD matrix function $\barG$ on $K$ it holds that 
\begin{equation*}
 \prob_{z\sim\Normal(x,\frac{r^2}{d}[G(x)+\barG(x)]^{-1})}\parenth{\abss{\vecnorm{z-x}{G(z)}^2-\vecnorm{z-x}{G(x)}^2}\leq \frac{2\epsilon r^2}{n}}\geq 1-\epsilon.
\end{equation*}
\end{itemize}
\end{definition}

Note that we make the ASC definition stricter than it was in \cite{pmlr-v247-kook24b}, by requiring the difference in local norms to be bounded on both sides. This modification does not have a major impact on the results, because the verification of ASC is based on Taylor approximation as in \cite{pmlr-v247-kook24b}, which automatically ensures that both sides are controlled.

\subsection{Sampling Strongly Logconcave Distributions}\label{sec_results_strongly}
Theorem \ref{th_main} establishes sufficient conditions for a regularized Dikin walk with a local metric to mix fast on $\alpha$-strongly logconcave and $\beta$-log-smooth distributions. Corollaries \ref{cor_logarithmic}, \ref{cor_m<n}, and \ref{cor_Lewis} apply Theorem \ref{th_main} to obtain quantitative mixing time upper bounds, when $K$ is a polytope and $G$ is soft-threshold logarithmic metric in Definition \ref{def_logarithmic} or soft-threshold Lewis metric in Definition \ref{def_Lewis}.

\begin{theorem}\label{th_main}
Let $\Pi$ be a target distribution with density  $\pi(x)\propto \mathbf{1}_K(x)e^{-f(x)}$, where $K\subseteq\real^n$ is an open and convex set, $\alpha I_n\preceq \nabla^2 f\preceq \beta I_n$ as in Eq.~\eqref{eq_distri}. If in Algorithm \ref{algo_main} we provide the local metric $H$ and regularization size $\lambda>0$ such that $G\defn H+\lambda I$ satisfies
\begin{itemize}
    \item SSC, LTSC and ASC in Definition~\ref{def_sc},
    \item $G(x)\succeq \beta I$  and the ellipsoid $E(x,G(x),1)\subseteq K$ for all $x\in K$,
    \item there exists $\crossbd\geq 1$ and $\euclidbd\geq \beta$, such that 
    \begin{align*}
        \min\braces{\vecnorm{y-x}{G(x)}^2,\vecnorm{y-x}{G(y)}^2}\leq \crossbd\cdot d_K(x,y)^2+\euclidbd\cdot\vecnorm{y-x}{2}^2, \forall x, y \in K,
    \end{align*}
\end{itemize}
then there exists a step size $r>0$ and a universal constant $C>0$ such that for any error tolerance $\epsilon>0$ and any $M$-warm initial distribution $\mu_0$, as long as
\begin{equation*}
T\geq C\parenth{\crossbd+\frac{\euclidbd}{\alpha}}n\log\parenth{\frac{\sqrt{M}}{\epsilon}},
\end{equation*}
the output distribution satisfies $\vecnorm{\mu_T-\Pi}{TV}\leq \epsilon$. 
\end{theorem}

Setting $K$ to be a polytope and $G$ to be the soft-threshold logarithmic metric in Definition~\ref{def_logarithmic} immediately leads to the following corollary. 
\begin{corollary}[\softlog, strongly logconcave target]\label{cor_logarithmic}
Under the same assumptions on the target distribution $\Pi$ in Theorem \ref{th_main}, we further assume that $K = \braces{x|Ax>b}$ is a (possibly unbounded) polytope for $A\in\real^{m\times n}$ and $b\in\real^m$. If in Algorithm \ref{algo_main} we provide the local metric $H$ and regularization size $\lambda=\beta$ in Definition \ref{def_logarithmic}, then there exists step-size $r>0$ and a universal constant $C>0$ such that for any error tolerance $\epsilon>0$ and any $M$-warm initial distribution $\mu_0$, as long as
\begin{equation*}
T\geq C\parenth{m+\kappa}n\log\parenth{\frac{\sqrt{M}}{\epsilon}}, 
\end{equation*}
the output distribution satisfies $\vecnorm{\mu_T-\Pi}{TV}\leq \epsilon$. 
\end{corollary}

Corollary \ref{cor_logarithmic} shows that \softdw mixes in $\Ot\parenth{(m+\kappa)n}$ for truncated logconcave sampling, ignoring constants and logarithmic factors. Since the uniform distribution on a polytope can be approximated as the limit of Gaussian distributions truncated on the same polytope with increasing covariance and $\kappa=1$, we recover the $O(mn)$ mixing time from a warm start for the standard Dikin walk~\cite{kannan_random_2012} in this limit. Compared to previous work on sampling truncated log-concave distributions~\cite{NEURIPS2023_mangoubi} and \cite{pmlr-v247-kook24b}, Corollary \ref{cor_logarithmic} has a few improvements.  In~\cite{NEURIPS2023_mangoubi}, \softdw is introduced with a mixing time bound of $\Ot(mn+n\beta R^2)$, where $R$ is the radius of a ball containing the polytope. Using our new isoperimetric inequality, we establish a mixing time of $\Ot(mn+\kappa n)$, removing the dependence on $R$ which can be as large as $\sqrt{n}$ in general. 

Furthermore, the regularized Dikin walk with the \softlog only needs to evaluate $f$ when computing the acceptance rates, and no first or second order derivatives of $f$ are required. Thus, Algorithm \ref{algo_main} does not impose any other restrictions on the specific forms of $f$ other than it can be evaluated. This flexibility enables wide statistical applications. For example, $f$ has the form $f(\theta)=\sum_{i} l(\theta;x_i)$, in Bayesian Lasso logistic regression \cite{NEURIPS2023_mangoubi,tian_efficient_2008} and differentially private optimization \cite{mangoubi2022faster}, where $l$ denotes a loss function and $\{x_i\}$ denote data points.  In contrast, the approach to perform non-uniform sampling truncated on $K$ in \cite{pmlr-v247-kook24b} is to construct a barrier function induced by $f$ and to add the Hessian of the barrier function into the local metric. This approach requires to evaluate the Hessian of $f$ at each step, which may be costly in above applications. 

Additionally, for truncated Gaussian sampling specifically, \cite{pmlr-v247-kook24b} introduces a new barrier walk with a mixing time of $\Ot(mn+n^2)$ by combining barriers from the polytope and the Gaussian distribution. Here, we prove a mixing time of $\Ot(mn)$, which is smaller when $m < n$, after reducing the condition number $\kappa$ to $1$ through an affine transformation. Concretely, this bound allows us to improve the mixing time to be $O(n)$ when $m=o(n)$, as highlighted in Corollary \ref{cor_m<n}.

\begin{corollary}[\softlog, Gaussian target]\label{cor_m<n}
Under the same assumption for target distribution $\Pi$ as in Theorem \ref{th_main}. We assume $K\defn \braces{x|Ax>b}$ is a (possibly unbounded) polytope for some $A\in\real^{m\times n}$ and $b\in\real^m$. If in Algorithm \ref{algo_main} we provide the local metric $H$ and regularization size $\lambda=\beta$ in Definition \ref{def_logarithmic}. If we further assume $m=o(n)$, then there exists step-size $r>0$ and a universal constant $C>0$ such that for any error tolerance $\epsilon>0$ and any $M$-warm initial distribution $\mu_0$, as long as
\begin{equation*}
    T\geq C\kappa n\log\parenth{\frac{\sqrt{M}}{\epsilon}},
\end{equation*}
the output distribution satisfies $\vecnorm{\mu_T-\Pi}{TV}\leq \epsilon$.  
\end{corollary}

In uniform sampling of polytopes, it is typically assumed that $m \geq n$ because otherwise the polytope is unbounded, and the resulting distribution becomes improper. In truncated strongly logconcave sampling, however, the target distribution remains well-defined even when $m$ is as small as $0$. Therefore, in cases where $m \ll n$, we expect a better mixing time than $O(n^2)$. This small-$m$ scenario arises in many statistical  applications, such as Bayesian linear models with Gaussian priors under linear inequality constraints  \cite{geweke_bayesian_1996,ghosal_bayesian_2022}, where $m$ reflects prior information and can be arbitrarily small. 

When $m=0$, Corollary \ref{cor_m<n} provides a mixing time upper bound which matches that of Random Walk Metropolis (RWM) for unconstrained logconcave sampling \cite{dwivedi2019log, chen_fast_2020, andrieu_explicit_2024}. Specifically, with $m=0$, the polytope part $A_x\tp A_x$ in the \softlog $G$ in Definition \ref{def_logarithmic} vanishes, making \softdw exactly the same algorithm as RWM, so a similar mixing time bound is expected. Our result of $\Ot(\kappa n)$ indeed matches the state-of-the-art RWM bound from \cite{andrieu_explicit_2024}, improving the results of \cite{NEURIPS2023_mangoubi}.

Furthermore, Corollary \ref{cor_m<n} demonstrates that for truncated Gaussian sampling, the mixing time scales linearly in $n$ when $m= O(1)$. In comparison, the upper bound in \cite{pmlr-v247-kook24b} is $\Ot(n^2)$ under the same conditions. This extra $n$-factor arises because \cite{pmlr-v247-kook24b} introduces
$(\nu,\bar\nu)$-Dikin amenability as an extension of strong self-concordance and $\bar\nu$-symmetry~\cite{laddha2020strong}. However, ensuring $(\nu,\bar\nu)$-Dikin amenability requires them to multiply an additional $n$ scaling to the Gaussian barrier, which forces them to choose a smaller step size. In contrast, we add an $\ell_2$-regularization to our local metric with a dimension-independent scaling. While this adjustment may appear simple, circumventing the need for $(\nu,\bar\nu)$-Dikin amenability in~\cite{pmlr-v247-kook24b} requires a new proof technique, which we outline below. 

Our proof deviates from the existing ones in two key ways. First, when we control the acceptance rate, we extend the coupling argument in \cite{andrieu_explicit_2024} to accommodate asymmetric proposals, which improves the radius $R$ dependent control in~\cite{NEURIPS2023_mangoubi}. Second, we develop a novel isoperimetric inequality under a mixed metric that combines Euclidean and cross-ratio distances, which allows us to have a direct analysis of regularized Dikin walks. 

Theorem \ref{th_main} can be applied beyond logarithmic metrics. It is well-known in truncated uniform sampling that the standard Dikin walk can be improved using Lewis-weighted Dikin walk in the case $m \gg n$~\cite{chen2018fast,pmlr-v247-kook24b}. Following this idea, in the case $m \gg n$, we define the regularized Dikin walk using Lewis metric in Definition \ref{def_Lewis}. Using the SSC, LSC, ASC properties established in \cite{pmlr-v247-kook24b}, we obtain the following Corollary \ref{cor_Lewis}.

\begin{corollary}[\reglewis, strongly logconcave target]\label{cor_Lewis}
Under the same assumption for target distribution $\Pi$ as in Theorem \ref{th_main}. We further assume $K\defn \braces{x|Ax>b}$ is a  polytope for some $A\in\real^{m\times n}$ and $b\in\real^m$. If in Algorithm~\ref{algo_main} we provide the local metric $H$ and regularization size $\lambda=\beta$ as in Definition \ref{def_Lewis}, then there exists step-size $r>0$ and a universal constant $C_1>0,C_2>0$ such that for any error tolerance $\epsilon>0$ and any $M$-warm initial distribution $\mu_0$, as long as
\begin{equation*}
T\geq C_1(\log m)^{C_2}\parenth{n^{3/2}+\kappa}n\log\parenth{\frac{\sqrt{M}}{\epsilon}}, 
\end{equation*}
the output distribution satisfies $\vecnorm{\mu_T-\Pi}{TV}\leq \epsilon$. 
\end{corollary}

The proofs of Theorem \ref{th_main}, Corollary \ref{cor_logarithmic}, \ref{cor_m<n} and \ref{cor_Lewis} are postponed to Section \ref{sec_mixing_strongly}. 

\subsection{Extension to Weakly Logconcave Distributions}\label{sec_results_weakly}

We extend the mixing time upper bound (Theorem \ref{th_main}) to distributions that are not necessarily $\alpha$-strongly logconcave, and instead only has a finite covariance matrix. This extension is summarized as Theorem~\ref{th_weakly}. Corollaries \ref{cor_logarithmic_weakly} and \ref{cor_Lewis_weakly} are the applications of Theorem~\ref{th_weakly} to polytopes, and  $G$ being soft-threshold logarithmic metric as in Definition \ref{def_logarithmic} or soft-threshold Lewis metric as in Definition \ref{def_Lewis} respectively. 

Here we emphasize that the upper bounds for weakly logconcave distributions loses a factor of the Kannan-Lovász-Simonovits (KLS) constant $\psi_n$, which appears in a new isoperimetric inequality proved by us under a mixed metric for weakly logconave measures. For a distribution $\mu$ over $\real^n$, its isoperimetric constant is
\begin{equation}\label{eq_KLS_def}
    \frac{1}{\psi_\mu}\defn \inf_{A\subseteq\real^n}\braces{\frac{\mu^+(A)}{\min\braces{\mu(A),1-\mu(A)}}},
\end{equation}
where $A$ is taken over all Borel sets of $\real^n$, and $\mu^+$ is defined to be the boundary measure of $\mu$. Given any Borel set $B\subset \real^n$, the boundary measure $\Pi^+$ of $B$ is defined to be
\begin{equation}\label{eq_def_boundary_euclid}
\begin{aligned}
    \Pi^+(B)=\underset{h\to 0^+}{\lim\inf} \frac{\Pi(B^h)-\Pi(B)}{h}, 
    \text{ where } B^h\defn \braces{x\in X|\exists a\in B,\vecnorm{x-a}{2}<h}.
\end{aligned}
\end{equation}
The KLS constant $\psi_n$ is defined by taking supremum over all isotropic and logconcave measures in $\real^n$,
\begin{equation*}
    \psi_n=\underset{\mu}{\sup}\,\psi_\mu.
\end{equation*}

\begin{theorem}\label{th_weakly}
Let $\Pi$ be a target distribution with density  $\pi(x)\propto \mathbf{1}_K(x)e^{-f(x)}$, where $K\subseteq \real^n$ is an open and convex set, $0\preceq\nabla^2 f\preceq \beta I_n$ as in Eq.~\eqref{eq_distri_weakly}. Suppose the covariance matrix of $\Pi$ is bounded as $\Sigma_\pi \preceq \eta I_n$. If in Algorithm \ref{algo_main} we provide the local metric $H$ and regularization size $\lambda>0$ such that $G\defn H+\lambda I$ satisfies
\begin{itemize}
    \item $G$ is SSC, LTSC, ASC
    \item $G(x)\succeq \beta I$  and the ellipsoid $E(x,G(x),1)\subseteq K$ for all $x\in K$.
    \item There exists $\crossbd\geq 1$ and $\euclidbd\geq \beta$, such that 
    \begin{align*}
        \min\braces{\vecnorm{y-x}{G(x)}^2,\vecnorm{y-x}{G(y)}^2}\leq \crossbd\cdot d_K(x,y)^2+\euclidbd\cdot\vecnorm{y-x}{2}^2, \forall x, y \in K,
    \end{align*}
\end{itemize}
Then there exists a step-size $r>0$ and a universal constant $C>0$ such that for any $M$-warm initial distribution $\mu_0$ and any error tolerance $\epsilon>0$, as long as
\begin{equation*}
    T\geq C\psi_n^2(\crossbd + {\euclidbd}{\spectr})n\log\parenth{\frac{\sqrt{M}}{\epsilon}},
\end{equation*}
the output distribution satisfies $\vecnorm{\mu_T-\Pi}{TV}\leq \epsilon$.
\end{theorem}

The famous KLS conjecture~\cite{kannan1995isoperimetric} conjectures that $\psi_n$ is upper bounded by a universal constant for any dimension $n$. A series of works  \cite{lee_eldans_2019,chen_almost_2021,klartag_bourgains_2022,klartag_logarithmic_2023} successfully proved $\psi_n=O(\sqrt{\log n})$, thus we lose at most by a log-factor in $n$ in the mixing times in Theorem \ref{th_weakly}.

\begin{corollary}[\softlog, weakly logconcave target]\label{cor_logarithmic_weakly}
Under the same assumption for target distribution $\Pi$ as in Theorem \ref{th_weakly}. We further assume $K\defn \braces{x|Ax>b}$ is a (possibly unbounded) polytope for some $A\in\real^{m\times n}$ and $b\in\real^m$. If in Algorithm \ref{algo_main} we set  $G$ to be the \softlog in Definition~\ref{def_logarithmic} with regularization size $\lambda\defn\beta$, then there exists step-size $r>0$ and a universal constant $C>0$ such that for any error tolerance $\epsilon>0$ and any $M$-warm initial distribution $\mu_0$, as long as
\begin{equation*}
T\geq C\psi_n^2\parenth{m+\beta{\spectr}}n\log\parenth{\frac{\sqrt{M}}{\epsilon}}, 
\end{equation*}
the output distribution satisfies $\vecnorm{\mu_T-\Pi}{TV}\leq \epsilon$. 
\end{corollary}

\begin{corollary}[\reglewis, weakly logconcave target]\label{cor_Lewis_weakly}
Under the same assumption for target distribution $\Pi$ as in Theorem \ref{th_weakly}. We further assume $K\defn \braces{x|Ax>b}$ is a (possibly unbounded) polytope for some $A\in\real^{m\times n}$ and $b\in\real^m$.  If in Algorithm \ref{algo_main} we provide the local metric $H$ and regularization size $\lambda=\beta$ in Definition \ref{def_Lewis}, then there exists step-size $r>0$ and a universal constant $C_1>0,C_2>0$ such that for any error tolerance $\epsilon>0$ and any $M$-warm initial distribution $\mu_0$, as long as
\begin{equation*}
T\geq C_1(\log m)^{C_2}\psi_n^2\parenth{n^{3/2}+\beta\spectr}n\log\parenth{\frac{\sqrt{M}}{\epsilon}}, 
\end{equation*}
the output distribution satisfies $\vecnorm{\mu_T-\Pi}{TV}\leq \epsilon$. 
\end{corollary}

The proofs for Theorem \ref{th_weakly}, Corollary \ref{cor_logarithmic_weakly} and Corollary \ref{cor_Lewis_weakly} are postponed to Section \ref{sec_mixing_weakly}.

\subsection{Beyond Worst-Case Analysis}\label{sec_beyond_worst_result}
In Section \ref{sec_beyond_worst_result}, we focus on sampling from an $\alpha$-strongly logconcave and $\beta$-log-smooth distribution as in Eq.~\eqref{eq_distri}, where we also confine $K$ to be a polytope $K\defn \braces{x|Ax>b}$ and let the local metric in Algorithm \ref{algo_main} be the \softlog in Definition \ref{def_logarithmic}.

Intuitively, when a target logconcave distribution is well concentrated inside the polytope, even if $m$ is large, the \softdw should behave similar to RMW and should share the same mixing time of $\Ot(\kappa n)$ rather than $\Ot((m+\kappa)n)$. Based on this intuition, in this section, we seek sufficient conditions on the target distribution for the \softdw to achieve a mixing time smaller than that in Corollary~\ref{cor_logarithmic}. More precisely, we demonstrate  that the mixing time of the \softdw depends only on a subset of constraints, those intersect a high probability region of the target distribution, rather than all the constraints. To formalize this phenomenon, we first have to define what we mean by a high probability region. 

\begin{definition}\label{def_rinf}
Let $\Pi$ be a target distribution with density $\pi(x) \propto \mathbf{1}_K(x)e^{-f(x)}$, where $f$ is twice differentiable and $\alpha$-convex. For all $s\in (0,1)$ we define $\rinf$ as the smallest radius (scaled by $\sqrt{n/\alpha}$) such that the ball centered at the mode contains at least $1-s$ probability mass of $\Pi$:
\begin{equation*}
\rinf(s)\defn \inf\braces{R\geq 0\bigg|\Pi\parenth{B\parenth{x^\star,R\sqrt{\frac{n}{\alpha}}}}\geq 1-s},
\end{equation*}
where $x^\star \defn \underset{x\in K}{\arg\min} f(x) $ is the mode of $\Pi$. 
\end{definition}
In particular, for $\alpha$-convex $f$, $\rinf$ is always upper bounded by the simpler function $\rinfbd$, defined as follows (see Lemma 1 in~\cite{dwivedi2019log} for a proof). 
\begin{equation}\label{eq_rinf(s)_upper_bound}
\rinf(s)\leq \rinfbd(s):= 2+2\max\braces{\frac{1}{n^{0.25}}\log^{0.25}\parenth{\frac{1}{s}},\frac{1}{n^{0.5}}\log^{0.5}\parenth{\frac1s}}.
\end{equation}
For $R,\delta>0$, we use $\cB_R$ and $\cB_R^{\delta}$ to denote the following balls,
\begin{equation}\label{eq_B_R_delta_def}
\cB_R \defn \ball\parenth{x^\star,R\sqrt{\frac{n}{\alpha}}},\quad 
\cB_{R}^\delta\defn \ball\parenth{x^\star,(R+\delta)\sqrt{\frac{n}{\alpha}}}.
\end{equation}
Additionally, given the polytope in Eq.~\eqref{eq_distri}, we define $\cM_{R}^\delta$ to be the number of linear constraints violated inside the ball $\cB_{R}^\delta$ as follows,
\begin{equation*}
\cM_{R}^{\delta}\defn \mathrm{Card}\{i\in [m]\,|\,\exists\, x\in \cB_{R}^\delta \text{ s.t. }a_i\tp x-b_i\leq 0\}.
\end{equation*}
Now we are ready to state Theorem \ref{th_prob_ball_intersection} which establishes the mixing time of the \softdw as a function of how $\cB_R^{\delta}$ intersects the polytope. 

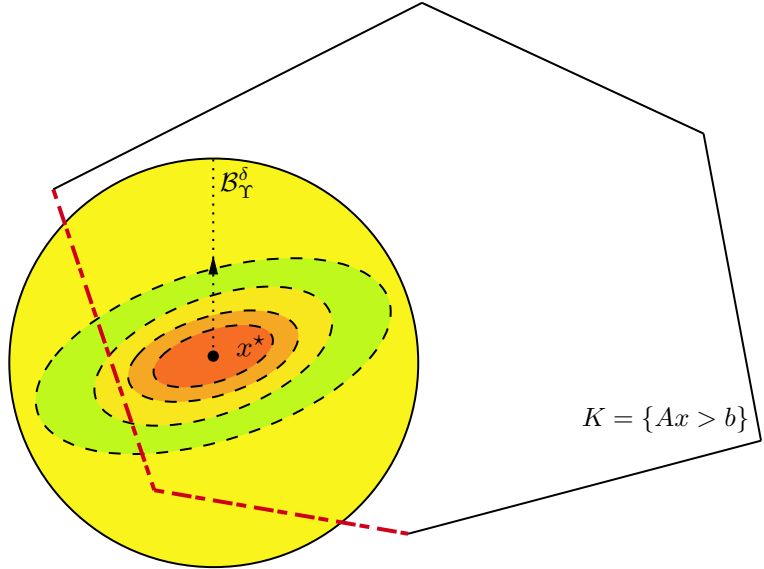
\begin{figure}
    \centering
    \tikzset{every picture/.style={line width=0.75pt}} 
\begin{tikzpicture}[x=0.75pt,y=0.75pt,yscale=-1,xscale=1]
\draw  [fill={rgb, 255:red, 248; green, 244; blue, 28 }  ,fill opacity=1 ] (77.83,193.85) .. controls (77.68,136.9) and (123.73,90.61) .. (180.68,90.46) .. controls (237.63,90.31) and (283.92,136.36) .. (284.07,193.31) .. controls (284.21,250.26) and (238.17,296.54) .. (181.22,296.69) .. controls (124.27,296.84) and (77.98,250.8) .. (77.83,193.85) -- cycle ;
\draw  [fill={rgb, 255:red, 190; green, 248; blue, 28 }  ,fill opacity=1 ][dash pattern={on 4.5pt off 4.5pt}] (100.83,190.2) .. controls (123.14,162.9) and (176.98,140.77) .. (221.08,140.77) .. controls (265.17,140.77) and (282.84,162.9) .. (260.53,190.2) .. controls (238.22,217.5) and (184.38,239.62) .. (140.28,239.62) .. controls (96.18,239.62) and (78.52,217.5) .. (100.83,190.2) -- cycle ;
\draw  [fill={rgb, 255:red, 248; green, 231; blue, 28 }  ,fill opacity=1 ][dash pattern={on 4.5pt off 4.5pt}] (128.05,190.2) .. controls (143.89,170.81) and (180.3,155.1) .. (209.37,155.1) .. controls (238.44,155.1) and (249.16,170.81) .. (233.31,190.2) .. controls (217.47,209.58) and (181.06,225.3) .. (151.99,225.3) .. controls (122.92,225.3) and (112.2,209.58) .. (128.05,190.2) -- cycle ;
\draw  [fill={rgb, 255:red, 245; green, 168; blue, 35 }  ,fill opacity=1 ][dash pattern={on 4.5pt off 4.5pt}] (142.11,190.2) .. controls (152.52,177.46) and (178.23,167.13) .. (199.53,167.13) .. controls (220.84,167.13) and (229.67,177.46) .. (219.25,190.2) .. controls (208.84,202.94) and (183.13,213.27) .. (161.83,213.27) .. controls (140.52,213.27) and (131.69,202.94) .. (142.11,190.2) -- cycle ;
\draw  [fill={rgb, 255:red, 245; green, 110; blue, 35 }  ,fill opacity=1 ][dash pattern={on 4.5pt off 4.5pt}] (154.64,190.2) .. controls (162.98,181.53) and (181.41,174.5) .. (195.79,174.5) .. controls (210.18,174.5) and (215.07,181.53) .. (206.72,190.2) .. controls (198.38,198.87) and (179.95,205.9) .. (165.57,205.9) .. controls (151.18,205.9) and (146.29,198.87) .. (154.64,190.2) -- cycle ;
\draw [color={rgb, 255:red, 208; green, 2; blue, 27 }  ,draw opacity=1 ][line width=1.5]  [dash pattern={on 3.75pt off 3pt on 7.5pt off 1.5pt}]  (100,106) -- (150.96,257.86) ;
\draw [fill={rgb, 255:red, 248; green, 231; blue, 28 }  ,fill opacity=1 ]   (100,106) -- (285.96,11.86) ;
\draw [fill={rgb, 255:red, 160; green, 51; blue, 51 }  ,fill opacity=1 ]   (285.96,11.86) -- (427.96,77.86) ;
\draw    (427.96,77.86) -- (456.96,232.86) ;
\draw [color={rgb, 255:red, 208; green, 2; blue, 27 }  ,draw opacity=1 ][line width=1.5]  [dash pattern={on 3.75pt off 3pt on 7.5pt off 1.5pt}]  (150.96,257.86) -- (257.04,276.09) -- (278.96,279.86) ;
\draw    (278.96,279.86) -- (456.96,232.86) ;
\draw  [dash pattern={on 0.84pt off 2.51pt}]  (180.68,190.2) -- (180.68,90.46) ;
\draw [shift={(180.68,140.33)}, rotate = 90] [fill={rgb, 255:red, 0; green, 0; blue, 0 }  ][line width=0.08]  [draw opacity=0] (8.4,-2.1) -- (0,0) -- (8.4,2.1) -- cycle    ;
\draw [shift={(180.68,190.2)}, rotate = 270] [color={rgb, 255:red, 0; green, 0; blue, 0 }  ][fill={rgb, 255:red, 0; green, 0; blue, 0 }  ][line width=0.75]      (0, 0) circle [x radius= 2.34, y radius= 2.34]   ;

\draw (184.57,175.71) node [anchor=north west][inner sep=0.75pt]  [rotate=-1.71] [align=left] {$\displaystyle  \begin{array}{{>{\displaystyle}l}}
x^{\star }\\
\end{array}$};
\draw (122,131) node [anchor=north west][inner sep=0.75pt]   [align=left] {$ $};
\draw (182.68,93.46) node [anchor=north west][inner sep=0.75pt]  [font=\small] [align=left] {$\displaystyle \mathcal{B}_{\Upsilon }^{\delta }$};
\draw (365,214) node [anchor=north west][inner sep=0.75pt]  [font=\small] [align=left] {$\displaystyle K=\{Ax >b\}$};
\end{tikzpicture}
\caption{An example of Theorem \ref{th_prob_ball_intersection} in $\real^2$ ($n=2$), where $K$ is the polytope and the ball $\cB_\Upsilon^\delta$ refers to the high-probability ball of the truncated distribution $\Pi$. The dashed ellipsoids are contours for the potential $f(x)$ of the distribution. The dashed segments are the constraints that are violated, the solid segments are untouched constraints, so $\cM_\Upsilon^\delta=2$, $m=6$.}
\label{fig_intersection_ball}
\end{figure}

\begin{theorem}\label{th_prob_ball_intersection}

Under the same assumptions on the target distribution as Theorem \ref{th_main}, if we further assume that $K\defn\braces{x|Ax>b}$ for some $A\in\real^{m\times n}$ and $b\in\real^m$. If in Algorithm \ref{algo_main} we provide the local metric $H$ and regularization size $\lambda=\beta$ in Definition \ref{def_logarithmic}, then there exists step-size $r>0$ and a universal constant $C>0$ such that for any error tolerance $\epsilon>0$ and any $M$-warm initial distribution $\mu_0$, As long as
\begin{equation}\label{eq_high_prob_mixing_time}
T\geq C \underset{\delta>0}{\inf}\brackets{\kappa n+\frac{m}{\delta^2}+n\cdot\cM^\delta_{\Upsilon}} \log\parenth{\frac{2M}{\epsilon}},
\end{equation}
where $\Upsilon\defn \rinf(\frac{\epsilon}{2M})$ denotes the radius function $\rinf$ valued at $\frac{\epsilon}{2M}$ (see Definition \ref{def_rinf}), the output distribution satisfies $\vecnorm{\mu_T-\Pi}{TV}\leq \epsilon$. 

To ease comparison, remark that the result also holds if we replace $\Upsilon$ in Eq. \eqref{eq_high_prob_mixing_time} with $\widehat \Upsilon\defn \rinfbd(\frac{\epsilon}{2M})$, defined in Eq.~\eqref{eq_rinf(s)_upper_bound}. 
\end{theorem}
Its proof is provided in Section~\ref{sec_prob_ball_intersection}.
In Theorem \ref{th_prob_ball_intersection}, note that $\delta$ can be tuned to obtain a more interpretable mixing time bound. If we take $\delta\to \infty$, then $m/\delta^2\to 0$ and $\cM_\Upsilon^\delta\to m$, we recover the $\Ot((m+\kappa) n)$ bound in Corollary \ref{cor_logarithmic}. 

Figure \ref{fig_intersection_ball} shows an example where applying Theorem~\ref{th_prob_ball_intersection} can be beneficial than directly applying Corollary~\ref{cor_logarithmic}. In general, for a distribution that is well-concentrated inside the polytope, where its high probability mass does not intersect many linear constraints of the polytope, choosing a suitable gap $\delta$ could make the $\cM_\Upsilon^\delta$ increase slower than $m$, as well as ensuring $m/\delta^2$ not increase too much.  For the special case where $m=O(n)$, and the number of violated constraints $\cM_\Upsilon^\delta=O(1)$  for some absolute constant $\delta>0$, we achieved a mixing time in $\Ot(\kappa n)$, which matches RMW mixing time in unconstrained sampling. 

\subsection{Practical Matters}\label{sec_practical_matters}
Finally, since our main results rely on the existence of warm starts and focus on the mixing time analysis which ignores per-iteration cost, we devote the rest of this section to discuss per-iteration complexity and how to obtain a warm initialization. 

\paragraph{Per-Iteration Complexity}
In our main results (Theorem \ref{th_main}, \ref{th_weakly} and \ref{th_prob_ball_intersection}), we assume that in each step, we can compute the local metric $G(x)$, its inverse $G(x)^{-1}$ and its determinant $\det G(x)$ exactly. The mixing times in Theorem \ref{th_main}, \ref{th_weakly} and \ref{th_prob_ball_intersection} are given without detailing per-step complexity. Under a computational model where arithmetic operations $\{+,-,\times,\div,\sqrt{\quad}\}$ over real numbers can be done exactly, the number of arithmetic operations needed for \softlog is $O(\max\braces{m,n}n^{\omega-1})$. For \regdw with \reglewis, we also need to compute the Lewis weights defined in Eq.~\eqref{eq_Lewis_weights_def}, which has to be done approximately. The high-accuracy solver for Lewis-weights in \cite{2022Fazel-high-precision}  can be applied with complexity $\Ot(\max\braces{m,n}n^{\omega-1})$.  Detailed per-step complexity analysis is provided in Appendix \ref{appendix_per_step_complexity}.

\paragraph{Feasible $M$-Warm Start}
Our main results (Theorem \ref{th_main}, \ref{th_weakly} and \ref{th_prob_ball_intersection}) assume an $M$-warm initial distribution $\mu_0$. For the truncated target distribution in Eq.~\eqref{eq_distri}, finding a suitable $\mu_0$ with small $M$-warmness is non-trivial. Here we proposed one initialization approach. Suppose that $K$ is contained within a ball of radius $\outr$ and $K$ contains a ball of radius $\inr$, there exists a uniform distribution over a certain ball $\ball(x_0,r_0)$ to be $M$-warm where $M$ satisfies
\begin{equation}\label{eq_warmness_bound}
\log M \leq 1 +n\log\frac{3\outr}{\inr}+n\cdot\max\braces{\frac12\log\parenth{{\beta\outr^2}},\log\parenth{{2\beta\outr\vecnorm{x^\dag-x^\star}2}}},
\end{equation}
where  $f$ is $\beta$-smooth, $x^\star$ denotes the global mode and $x^\dag$ denotes the mode of $\pi(x)\propto e^{-f(x)}$ within the polytope $K$,
\begin{equation}\label{eq_optimz_global_local}
x^\star \defn\underset{x\in\real^n}{\arg\min} f(x), \quad x^\dag\defn \underset{x\in K}{\arg\min} f(x).
\end{equation}
Both $x^\star$ and $x^\dag$ are unique due to the $\alpha$-convexity assumption on $f$. To compute $x_0$ and $r_0$, we need to solve the optimization problem in Eq.~\eqref{eq_optimz_global_local}. A $\delta$-approximation of $x^\star$ and $x^\dag$ can be computed in $\Ot(\kappa\log\frac{1}{\delta})$ steps of gradient descents/ projected gradient descents respectively  \cite{bubeck_convex_2015}. 
The proof of Eq.~\eqref{eq_warmness_bound} is left to Appendix \ref{sec_disc_warm_start}. Although our mixing time results do not rely on the $\outr$-boundedness and $\inr$-inclusion of $K$, these assumptions often arise in sampling literature  \cite{Kannan1997RandomWA,Lov06hrcorner, lovasz_simulated_2006,NEURIPS2023_mangoubi}.


Note that a factor of $n$ is lost in our mixing time upper bounds in Theorem \ref{th_main}, \ref{th_weakly}, and  \ref{th_prob_ball_intersection} if we plug in the warmness bound in Eq.~\eqref{eq_warmness_bound}. However, we argue that this warmness bound is still reasonable in several aspects. In unconstrained logconcave sampling, a feasible start concentrated around the mode often has a warmness $M=O(\kappa^n)$ \cite{dwivedi2019log}, which is also exponential in $n$. As a result, a factor of $n$ is also lost if the $M$-warmness is plugged into the $\log M$ term in the mixing time (see~\cite{dwivedi2019log}).   Similarly, it is also common to lose a factor of $n$ in general-purpose samplers like \hitandrun due to non-warmness \cite{Lov06hrcorner}. Nonetheless, due to the high accuracy nature of Dikin walks, the mixing time from our initialization depends logarithmically on $\frac{\outr}{\inr}$, $\beta\outr^2$ and $\beta\outr\vecnorm{x^\dag-x^\star}{2}$. Hence, as long as these terms scale polynomially in $n$, the additional complexity from them remains  logarithmic. Reducing the extra $n$-factor in mixing time bounds while using feasible initializations remains an area for future exploration, with techniques such as Gaussian cooling~\cite{pmlr-v247-kook24b} and conductance profile~\cite{chen_fast_2020} presenting promising possibilities.


\subsection{Main Lemmas and Proof Outline}\label{sec_main_lemmas}
Now we provide the outline and the necessary lemmas for establishing mixing times of \softdw for sampling truncated logconcave distributions.

We bound the mixing time by finding a lower bound for the conductance of the transition kernel $\barT$, which is the transition kernel $\cT$ after lazification, and the connection between the mixing time and conductance can be found in \cite{lovasz1993random}. For the sake of completeness, we list it as Lemma \ref{lem_Lovas}. Since the lazification $\barT$ at most shrinks the conductance by a factor of $2$, we only need to bound the conductance of the transition kernel $\cT$ before lazification. 

To bound the conductance, we first bound the overlap between transition kernels $\cT_x$ and $\cT_y$ (measured in total-variation distance) for close enough (measured in the local metric $G(x)$) $x$ and $y$ in $K$. This transition overlap result is summarized as Lemma \ref{lem_TV_control}.

Lemma \ref{lem_TV_control} shows that two subsets of $K$ with bad conductance $A_1'$ and $A_2'$ are far away in local metric $G$. To continue to get a lower bound of conductance, we need to show that $K\setminus (A_1'\cup A_2')$ is also large if both $A_1'$ and $A_2'$ are large, so we need a certain type of isoperimetric inequality. Since we need to consider both the logconcave distribution and the polytope which it is truncated on, we design a mixed metric of Euclidean and cross-ratio distances, and the isoperimetric inequality under this mixed metric is listed as Lemma~\ref{lem_isoperimetric}.

We first introduce  the important concepts of ($s$-)conductance of Markov chains. Given the transition kernel $\cT$ of a Markov chain, and $\Pi$ be its stationary distribution, we  define its conductance $\Phi$ and $s$-conductance $\Phi_s$ to be (we require $0<s<1/2$):

\begin{equation}\label{eq_conductance_def}
 \Phi\defn \underset{0<\Pi(A)\leq \frac12 }\inf\frac{\int_{A}{\cT}_u(A^c)\Pi(du)}{\Pi(A)}, \quad     
\Phi_s\defn \underset{s<\Pi(A)\leq \frac12}{\inf}\frac{\int_A\cT_u(A^c)\Pi(du)}{\Pi(A)-s}.
\end{equation}

\begin{lemma}[\cite{lovasz1993random},\cite{lovasz_hit-and-run_1999}]\label{lem_Lovas}
For a reversible lazy Markov chain defined by the transition kernel $\widetilde{\cT}$, let $\Pi$ denote its unique stationary distribution and $\widetilde\Phi$, $\widetilde\Phi_s$ denote its conductance, $s$-conductance respectively. Given an $M$-warm start $\mu_0$, the convergence to $\Pi$ can be controlled using the conductance: 
\begin{equation}\label{eq_Lovas}
    \vecnorm{\widetilde{\cT}^k(\mu_0)-\Pi}{TV}\leq \sqrt{M}\parenth{1-\frac12\widetilde\Phi^2}^k.
\end{equation}
One can achieve a similar result using the notion of $s$-conductance: 
\begin{equation}\label{eq_s_Lovas}
\vecnorm{\tcT^k(\mu_0)-\Pi}{TV}\leq Ms+M\parenth{1-\frac{\widetilde\Phi_s^2}{2}}^k\leq Ms+M\exp\parenth{-\frac{k\widetilde\Phi_s^2}{2}}.
\end{equation}
\end{lemma}
Since Lemma \ref{lem_Lovas} is well-known and its proof can be found in \cite{lovasz1993random}, we omit its proof in this paper. 

In controlling the transition overlap in Lemma \ref{lem_TV_control}, the acceptance rate is controlled around $1/2$ globally (we extend the close coupling argument in \cite{andrieu_explicit_2024} to our case of assymetrical proposal distributions), we do not need to cut off small probability regions, and using the conductance $\Phi$ is adequate for Theorem \ref{th_main} and \ref{th_weakly}. However, when showing the mixing time depends on a fraction of linear constraints in Theorem \ref{th_prob_ball_intersection}, we need to cut off the part of $K$ that is too close to the remaining linear constraints, thus the $s$-conductance and Eq.~\eqref{eq_s_Lovas} are needed. 

\begin{lemma}[Transition Overlap]\label{lem_TV_control}
Let $G$ be a SSC, ASC, LTSC matrix function defined on $K$, and we further assume that $G(x)\succeq \beta I_n$ for all $x\in K$. Let $\cT$ be the transition overlap of \softdw as defined in Eq.~\ref{eq_Transition_kernel_def}, then we can set a step-size $r>0$ such that for all $x,y$ such that $\vecnorm{x-y}{G(x)}\leq \frac{r}{10\sqrt{n}}$, we have  
\begin{equation*}
\vecnorm{\cT_x-\cT_y}{TV}\leq \frac{4}{5}     
\end{equation*}
\end{lemma}

Note that Lemma \ref{lem_TV_control} is very similar to the one-step coupling argument in \cite{pmlr-v247-kook24b}, and we can directly apply Lemma B.3 in \cite{pmlr-v247-kook24b} and get the transition overlap we desire. However, we were motivated by the close coupling argument in \cite{andrieu_explicit_2024}, and we extend it to our assymetrical proposal distributions, which is  a different angle from the direct computation appeared in \cite{pmlr-v247-kook24b}. We also provide a modular and streamlined format of proof, making it convenient for readers and facilitates potential applications. Namely, we separate the proof of Lemma~\ref{lem_TV_control} into three parts: bounding the acceptance rate to be around $1/2$ globally,  extending the close coupling argument in \cite{andrieu_explicit_2024}, and controlling the TV-distance between two proposal distributions. So we still provide our proof of Lemma \ref{lem_TV_control} in Appendix~\ref{appendix_transit_overlap}.

Now we can give the isoperimetric inequality under a combination of Euclidean distance and cross-ratio distance in this extended sense. 

\begin{lemma}\label{lem_isoperimetric}
    Suppose $\Pi$ is a probability distribution supported on a (possibly unbounded) convex set $K\subseteq \real^n$ and  is more logconcave than Gaussian with covariance $\frac{1}{\alpha}I_n$ (see Eq.~\eqref{eq_more_log_concave}). Assume $K$ is partitioned into three measurable sets $K=S_1\sqcup S_2\sqcup S_3$, then we have 
    \begin{equation}
        \Pi(S_3)\geq d'(S_1,S_2)\Pi(S_1)\Pi(S_2),
    \end{equation}
where $d'$ is a mixed distance defined by $d'(x,y)=\max\braces{d_K(x,y),\log(2)\sqrt{\alpha}\|x-y\|_2}$ and the corresponding distance between two sets $S_1$ and $S_2$ is defined as: $d'(S_1,S_2)\defn \underset{(x,y)\in S_1\times S_2}{\inf}d'(x,y)$. Note that here $d_K$ refers to the cross-ratio distance defined over $K$ in the extended sense as in Definition \ref{def_cross_ratio_unbounded}. 
\end{lemma}

The reason we need Lemma \ref{lem_isoperimetric} is that, to use Lemma \ref{lem_TV_control} to bound the conductance, we need a suitable isoperimetric inequality for the local metric $G$. The conventional wisdom is to use the isoperimetric inequality for cross-ratio distance $d_K$ over the polytope $K$ \cite{NEURIPS2023_mangoubi,pmlr-v247-kook24b}, and furthermore an upper bound of $G$ by local metric $d_K$ is needed. This upper bound is often found by the $\bar\nu$-symmetry of a local metric introduced in \cite{laddha2020strong}. However, such an upper bound does not exist in our case, since the Euclidean term $\lambda I$ simply can not be bounded by cross-ratio distance. In fact, for unbounded convex set $K$, we can always find some $x,y$ such that the Euclidean distance $\beta \vecnorm{y-x}{2}$ is of higher order of magnitude than $d_K(x,y)$.

We defer the proof of Lemma \ref{lem_isoperimetric} to Section \ref{sec_new_iso_strongly}. Isoperimetric inequality under either metric (cross-ratio distance/ Euclidean distance) is already well-known, while their combination is not trivial in that the new distance $d'(x,y)\defn \max\braces{d_K(x,y),\sqrt{\alpha}\log2\vecnorm{x-y}2}$ for two sets $d'(S_1,S_2)$ may not be achieved when either $d_K(S_1,S_2)$ or $d_\text{euclid}(S_1,S_2)$ is achieved. In the proof, we use a localization argument proposed in \cite{kannan1995isoperimetric} to reduce the $n$-dimensional integral to $1$-dimensional integral. 

In order to extend the mixing time results to weakly logconcave distributions as in Theorem~\ref{th_weakly}, we need to extend this new isoperimetry to weakly logconcave distributions. However, given a convex body $K$, it is well-known that the cross-ratio distance $d_K$ in Definition~\ref{def_cross_ratio_unbounded} does not satisfy triangle inequality, thus $d_K$ is not a metric. To circumvent this inconvenience, we use the Hilbert metric $\dH_K$ over a convex set $K$:
\begin{equation}
\dH_K(x,y)\defn \log\parenth{1+d_K(x,y)},
\end{equation}
where $d_K$ is the cross-ratio distance  defined in the extended sense for possibily unbounded convex sets as in Definition $\ref{def_cross_ratio_unbounded}$. Unlike cross-ratio distances, Hilbert metric satisfies the triangle inequality. As we will see, since $\dH_K(x,y)\approx d_K(x,y)$ for close $x,y$ by definition, thus changing $d_K$ to $\dH_K$ would not impact the application of the isoperimetric inequality to bound mixing times. Now we give the new isoperimetry for weakly logconcave distributions. 

\begin{lemma}\label{lem_isoperimetric_weakly}
Suppose $\Pi$ is a logconcave distribution supported on a open convex set $K\subseteq \real^n$, assume that $\Pi$ has a bounded covariance matrix $\Sigma_\pi$, and let $\spectr=\vecnorm{\Sigma_\pi}{2}$ denote its spectral norm. Then for any Borel measurable decomposition $K=S_1\sqcup S_2\sqcup S_3$, we have
\begin{equation}
\Pi(S_3)\geq \frac{1}{6\max\braces{1,\psi_n}}\cdot\dn(S_1,S_2)\Pi(S_1)\Pi(S_2),
\end{equation}
where $\dn$ is a mixed metric defined by $\dn(x,y)=\max\braces{\frac{(\log2)}{\sqrt{\spectr}}\vecnorm{x-y}{2},\dH_K(x,y)}$, and $\psi_n$ is the KLS constant. 
\end{lemma}

\cite{klartag_logarithmic_2023} proved the KLS constant $\psi_n$ satisfies $\psi_n=O(\sqrt{\log n})$, thus we lost at most a logarithmic factor in the mixing time for transitioning to weakly logconcave distributions. 

The proof of Lemma \ref{lem_isoperimetric_weakly} is left to Section \ref{sec_new_iso_weakly}. The key technique behind this extension to weakly logconcave distributions is the stochastic localization scheme introduced in \cite{eldan_thin_2013}. The main idea of stochastic localization is to smoothly  multiply the weaky logconcave distribution with a Gaussian part, and the modified distribution is $\alpha$-strongly logconcave. As a result, we can apply our new isoperimetry (Lemma~\ref{lem_isoperimetric}) on $\alpha$-strongly logconcave distributions to the modified distribution. While stochastic localization has been studied extensively to upper bound the KLS constant $\psi_n$ \cite{lee_eldans_2019,chen_almost_2021,klartag_bourgains_2022,klartag_logarithmic_2023}, there are still several difficulties to apply it to our new isoperimetry. The evolution of the measures of sets of arbitrary initial size (for $S_1,S_2,S_3$ as  Lemma \ref{lem_isoperimetric_weakly}) is unclear \cite{chen2022hitandrun}. We circumvent this issue by instead proving a boundary version of the isoperimetric inequality. To do that, we introduce the boundary measures \cite{bobkov_isoperimetric_1997} for a general metric space, and  also the co-area formula \cite{bobkov_isoperimetric_1997} for a general metric spaces that allows us to convert back to Lemma \ref{lem_isoperimetric_weakly}. To prove the boundary version of isoperimetry, we used an bound on the evolution of $L^2$-functions from \cite{klartag_logarithmic_2023} over the distribution $\Pi$. Finally, we show that the target distribution is compatible with the topology under $\dn$ by proving the open convex set $K$ under the mixed metric $\dn$ has the usual Euclidean topology, so we can apply the co-area formula correctly. 

\section{Combining Euclidean \& Hilbert Isoperimetries}\label{sec_new_iso}
In this section, we prove a new isoperimetric inequality which uses a metric that combines the Euclidean metric and the cross-ratio distance/Hilbert metric over the convex set $K$. In Section \ref{sec_new_iso_strongly} we prove Lemma \ref{lem_isoperimetric}, the isoperimetric inequality for $\alpha$-strongly logconcave distributions truncated on a convex set $K$. In Section \ref{sec_new_iso_weakly}, using stochastic localization, we extend the isoperimetric inequality to weakly log-concave distributions truncated on $K$ with finite covariance, and prove Lemma \ref{lem_isoperimetric_weakly}. 

\subsection{Isoperimetry for Strongly Logconcave Distributions}\label{sec_new_iso_strongly}

For unconstrained logconcave sampling on Euclidean space, an isoperimetric inequality for $\alpha-$strongly logconcave distributions using Euclidean metric is well-studied~\cite{cousins_gaussian_2018} (listed as Fact~\ref{fact_iso_eu}). While in constrained sampling, an isoperimetric inequality for logconcave distributions using cross-ratio distance is often used instead~\cite{lovasz2007geometry} (listed as Fact~\ref{fact_iso_cross}). This is because one can often use cross-ratio distance to bound the metric induced by the polytope using the concept of $\bar\nu-$symmetry introduced in \cite{laddha2020strong}. 

Both our local metrics in Definition~\ref{def_logarithmic} and~\ref{def_Lewis} can be viewed as the summation of a Hessian metric and a Euclidean metric. Since one metric cannot be bounded by the other with dimension-independent factors, the two isoperimetric inequalities above are not useful for obtaining tight mixing time bounds. We propose a new isoperimetric inequality that combines Fact~\ref{fact_iso_eu} and Fact~\ref{fact_iso_cross}. Following the localization lemma for proving isoperimetric inequalities \cite{lovasz2007geometry,cousins_gaussian_2018}, we transform the inequality over $n$-dimensional integrals into an inequality over $1$-dimensional integrals. Then we check the 1-dimensional isoperimetric inequality that uses the maximum of two metrics.


\begin{fact}[Theorem 5.4 from \cite{cousins_gaussian_2018}]\label{fact_iso_eu}
Assume $\Pi$ is a probability distribution on $\real^n$ and is more logconcave than Gaussian with covariance $\frac{1}{\alpha}I_n$ (see Eq.~\eqref{eq_more_log_concave})  then we have 
\begin{equation*}
\Pi(S_3)\geq {(\log2)\sqrt{\alpha}\cdot d_{\mathrm{euclid}}(S_1,S_2)}\Pi(S_1)\Pi(S_2).
\end{equation*}
\end{fact}

\begin{fact}[Theorem 2.5 from\cite{lovasz2007geometry}]\label{fact_iso_cross}
Let $\Pi$ be a logconcave distribution whose support $K\subseteq \real^n$ is partitioned into three measurable sets: $K=S_1 \sqcup S_2\sqcup S_3$, then we have 
\begin{equation*}
    \Pi(S_3)\geq d_K(S_1,S_2) \Pi(S_1)\Pi(S_2) .
\end{equation*}
\end{fact}

\begin{fact}[Bounds for Integrals over $\real^1$, see \cite{lovasz2007geometry,kannan1995isoperimetric}]\label{fact_1d_integrals_bound}
For any $a<b<c<d$ in $\real$, and let $g:\real\to \real_{++}$ be a logconcave function
, we have  
\begin{equation*}
\frac{\int_a^d g(t)dt\int_b^c g(t)dt}{\int_a^b g(t)dt\int_c^d g(t)dt}\geq\frac{(d-a)(c-b)}{(b-a)(d-c)}.    
\end{equation*}
Moreover, for any one-dimensional isotropic function $f$, and any partition $S_1,S_2,S_3$ of the real line, 
    \begin{equation*}
    \Pi_f(S_3)\geq \log(2) d_{\mathrm{euclid}}(S_1,S_2)\Pi_f(S_1)\Pi_f(S_2).
    \end{equation*}

\end{fact}

Now we are prepared to give the proof of the new isoperimetric inequality (Lemma \ref{lem_isoperimetric}), and the proof is similar to Theorem 2.5 in \cite{lovasz2007geometry}.

\begin{proof}[Proof of Lemma \ref{lem_isoperimetric}]
Let $\pi$ be the (unnormalized) density function of $\Pi$ supported on $K$, then $f$ is continuous over $K$ . Let $h_i$ be the indicator function of $S_i$ for $i=1,2,3$, and let $h_4$ be the indicator function of $K$. It suffices to prove that:
\begin{equation}\label{eq_isoperimetric_n}
    d'(S_1,S_2)(\int_{\real^n} \pi h_1)(\int_{\real^n} \pi h_2)\leq (\int_{\real^n} \pi h_3)(\int_{\real^n} \pi h_4).
\end{equation}

For any $a,b\in K$, any non-negative linear function $l:[a,b]\to \real_{++}$. Set $v(t)=ta+(1-t)b$, and we define the following integral over the needle $N=([a,b],l)$:
\begin{equation*}
    J_i=\int_0^1h_i(v(t))\pi(v(t))l(v(t))^{n-1}dt.
\end{equation*}

According to the localization lemma (the detailed discription and proof appeared as Corollary 2.2 in \cite{kannan1995isoperimetric}), we can reduce the problem to $n$-dimensional integral inequality to $1$-dimensinoal integrals. It suffices to prove that:
\begin{equation}\label{eq_localization_1d}
    d'(S_1,S_2)J_1\cdot J_2\leq J_3\cdot J_4.
\end{equation}
For convenience, we define $I_i\defn\{t|h_i(v(t))>0\}$ for $i=1,2,3,4$ to be subsets of $[0,1]$, and $I_4=[0,1]=I_1\sqcup I_2\sqcup I_3$. We first prove the special case that $I_1,I_2,I_3$ are intervals $[0,u_1],[u_2,1]$ and $(u_1,u_2)$ respectively.

If $u_1=0$ or $u_2=1$, then $J_1=0$ or $J_2=0$, thus Eq.~\eqref{eq_localization_1d} is trivially true. If $u_1=u_2$, this implies $v(u_1)=v(u_2)$, since we also have $v(u_1)\in S_1, v(u_2)\in S_2$, thus $d_K(v(u_1),v(u_2))=0$ and $\vecnorm{v(u_1)-v(u_2)}2=0$, this implies $d'(S_1,S_2)\leq d'(v(u_1),v(u_2))=0$, again the Inequality \eqref{eq_localization_1d} is trivially true. So we can assume that $0<u_1<u_2<1$.

Set $c_i=u_ia+(1-u_i)b$ for $i=1,2$. It is clear that $c_i\in S_i$, thus $d'(c_1,c_2)\geq d'(S_1,S_2)$, that is
\begin{equation*}
\max\braces{d_K(c_1,c_2),(\log2)\sqrt{\alpha}\|c_1-c_2\|_2}\geq d'(S_1,S_2).
\end{equation*}
So we only need to prove the two following inequalities:
\begin{equation}\label{eq_reduced_to_1d}
d_K(c_1,c_2){J_1\cdot J_2}\leq {J_3\cdot J_4}, \quad (\log2)\sqrt{\alpha}\vecnorm{c_1-c_2}{2}J_1J_2\leq J_3J_4.    
\end{equation}

We can prove the  first inequality in Eq.~\eqref{eq_reduced_to_1d}. we have 
\begin{equation*}
\quad d_K(c_1,c_2)\overset{(i)}\leq \frac{\vecnorm{b-a}2\vecnorm{c_2-c_1}2}{\vecnorm{c_1-a}2\vecnorm{b-c_2}2}\overset{(ii)}\leq \frac{J_3\cdot J_4}{J_1\cdot J_2},
\end{equation*}
where inequality $(i)$ holds because $a,c_1,c_2,b\in K$ are on the same line, and the intersection $p,q$ of $\widebar{c_1c_2}$ (in the order $p,c_1,c_2,q$) with $K$ always satisfies $\vecnorm{p-c_1}2\geq \vecnorm{c_1-a}2$ and  $\vecnorm{q-c_2}2\geq\vecnorm{b-c_2}2$ including the case when $p$ or $q$ are at the infinity. Inequality $(ii)$ holds since $\pi(v(t))l(v(t))^{n-1}$ is a logconcave function in $t$ and we use Fact~\ref{fact_1d_integrals_bound}.

Then we prove the second inequality in Eq.~\eqref{eq_reduced_to_1d}, for convenience, we can directly employ Fact~\ref{fact_iso_eu} over the $1$-dimensional special case. We observe the integrand of $J_i$ ($i=1,2,3,4$) satisfies  
\begin{equation*}
\begin{aligned}
\pi_{1d}(t)=&h_i(v(t))\pi(v(t))l(v(t))^{n-1}\\
\overset{(i)}\propto &h_i(v(t))q(v(t))\exp\brackets{-\frac\alpha2{v(t)\tp v(t)}}l(v(t))^{n-1}\\
\propto &{ h_i(v(t))q(v(t))l(v(t))^{n-1}} \exp\parenth{-{\alpha}(a-b)\tp bt}\exp\parenth{-\frac{\alpha}{2}\vecnorm{a-b}2^2t^2},
\end{aligned}    
\end{equation*}
where in proportion $(i)$, $q$ is a logconcave function since we assumed $\pi$ to be more logconcave than Gaussian with covariance $\frac{1}{\alpha}I_n$.
We notice that if $g:\real^n\to\real^+$ is a logconcave function and $v:\real\to \real^n$ is linear, then $g\circ v$ is logconcave. Since $h_i, f, l$ are logconcave, thus $h_i(v(t)),q(v(t)),l(v(t))^{n-1}$ are logconcave in $t$, and $\exp(-\alpha (a-b)\tp bt)$ is clearly logconcave. Together with the fact that the product of logconcave functions is logconcave, $\pi_{1d}$ is more logconcave than Gaussian with variance $\alpha^{-1}\vecnorm{a-b}{2}^{-2}$. Viewing $\pi_{1d}$ as a probability density of $\Pi_{1d}$ over $[0,1]$, we apply Fact~\ref{fact_iso_eu}:
\begin{equation*}
\Pi_{1d}(I_3)\geq (\log2) \sqrt{\alpha}\vecnorm{a-b}{2} \Pi_{1d}(I_1)\Pi_{1d}(I_2),
\end{equation*}
here $\Pi_{1d}(I_i)$ satisfies $\Pi_{1d}(I_i)=\frac{J_i}{J_4}$ for $i=1,2,3$. Thus we have 
\begin{equation*}
(\log2)\sqrt{\alpha}\vecnorm{c_1-c_2}{2} \overset{(i)}\leq (\log2)\sqrt{\alpha}\vecnorm{a-b}{2}\leq \frac{J_3J_4}{J_1J_2},
\end{equation*}
where inequality $(i)$ satisfies since $c_1,c_2$ are inside the line segment $\widebar{ab}$. So both inequalities in Eq.~\eqref{eq_reduced_to_1d} are proved, thus we proved that $d'(S_1,S_2)\cdot J_1\cdot J_2\leq J_3\cdot J_4$. 

Following the same combinatorial argument in Theorem 5.2 from \cite{kannan1995isoperimetric}, we can prove the Eq.~\eqref{eq_localization_1d} for general $1$-dimensional measurable sets $I_1,I_2,I_3$, the details are omitted here. As a result of the localization lemma, we have proved the Eq.~\eqref{eq_isoperimetric_n} in $\real^n$ and thus the isoperimetric inequality under $d'$.

\end{proof}

\subsection{Extension to Weakly Logconcave Distributions}\label{sec_new_iso_weakly}
In this section, we prove Lemma \ref{lem_isoperimetric_weakly}, which extends the new isoperimetry in Lemma \ref{lem_isoperimetric} to weakly logconcave distributions. Instead of strong logconcavity, we only assume that the distribution $\Pi$ is supported on $K\subseteq \real^n$ an open and convex set, and $\Pi$ has a bounded covariance matrix $\Sigma_\pi$. 

Without loss of generality, we can further assume that  $\Pi$ is an isotropic logconcave distribution. In other words, $\Sigma_\pi=I_n$. This is because one can apply an affine transformation to $\Pi$ so that $\Sigma_\pi=I_n$, and the effect of this transformation on the Euclidean and Hilbert metrics can be quantified. To be rigorous, we provide a proof in Appendix \ref{appendix_reduction_to_isotropy}. 

To prove Lemma \ref{lem_isoperimetric_weakly}, it is easier to first prove an isoperimetric inequality with boundary measures, and then use the co-area formula in a metric space to recover the isoperimetric inequality on three sets $K = S_1\sqcup S_2\sqcup S_3$.  To do this, we first extend the notion of boundary measures to a more general metric  than Euclidean metric as in  Eq.~\eqref{eq_def_boundary_euclid}, and we also introduce gradient modulus in a general metric space. These two definitions are listed in Definition~\ref{def_boundary_measure}. We also list the co-area formula for a general metric space in Lemma \ref{lem_coarea}.

\begin{definition}[boundary measure \cite{bobkov_isoperimetric_1997}]\label{def_boundary_measure}
Let $(X,d)$ be a separable metric space, and let $\cB(d)$ denotes the Borel $\sigma$-algebra induced by the topology of $(X,d)$. Assume $\Pi$ is a probability measure defined over $\cB(d)$.  If $B\in \cB(d)$, then the boundary measure $\Pi^+$ of $B$ is defined to be
\begin{equation*}
    \Pi^+(B)=\underset{h\to 0^+}{\lim\inf} \frac{\Pi(B^h)-\Pi(B)}{h},
\end{equation*}
where $B^h\defn \braces{x\in X|\exists a\in B,d(x,a)<h}$ is the open $h$-neighborhood of $A$. For any function $f:X\to \real$, we define the modulus of its gradient $\abss{\nabla f(x)}$\footnote{$\abss{\nabla f(x)}$ is $\cB(d)$-measurable if $f$ is continuous} to be:
\begin{equation}\label{eq_def_modulus_grad}
    \abss{\nabla f(x)}=\underset{d(x,y)\to 0^+}{\lim\sup}\frac{\abss{f(x)-f(y)}}{d(x,y)},
\end{equation}
and we set $\abss{\nabla f(x)}=0$ if $x$ is an isolated point in $X$ 
\end{definition}

\begin{lemma}[co-area inequality \cite{bobkov_isoperimetric_1997}]\label{lem_coarea}
Let $(X,d)$ be a separable metric space, and let $\Pi$ denotes a probability measure defined on $\cB(d)$. Assume $\rho$ is a function on $X$ with a finite Lipschitz constant, then:
\begin{equation}
    \int_{X}\abss{\nabla \rho(x)}\Pi(dx)\geq \int_{-\infty}^{+\infty}\Pi^+\braces{x\in X|\rho(x)>t}dt.
\end{equation}
\end{lemma}
The proof of Lemma \ref{lem_coarea} can be found in \cite{bobkov_isoperimetric_1997}. One issue to apply the lemma is that our target probability distribution $\Pi$ in Equation \eqref{eq_distri_weakly} is only defined over the Borel $\sigma$-algebra generated by the Euclidean topology, so to use boundary measure and co-area inequality for the new metric, we need to show $\Pi$ is also measurable over the Borel $\sigma$-algebra generated by our new metric $\dn$. Formally, we show that the open convex set $K$ equipped with the mixed metric  $\dn(x,y)\defn \max\braces{(\log2)\vecnorm{x-y}2,\dH_K(x,y)}$ satisfies the condition to be a metric space, and the topology induced by $\dn$ is the usual Euclidean topology, and we summarize this in the following Lemma \ref{lem_homeomorphism}.

\begin{lemma}\label{lem_homeomorphism}
Let $K$ denotes an open convex subset of $\real^n$,  we define $\dn$ to be the following binary function over $K$:
\begin{equation}
\dn(x,y)\defn \max\braces{(\log2)\vecnorm{x-y}{2},\dH_K(x,y)},
\end{equation}
then  $(K,\dn)$ is a metric space and $\dn$ induces the  Euclidean topology over $K$.
\end{lemma}
The proof of Lemma \ref{lem_homeomorphism} is left to Appendix \ref{appendix_homeomorphism}, where we check the definitions of metric spaces, and verify that the identity map is a homeomorphism between $(K,\dn)$ and $(K,\vecnorm{\cdot}{2})$. 
Next, we prove a new isoperimetric inequality with boundary measures for isotropic logconcave distributions.

\begin{lemma}\label{lem_isoperimetric_boundary}
Assume $\Pi$ is an isotropic logconcave distribution supported on an open convex set $K\subseteq \real^n$. Then for any Borel measurable set $B\subseteq K$, we have 
\begin{equation}
\Pi^+(B)\geq \frac{1}{6\max\braces{1,\psi_n}}\Pi(B)\brackets{1-\Pi(B)},
\end{equation}
where $\psi_n$ denotes the KLS constant, and the metric $\dn$ in defining boundary measure $\Pi^+$ is the mixed metric given by:
\begin{equation*}
\dn(x,y)\defn\max\braces{(\log2)\vecnorm{x-y}{2},\dH_K(x,y)}
\end{equation*}
\end{lemma}

It is worth mention that here the KLS constant $\psi_n$ is still under the usual definition by Euclidean metric as in Eq.~\eqref{eq_KLS_def} and Eq.~\eqref{eq_def_boundary_euclid}. So we still have  $\psi_n=O(\sqrt{\log n})$ by \cite{klartag_logarithmic_2023}. However, the boundary measure $\Pi^+$ in Lemma \ref{lem_isoperimetric_boundary} is defined using our mixed metric $\dn$. 

The idea to prove Lemma \ref{lem_isoperimetric_boundary} is to use stochastic localization, where a family of random measures $\parenth{\Pi_t}_t$ is defined for $t\geq 0$ by stochastic differential equations. At time $t > 0$,  the measure $\Pi_t$ is $t$-strongly logconcave, thus we can apply the mixed isoperimetry for strongly-logconcave distributions in Lemma~\ref{lem_isoperimetric}. Meanwhile, Taking $t \leq \psi_n^{-2}$ allows us to control the right-hand side via approximate conservation of variance in Lemma~\ref{lem_evolution_function}. 

\paragraph{Stochastic localization} Now we introduce stochastic localization, following the notation in~\cite{chen_almost_2021}. Assume $\Pi$ is a logconcave distribution. For $t\geq 0$ and $\theta\in\real^n$, let $\pi_{t,\theta}$ denote the following probability density function
\begin{equation*}
\pi_{t,\theta}(x)=\frac{1}{Z(t,\theta)}\exp\parenth{\theta\tp x-\frac{t}2{x\tp x}}\pi(x).
\end{equation*}
The mean of the probability density $\pi_{t,\theta}$ is denoted by
\begin{equation*}
a(t,\theta)=\int_{\real^n} x\,\pi_{t,\theta}(x)dx
\end{equation*}
We consider the stochastic evolution of target distribution $\Pi$ by defining the following stochastic differential equation in terms of $\theta$
\begin{equation}\label{eq_stochastic_theta}
    d\theta_t=dW_t+a(t,\theta_t)dt,\quad \theta_0=0,
\end{equation}
where $(W_t)$ is standard Brownian motion. 
Then the distribution at time $t$ is defined as $\pi_t(x)\defn \pi_{t,\theta_t}(x)$. The existence and uniqueness of the solution in time $[0,t]$ for Equation \eqref{eq_stochastic_theta} can be shown by verifying that $a(t,\theta)$ is a bounded function that is Lipschitz with respect to $t$ and $\theta$. Using It\^{o} formula, we obtain that for $x\in\real^n$:
\begin{equation*}
    d\pi_t(x)=\pi_t(x)\angles{x-a_t,dW_t}. 
\end{equation*}
Thus the process $(\pi_t(x))_{t}$ is a martingale with respect to the filtration induced by the Brownian motion, and we have for $t\geq0$ and $x\in\real^n$:
\begin{equation*}
\Exs \pi_t(x)=\pi_0(x)=\pi(x).
\end{equation*}

The stochastic localization technique developed in \cite{eldan_thin_2013} has been successful in upper bounding the KLS constant $\psi_n$. A series of works  \cite{lee_eldans_2019,chen_almost_2021,klartag_bourgains_2022,klartag_logarithmic_2023} use stochastic localization to provide tighter upper bounds for the KLS constant $\psi_n$, with the state-of-the-art bound $\psi_n=O(\sqrt{\log n})$ obtained by~\cite{klartag_logarithmic_2023}. While it has an elaborate proof, for the purpose of our result, we only need the approximate conservation of variance along the stochastic process in Lemma $3.1$ in \cite{klartag_logarithmic_2023} which we restate below.

\begin{lemma}[Approximate conservation of variance, Lemma $3.1$ in \cite{klartag_logarithmic_2023}]\label{lem_evolution_function} For any $t\geq 0$ and $f\in L^2(\Pi)$:
\begin{equation*}
    \Exs\Var_{\pi_t}\parenth{f}\leq \Var_{\pi_0}(f)\leq \parenth{2+\frac{t}{\lambda_0}}\Exs\Var_{\pi_t}(f),
\end{equation*}
where $\lambda_0=\frac{1}{C_P(\Pi)}$ is the reciprocal of the Poincar\'e constant of the distribution $\Pi$. 
\end{lemma}

The proof of Lemma \ref{lem_evolution_function} can be found in \cite{klartag_logarithmic_2023} and is omitted. Now we are ready to prove Lemma \ref{lem_isoperimetric_boundary}.
\begin{proof}[Proof of Lemma \ref{lem_isoperimetric_boundary}]
Given $\Pi$ an isotropic logconcave distribution, then we define the following mixed metric $\dn$  by
\begin{equation}\label{eq_dn_isotropy}
    \dn(x,y)=\max\braces{\dH_K(x,y),(\log2)\vecnorm{x-y}{2}}.
\end{equation}
For any Borel set $B\subseteq K$, $\Pi^+(B), \Pi(B)$ are well-defined quantities according to the homeomorphism between $\dn$ and $\vecnorm{\cdot}{2}$ (see Lemma \ref{lem_homeomorphism} and Definition \ref{def_boundary_measure}). We run stochastic localization as defined in Equation \eqref{eq_stochastic_theta}  until $t=\psi_n^{-2}$, and we have 
\begin{equation}\label{eq_boundary_measure_I}
\begin{aligned}
\Pi^+( B) &= \underset{h\to 0+}{\lim\inf}\frac{\Pi(B^h)-\Pi(B)}{h}\overset{(i)}=\underset{h\to 0+}{\lim\inf}\frac{\Exs\braces{\Pi_t(B^h)-\Pi_t(B)}}{h}\\
&\overset{(ii)}\geq \Exs\braces{\underset{h\to 0+}{\lim\inf}\frac{\Pi_t(B^h)-\Pi_t(B)}{h}} =\Exs\braces{\Pi_t^+(B)},
\end{aligned}
\end{equation}
where equality $(i)$ used the martingale property of $(\Pi_t)_t$, and inequality $(ii)$ applies Fatou's lemma. In order to lower bound $\Pi_t^+(B)$, we can apply Lemma~\ref{lem_isoperimetric} to the distribution $\Pi_t$ since it is $t$-strongly logconcave. The distance $d'$ in Lemma \ref{lem_isoperimetric} can be lower bounded by $\frac{1}{\max\braces{1,\psi_n}}\dn$ by definition in Eq.~\eqref{eq_dn_isotropy}. To be explicit, for any measurable decomposition $K\defn S_1\sqcup S_2\sqcup S_3$, we have
\begin{equation*}
\begin{aligned}
\Pi(S_3)&\overset{(i)}\geq \underset{x\in S_1,y\in S_2}\inf \max\braces{d_K(x,y),(\log2)\sqrt{t}\vecnorm{x-y}{2}}\Pi(S_1)\Pi(S_2)\\
&\overset{(ii)}\geq \underset{x\in S_1,y\in S_2}\inf \max\braces{\dH_K(x,y),(\log2)\sqrt{t}\vecnorm{x-y}{2}}\Pi(S_1)\Pi(S_2)\\
&\overset{(iii)}\geq \underset{x\in S_1,y\in S_2}\inf \frac{1}{\max\braces{1,\psi_n}}\max\braces{\dH_K(x,y),(\log2)\vecnorm{x-y}{2}}\Pi(S_1)\Pi(S_2)\\
&=\frac{1}{\max\braces{1,\psi_n}}\dn(S_1,S_2)\Pi(S_1)\Pi(S_2),
\end{aligned}
\end{equation*}
where inequality $(i)$ follows from Lemma \ref{lem_isoperimetric}, inequality $(ii)$ follows from the fact $\log(1+x)\leq x$ for all $x>0$, and inequality $(iii)$ holds  since $1\leq \max\braces{1,\psi_n}$. 
Then the boundary version of the above isoperimetric inequality is easily implied if we take the limit $S_1\to B$, $S_2\to K\setminus B$ and $S_3\to \partial B$. So at $t=\psi_n^{-2}$ we have 
\begin{equation}
    \Pi_t^+(B)\geq \frac{1}{\max\braces{1,\psi_n}}\Pi_t(B)\parenth{1-\Pi_t(B)}.
\end{equation}
As a result, we can continue to bound Eq.~\eqref{eq_boundary_measure_I} by 
\begin{equation}
\begin{aligned}
\Pi^+(B)&\geq \Exs\braces{\Pi^+_{t}(B)}\geq \frac{1}{\max\braces{1,\psi_n}}\Exs\braces{\Pi_t(B)\parenth{1-\Pi_t(B)}}\\
&=\frac{1}{\max\braces{1,\psi_n}}\Exs\braces{\var_{\pi_t}\parenth{\mathbf{1}_B}}\\
&\overset{(i)}\geq \frac{1}{\max\braces{1,\psi_n}}\frac{1}{\brackets{2+C_P(\Pi)\cdot\psi_n^{-2}}}\var_{\Pi}(\mathbf{1}_B)\\
&\overset{(ii)}\geq \frac1{6\max\braces{1,\psi_n}} \Pi(B)\parenth{1-\Pi(B)},
\end{aligned}
\end{equation}
where inequality $(i)$ is an application of Lemma \ref{lem_evolution_function} with $t=\psi_n^{-2}$, and inequality $(ii)$ follows from the fact that Poincar\'e constant and the isoperimetric constant are closely related due to Cheeger's inequality~\cite{cheeger1970lower} and Buser and Ledoux~\cite{buser1982note,ledoux2004spectral}. Namely, for any probability measure $\mu$ on $\real^n$, $\frac{1}{4}\leq \psi_\mu^2 / C_P(\mu) \leq \pi$.
Then $C_P(\Pi)\leq 4\psi_n^2$ since $\Pi$ is isotropic and logconcave.
\end{proof}

With Lemma \ref{lem_homeomorphism} we know that $(K,\dn)$ induces the usual Euclidean topology, and this indicates that the $\sigma$-algebra over which the probability measure $\Pi$ is defined is the same as the Borel $\sigma$-algebra induced by $\dn$-topology.  Let $\cB(K)$ denotes Borel $\sigma$-algebra induced by Euclidean topology over $K$, assume $\Pi$ is a probability measure supported on $K$.

Having proved the boundary version of the isoperimetric inequality (Lemma \ref{lem_isoperimetric_boundary}), we attempt to prove the corresponding isoperimetric inequality under the decomposition $K=S_1\sqcup S_2\sqcup S_3$ as Lemma \ref{lem_isoperimetric_weakly}. For reasons talked above, we only need to prove Lemma \ref{lem_isoperimetric_weakly} assuming that the distribution $\Pi$ is isotropic.

\begin{proof}[Proof of Lemma \ref{lem_isoperimetric_weakly}]
Based on the discussion in Appendix \ref{appendix_reduction_to_isotropy}, it suffices to consider isotropic logconcave $\Pi$. Let $\dn$ denote the following mixed metric
\begin{equation}
    \dn(x,y)\defn \max\braces{(\log2)\vecnorm{x-y}{2},\dH_K(x,y)}.
\end{equation}
If $\dn(S_1,S_2)=0$, then the isoperimetric inequality is trivial. We can assume that $\dn(S_1,S_2)>0$. Let $\rho$ be defined as 
\begin{equation*}
    \rho(x)\defn \frac{\dn(x,S_1)}{\dn(x,S_2)+\dn(x,S_1)},
\end{equation*}
where we define $\dn(x,S)\defn \inf_{z\in S} \dn(x,z)$. We bound $\abss{\nabla \rho(x)}$, the modulus of the gradient of $\rho$, as follows
\begin{equation*}
\begin{aligned}
\frac{\abss{\rho(x)-\rho(y)}}{\dn(x,y)}&=\frac{1}{\dn(x,y)}
\frac{\abss{\dn(y,S_2)\brackets{\dn(x,S_1)-\dn(y,S_1)}+\dn(y,S_1)\brackets{\dn(y,S_2)-\dn(x,S_2)}}}
{\brackets{\dn(x,S_2)+\dn(x,S_1)}\brackets{\dn(y,S_2)+\dn(y,S_1)}}\\
&\overset{(i)}\leq\frac{1}{\dn(x,y)}\cdot
\frac{{\dn(y,S_2)\dn(x,y)+\dn(y,S_1)\dn(x,y)}}
{\brackets{\dn(x,S_2)+\dn(x,S_1)}\brackets{\dn(y,S_2)+\dn(y,S_1)}}
\overset{(ii)}\leq \frac{1}{{\dn(S_1,S_2)}}.
\end{aligned}
\end{equation*}
Inequality $(i)$ follows from the fact that $\abss{d(x,S)-d(y,S)}\leq d(x,y)$ for any set $S$, which is a consequence of the triangle inequality
    $\inf_{z\in S} d(x,z)\leq \inf_{z\in S} \brackets{d(x,y)+d(y,z)}$.
Inequality $(ii)$ follows from $\dn(S_1,S_2)\leq \dn(x,S_1)+\dn(x,S_2)$.
As a result, the modulus of the gradient of $\rho$ is bounded uniformly as follows
\begin{align}
\label{eq:modulus_gradient_bound}
\abss{\nabla \rho(x)}\defn \underset{d(x,y)\to 0^+}{\lim\sup}\frac{\abss{\rho(x)-\rho(y)}}{d(x,y)}\leq \frac{1}{\dn(S_1,S_2)}, \forall x \in K.  
\end{align}
Following an argument in the proof of Theorem 2.6 in~\cite{lovasz1993random}, we may assume that $S_1$ and $S_2$ are open without loss of generality. This avoids the necessity to deal to boundaries of $S_1$ and $S_2$. Applying the co-area formula in Lemma \ref{lem_coarea} over $(K,\dn,\Pi)$, we obtain:
\begin{equation*}
\begin{aligned}
\Pi(S_3)&\overset{(i)}\geq \dn(S_1,S_2)\int_{S_3}\abss{\nabla \rho(x)}\Pi(dx) \\
&\overset{(ii)}=\dn(S_1,S_2)\int_K\abss{\nabla\rho(x)}\Pi(dx)\\
&\overset{(iii)}\geq \dn(S_1,S_2)\int_{0}^{1}\Pi^+\braces{x\in X|\rho(x)>t}dt\\
&\overset{(iv)}\geq \dn(S_1,S_2)\int_{0}^1 \frac{1}{6\max\braces{1,\psi_n}}\Pi\braces{\rho(x)>t}\cdot \Pi\braces{\rho(x)\leq t}dt\\
&\overset{(v)}\geq \frac{\dn(S_1,S_2)}{6\max\braces{1,\psi_n}}\Pi(S_1)\Pi(S_2).
\end{aligned}
\end{equation*}
Inequality $(i)$ follows from Eq.~\eqref{eq:modulus_gradient_bound}. Equality $(ii)$ follows from $\abss{\nabla \rho (x)}=0$ on open $S_1,S_2$. Inequality $(iii)$ applies co-area formula and the fact that $\Pi^+\parenth{B}$ is $0$ for $B=\varnothing,K$. Inequality $(iv)$ applies the isoperimetric inequality with boundary measure in Lemma \ref{lem_isoperimetric_boundary}. Inequality $(v)$ holds since for any $t\in (0,1)$, we have
\begin{equation*}
    S_1\subseteq \braces{x\in K|\rho(x)\leq t} \quad\text{and}\quad S_2\subseteq \braces{x\in K|\rho(x)>t}.
\end{equation*}
\end{proof}


\section{Proofs of Main Results}\label{sec_conductance_proofs}
In this section, we prove our upper bounds of mixing times of \softdw. In Section \ref{sec_mixing_strongly}, we focus on the \softdw applied to $\alpha$-strongly logconcave and $\beta$-log-smooth distributions truncated on $K$, and we prove Theorem \ref{th_main} with its corollaries \ref{cor_logarithmic}, \ref{cor_m<n} and \ref{cor_Lewis}. In Section \ref{sec_mixing_weakly}, we remove the assumption of $\alpha$-strong logconcaveness, and prove Theorem \ref{th_weakly} with its corollaries \ref{cor_logarithmic_weakly} and \ref{cor_Lewis_weakly}. Finally, in Section \ref{sec_prob_ball_intersection}, we go beyond worst-case analysis and prove Theorem \ref{th_prob_ball_intersection}.

\subsection{Mixing for Strongly Logconcave Distributions}\label{sec_mixing_strongly}

We first prove Theorem \ref{th_main}, where we lower bound the conductance and apply Lemma \ref{lem_Lovas} to obtain the desired mixing time upper bound. In the proof we use the transition overlap (Lemma \ref{lem_TV_control}) and apply our combined isoperimetry on $\alpha$-strongly logconcave distributions. 

\begin{proof}[proof of Theorem \ref{th_main}]
Set the step-size $r$ indicated in Lemma \ref{lem_TV_control}, we try to bound the conductance of the Markov chain: For any partition $K=A_1\sqcup A_2$, we prove the following inequality
\begin{equation}\label{eq_final_partition}
    \int_{A_1}\cT_u(A_2)\pi(u)du\geq \frac{C}{\sqrt{n\parenth{\crossbd+\frac{\euclidbd}{\alpha}}}}\min\braces{\Pi(A_1),\Pi(A_2)},
\end{equation}
for some absolute constant $C$. To prove Eq.~\eqref{eq_final_partition}. We define two bad conductance subsets of $A_1,A_2$ to be:
\begin{equation*}
A_1'\defn\braces{u\in  A_1\bigg|\cT_u(A_2)<\frac{1}{10}},\quad A_2'\defn \braces{u\in  A_2\bigg|\cT_u(A_1)<\frac{1}{10}}.
\end{equation*}
We divide the proof into two scenarios. If $\Pi(A_1')\leq \frac{1}{2}\Pi(A_1)$ or $\Pi(A_2')\leq \frac{1}{2}\Pi(A_2)$. Since $\Pi$ is the stationary distribution of $\cT$, we have $\int_{A_1}\cT_x(A_2)\pi(x)dx=\int_{A_2}\cT_x(A_1)\pi(x)dx$. So we may assume $\Pi(A_1')\leq \frac{1}{2}\Pi(A_1)$ without loss of generality, and we have
\begin{equation*}
\int_{A_1}\cT_x(A_2)\pi(x)dx\geq \int_{A_1\setminus A_1'} \cT_x(A_2)\pi(x)dx
\geq \frac{1}{10}\Pi(A_1\setminus A_1')\geq \frac{1}{20}\min\braces{\Pi(A_1),\Pi(A_2)}.
\end{equation*}
Thus Eq.~\eqref{eq_final_partition} is verified in this scenario. 

Now we deal with the case $\Pi(A_1')>\frac12\Pi(A_1)$ and $\Pi(A_2')>\frac12\Pi(A_2)$. For any $u\in A_1'$, $v\in A_2'$, by definition of $A_1'$ and $A_2'$, we have
\begin{equation*}
\begin{aligned}
    \vecnorm{\cT_u-\cT_v}{TV}&\geq \cT_u(A_1)-\cT_v(A_1)=1-\cT_u(A_2)-\cT_v(A_1)\\
    &\geq 1-\frac{1}{10}-\frac{1}{10}>\frac45.
\end{aligned}
\end{equation*}
Thus, by Lemma \ref{lem_TV_control} and the symmetry between $x,y$, we have
\begin{equation*}
\begin{aligned}
    &\vecnorm{x-y}{G(x)}>\frac{r}{10\sqrt{n}},\quad \vecnorm{x-y}{G(y)}>\frac{r}{10\sqrt{n}}\\
   \Rightarrow&\frac{r^2}{100n}< \min\braces{\vecnorm{x-y}{G(x)}^2,\vecnorm{x-y}{G(y)}^2}\\
   &\frac{r^2}{100n}\overset{(i)}<C_K d_K(x,y)^2+C_E\vecnorm{x-y}2^2\overset{(ii)}\leq 
   \parenth{C_K+\frac{C_E}{(\log2)^2\alpha}}d'(x,y)^2,
\end{aligned}
\end{equation*}
where inequality $(i)$ holds due to our assumption on $G$, and inequality $(ii)$ holds since we define the new metric $d'(u,v)\defn \max\braces{d_K(u,v),\log(2)\sqrt{\alpha}\|u-v\|_2}$. We also have:
\begin{equation}\label{eq_log_final_ergodic}
\begin{aligned}
    \int_{A_1}\cT_u(A_2)\pi(u)du&=\frac12 \brackets{\int_{A_1}\cT_u(A_2)\pi(u)du+\int_{A_2}\cT_u(A_1)\pi(u)du}\\
    &\geq \frac12\brackets{\int_{A_1\setminus A_1'}\cT_u(A_2)\pi(u)du+\int_{A_2\setminus A_2'}\cT_u(A_1)\pi(u)du}\\
    &\geq \frac{1}{20}\Pi(K\setminus (A_1'\cup A_2')).
\end{aligned}
\end{equation}

Since we assumed that $f(x)$ is $\alpha$-strongly convex, $f(x)-\frac{\alpha}{2}\vecnorm{x-x^\star}2^2$ is a convex function where $x^\star$ is any point in $K$. As a result, for $x\in K$, the probability density function of $\Pi$ can be written as: $\pi(x)\propto \exp\parenth{-\frac{\alpha}{2}\vecnorm{x-x^\star}2^2}\phi(x)$, where $\phi(x)=\exp\parenth{-f(x)+\frac{\alpha}{2}\vecnorm{x-x^\star}2^2}$ is logconcave. Thus we can apply the isoperimetric inequality in Lemma \ref{lem_isoperimetric}, and we have 
\begin{equation}\label{eq_log_final_iso}
\Pi(K\setminus (A_1'\cap A_2'))\geq d'(A_1', A_2') \Pi(A_1')\Pi(A_2').
\end{equation}    
Now insert Eq.~\eqref{eq_log_final_iso} into Eq.~\eqref{eq_log_final_ergodic}, applying the lower bound of $d'(A_1',A_2')$, we deduce that: 
\begin{equation*}
\begin{aligned}
    \int_{A_1}\cT_u(A_2)\pi(u)du&\geq \frac{1}{20}d'(A_1', A_2') \Pi(A_1')\Pi(A_2')\\
    &\overset{(i)}\geq \frac{r}{200\times \sqrt{3n(\crossbd+\frac{\euclidbd}{\alpha})}}\cdot \frac12\Pi(A_1)\cdot\frac12\Pi(A_2)\\
    &\overset{(ii)}\geq \frac{r}{800\times\sqrt{3n(\crossbd+\frac{\euclidbd}{\alpha})}}\cdot\frac12 \min\braces{\Pi(A_1),\Pi(A_2)},
\end{aligned}
\end{equation*}
where inequality $(i)$ due to the assumption $\Pi(A_i')\geq \frac12 \Pi(A_i)$ for $i=1,2$. Inequality $(ii)$ holds because $\Pi(A_1)+\Pi(A_2)=1$. 

Combining the two situations, we proved that the conductance for the Markov chain $\TT$ satisfies $\Phi\geq c \sqrt{\crossbd+\frac{\euclidbd}{\alpha}}$ for some absolute constant $c>0$ since the step-size $r$ is a fixed constant. In order to use Lemma \ref{lem_Lovas}, we need to use the conductance $\barPhi$ for the Lazy version $\widebar\TT$, and this at most shrinks the conductance by a constant factor of $2$.

According to Lemma \ref{lem_Lovas}, there exists an absolute constant $C>0$ such that for all error tolerance $\epsilon>0$ and $M$-warm initial distribution $\mu_0$, the mixing time can be bounded as:
\begin{equation*}
T_{\mathrm{mix}} (\epsilon,\mu_0)\leq \frac{2}{\barPhi^2}\log\parenth{\frac{\sqrt{M}}\epsilon}
\leq \frac{8}{\Phi^2}\log\parenth{\frac{\sqrt{M}}{\epsilon}}\leq C\cdot   n\parenth{\crossbd+\frac{\euclidbd}{\alpha}}\log\parenth{\frac{\sqrt{M}}{\epsilon}}.
\end{equation*}
\end{proof}

Next, we derive its corollaries. For Corollary \ref{cor_logarithmic} and \ref{cor_m<n}, we notice that Corollary \ref{cor_m<n} is directly implied by Corollary \ref{cor_logarithmic} by letting $m=o(n)$. In order to prove Corollary \ref{cor_logarithmic}, we need to show that the \softlog $G$ in Definition \ref{def_logarithmic} setting the regularization size $\lambda\defn \beta$ satisfies all the conditions in Theorem \ref{th_main}. We divide these properties into the following two lemmas, and Corollary \ref{cor_logarithmic} is directly implied by the following two lemmas.

\begin{lemma}\label{lem_logarithmic_sc}
Given a (possibly unbounded) polytope $K=\braces{x|Ax>b}$ for $A\in\real^{m\times n}$ and $b\in\real^m$, the \softlog $G$ in Definition \ref{def_logarithmic} is SSC, ASC and LTSC.
\end{lemma}

\begin{lemma}\label{lem_logarithmic_metric_conversion}
Given a (possibly unbounded) polytope $K=\braces{x|Ax>b}$ for $A\in\real^{m\times n}$ and $b\in\real^m$, the \softlog $G$ in Definition \ref{def_logarithmic} satisfies: for all $x,y\in K$
\begin{equation}
\min\braces{\vecnorm{y-x}{G(x)},\vecnorm{y-x}{G(y)}}\leq m d_K(x,y)^2+\beta \vecnorm{x-y}{2}^2.
\end{equation}
Moreover, we have $E(x,G(x),1)\subseteq K$.
\end{lemma}

The proof of Lemma \ref{lem_logarithmic_sc} is left to Appendix \ref{appendix_logarithmic}.  If the polytope $K$ is bounded, these properties are already implied in \cite{pmlr-v247-kook24b} since the unregularized logarithmic metric $H(x)\defn A_x A_x\tp$ for $K$ is an invertible local metric. \cite{pmlr-v247-kook24b} proved that $H(x)$ is SSC, SASC and SLTSC, thus $G(x)\defn H(x)+\lambda I$ is SSC, ASC, LTSC by the additivity of these properties. However, we emphasize that in this paper, we do not require the polytope $K$ to be bounded, and we even allow $m<n$, thus $H(x)= A_xA_x\tp$ may not be invertible, and those self-concordance properties are not well-defined for non-invertible matrices, thus we can not use the additivity to prove those self-concordance properties for $G(x)=H(x)+\lambda I$.

To be rigorous, we provide the proof of SSC, LTSC, and ASC for the \softlog $G$ defined in Definition \ref{def_logarithmic} in Appendix \ref{appendix_logarithmic}. The SSC and LTSC of \softlog is proved by a limit argument: we add artificial constraints so that $H(x)\defn A_xA_x\tp$ is invertible and let these constraints vanish. We proved ASC of \softlog by concentration of Gaussian polynomials, which also appears in \cite{sachdeva2016mixing}, and the intuition is that adding a regularization term $\lambda I$ only makes the Gaussian concentration tighter. 

Now we prove Lemma \ref{lem_logarithmic_metric_conversion}, we introduce  the concept of $\bar\nu$-symmetry from \cite{laddha2020strong} for convenience. The reason we need Definition \ref{def_bar_nu} is that we can upper bound the $\bar\nu$-symmetric  local metric $H$  by $\sqrt{\bar\nu}d_K$, and this important property is summarized as Fact~\ref{fact_cross_ratio_by_local}.

\begin{definition}[from \cite{laddha2020strong}]\label{def_bar_nu}
For any convex set $K\subseteq \real^n$, we define a PSD matrix function $H:K\to \mathbb{S}_{+}^n$, and let $E(x,H(x),r)$ denote the following ellipsoid
\begin{equation*}
    E(x,H(x),r)\defn \{z|(z-x)\tp H(x) (z-x)\leq r^2\}. 
\end{equation*}
We define $H$ to be $\bar\nu$-symmetric if for any $x\in K$, we have 
\begin{equation*}
    E(x,H(x),1)\subseteq K\cap (2x-K)\subseteq E(x,H(x),\sqrt{\bar\nu})
\end{equation*}
\end{definition}

\begin{fact}[revised from Lemma 2.3 in \cite{laddha2020strong}]\label{fact_cross_ratio_by_local}
We use $d_K(x,y)$ to denote the extended cross-ratio distance (as defined in Definition \ref{def_cross_ratio_unbounded}) between $x$ and $y$ in the convex set $K$, assume $H$ is a $\bar\nu-$symmetric local metric, then we have 
\begin{equation}\label{eq_cross_ratio_by_local}
    d_K(x,y)\geq \frac1{\sqrt{\bar\nu}} \min\braces{\vecnorm{x-y}{H(x)},\vecnorm{x-y}{H(y)}}.
\end{equation}
\end{fact}

The proof of Fact~\ref{fact_cross_ratio_by_local} is postponed to Appendix \ref{sec_cross_ratio_by_local}, and it differs from the proof in \cite{laddha2020strong} in two important aspects: First, since we defined the cross-ratio distance in the extended sense, we no longer require $K$ to be a compact convex set, and we divide the proof into several cases. Second, we notice some computational problems in the symmetry argument of \cite{laddha2020strong} and cannot recover the inequality with only $\vecnorm{\cdot}{H(x)}$ on the RHS of Eq.~\eqref{eq_cross_ratio_by_local}, so we provide an weaker inequality with minimum of $\vecnorm{\cdot}{H}$ computed at both $x$ and $y$. For the sake of rigor and completeness, we include the proof in Appendix \ref{sec_cross_ratio_by_local}.

Now we are ready to give the proof of Lemma \ref{lem_logarithmic_metric_conversion} using the $\bar\nu$-symmetry and Fact \ref{fact_cross_ratio_by_local}. 

\begin{proof}[proof of Lemma \ref{lem_logarithmic_metric_conversion}]
We first prove that the unregularized logarithmic metric  $H(x)\defn A_xA_x\tp$ is $\bar\nu$-symmetric with $\bar\nu=m$. 
For all $z$ such that $z\in K\cap (2x-K)$, this is equivalent to for all $i\in[m]$, we have
\begin{equation*}
    a_i\tp z-b_i>0,\quad a_i\tp (2x-z)-b_i>0,
\end{equation*}
and this translates to $\underset{i\in[m]}\max\abss{s_{x,i}^{-1}\parenth{a_i\tp (z-x)}}\leq 1$. Then we have the desired bound:
\begin{equation*}
    \vecnorm{z-x}{H(x)}^2 \leq m\vecnorm{S_x^{-1}A(z-x)}{\infty}^2=m\cdot\underset{i\in[m]}\max\abss{s_{x,i}^{-1}\parenth{a_i\tp (z-x)}}\leq m,
\end{equation*}
On the other hand, if we assume $z\in E(x,H(x),1)$, then we have 
\begin{equation*}
  \underset{i\in[m]}\max\abss{s_{x,i}^{-1}\parenth{a_i\tp (z-x)}}=\vecnorm{S_x^{-1}A(z-x)}{\infty}\leq \vecnorm{S_x^{-1}A(z-x)}{2}\leq 1.
\end{equation*}
Thus, for all $i\in [m]$, $a_i\tp z-b_i>0$ and $a_i\tp (2x-z)-b_i>0$. 
We have $z\in K\cap (2x-K)$. In conclusion, $H$ is $\bar\nu$-symmetric with $\bar\nu=m$. As a result of Fact \ref{fact_cross_ratio_by_local}, for any $x,y\in K$, we have
\begin{equation*}
    \min\braces{\vecnorm{y-x}{H(x)}^2,\vecnorm{y-x}{H(y)}^2}\leq m d_K(x,y)^2.
\end{equation*}
By definition of $G$ and setting $\lambda\defn \beta$, we have 
\begin{equation*}
\begin{aligned}
\min\braces{\vecnorm{y-x}{G(x)}^2,\vecnorm{y-x}{G(y)}^2}&= \min\braces{\vecnorm{y-x}{H(x)}^2,\vecnorm{y-x}{H(y)}^2}+\beta \vecnorm{x-y}{2}^2\\
&\leq m\cdot d_K(x,y)^2+\beta \vecnorm{x-y}{2}^2.
\end{aligned}
\end{equation*}
Moreover, since we proved $H$ is $\bar\nu$-symmetric, so we have 
\begin{equation*}
    E(x,G(x),1)\subseteq E(x,H(x),1)\subseteq K\cap (2x-K)\subseteq K.
\end{equation*}
\end{proof}

The remaining task in this section is to prove Corollary \ref{cor_Lewis}. Similarly, we need to verify the conditions in Theorem \ref{th_main} for \reglewis $G$ in Definition \ref{def_Lewis}. These properties are nicely implied in \cite{pmlr-v247-kook24b}. Since the unregularized Lewis metric $H(x)\defn c_1\sqrt{n}(\log m)^{c_2}A_x\tp W_x A_x$ with Lewis weights $w_x$ defined in Eq.~\eqref{eq_Lewis_weights_def}. We list the properties here in Fact \ref{fact_unreg_Lewis_properties}. 

\begin{fact}[Lemma E.5, E.7 and E.12 \cite{pmlr-v247-kook24b}]\label{fact_unreg_Lewis_properties}
There exists positive constants $c_1$ and $c_2$ such that the unregularized Lewis-weights metric $H(x)\defn\sqrt{n} c_1(\log m)^{c_2} A_x\tp W_x A_x$ is SSC, SLTSC, SASC and $\bar\nu$-symmetric with $\bar\nu=O((\log m)^{c_2}n^{3/2})$.
\end{fact}

\begin{fact}[Lemma D.12 and D.14 \cite{pmlr-v247-kook24b}]\label{fact_addition}
Given PSD matrix functions $G_i$ on $K_i$ for $i=1,\ldots,l$, let $G\defn \sum_{i} G_i$ be PD on $\cap_i K_i$. We have
\begin{itemize}
    \item if $G_i$ is SLTSC on $K_i$, then $G$ is LTSC on $\cap_i K_i$;   
    \item if $l=O(1)$ and $g_i$ is SASC on $K_i$, then $G$ is ASC on $\cap_i K_i$.
\end{itemize}  
\end{fact}

With Fact \ref{fact_unreg_Lewis_properties} and \ref{fact_addition}, we provide the proof of Corollary \ref{cor_Lewis} by verifying all the conditions in Theorem \ref{th_main} for \reglewis $G$ with regularization size $\lambda\defn \beta$ in Definition \ref{def_Lewis}. 
\begin{proof}[proof of Corollary \ref{cor_Lewis}]
Define $H(x)\defn c_1\sqrt{n}(\log m)^{c_2}A_x\tp W_x A_x$ as in Fact \ref{fact_unreg_Lewis_properties}, then the \reglewis 
$G$ satisfies $G(x)\defn H(x)+\beta I$. 
We first show that  $G(x)\defn H(x)+\beta I$ is SSC, where $H(x)$ is  the unregularized Lewis metric. For any direction $h\in\real^n$, we have
\begin{equation*}
\begin{aligned}
&\vecnorm{G(x)^{-\frac12}\deri G(x)[h]G(x)^{-\frac12}}{F}= \vecnorm{G(x)^{-\frac12}\deri H(x)[h]G(x)^{-\frac12}}{F}\\
\overset{(i)}\leq& \vecnorm{H(x)^{-\frac12}\deri H(x)[h]H(x)^{-\frac12}}{F}\overset{(ii)}\leq 2\vecnorm{h}{H(x)}\leq 2\vecnorm{h}{G(x)},
\end{aligned}
\end{equation*}
where inequality $(i)$ holds since $G(x)\succeq H(x)$ (see more details in Lemma \ref{lem_matrix_inequality}), and inequality $(ii)$ follows from Fact \ref{fact_unreg_Lewis_properties} that $H$ is SSC. Moreover, since $\beta I$ is a constant matrix function defined over $\real^n$, thus $\beta I$ is trivially SASC and SLTSC by definition. By the additivity in Fact \ref{fact_addition} and $H$ is SASC and SLTSC in Fact \ref{fact_unreg_Lewis_properties}, thus  $G(x)\defn H(x)+\beta I$ is ASC and LTSC.  

Next we prove the metric inequality, it is obvious that $G(x)\succeq \beta I$. Since $H$ is $\bar\nu$-symmetric with $\bar\nu=\Ot(n^{3/2})$. Thus, by Fact \ref{fact_cross_ratio_by_local}, we have the following metric inequality
\begin{equation*}
\begin{aligned}
\min\braces{\vecnorm{y-x}{G(x)}^2,\vecnorm{y-x}{G(y)}^2}&=\min\braces{\vecnorm{y-x}{H(x)}^2,\vecnorm{y-x}{H(y)}^2}\beta \vecnorm{y-x}{2}^2 \\
&\leq  C\cdot (\log m)^{c_2} \cdot n^{3/2} d_K(x,y)^2 +\beta \vecnorm{y-x}{2}^2, 
\end{aligned}
\end{equation*}
where $C>0$ is some absolute constant. Moreover, since $H$ is $\bar\nu$-symmetric and $G(x)\succeq H(x)$, we have $E(x,G(x),1)\subseteq E(x,H(x),1)\subseteq K$.  Thus Corollary \ref{cor_Lewis} is proved by inserting $\euclidbd=\beta$ and $\crossbd=(\log m)^{c_2}n^{3/2}$.

\end{proof}

\subsection{Extension to Weakly Logconcave Distributions}\label{sec_mixing_weakly}
In this section, we prove Theorem \ref{th_weakly} and its corresponding Corollaries \ref{cor_logarithmic_weakly} and \ref{cor_Lewis_weakly}. Actually, the only task in Section \ref{sec_mixing_weakly} is to prove Theorem \ref{th_weakly}. This is because  the required properties of the local metric $G$ in Theorem \ref{th_weakly} are the same as in Theorem~\ref{th_main}. Those properties of both \softlog (Definition~\ref{def_logarithmic}) and \reglewis (Definition \ref{def_Lewis}) are already proved in Section \ref{sec_mixing_strongly}. Thus Corollaries \ref{cor_logarithmic_weakly} and \ref{cor_Lewis_weakly} can be derived directly from Theorem~\ref{th_weakly}. The proof follows similar conductance arguments from the proof of Theorem~\ref{th_main}  in Section \ref{sec_mixing_strongly}, except that we use the isoperimetric inequality over weakly logconcave measures. 

\begin{proof}[proof of Theorem \ref{th_weakly}]
Set the step-size $r$ indicated in Lemma \ref{lem_TV_control}, we try to bound the conductance of the Markov chain: For any partition $K=A_1\sqcup A_2$, we will prove the following inequality:
\begin{equation}\label{eq_final_partition_weakly}
    \int_{A_1}\cT_u(A_2)\pi(u)du\geq \frac{C}{\sqrt{n\parenth{\crossbd+\spectr{\euclidbd}}}}\min\braces{\Pi(A_1),\Pi(A_2)},
\end{equation}
where $C$ is some absolute constant. Let $K=A_1\sqcup A_2$ be any measurable partition, and we define two bad conductance subsets of $A_1,A_2$ to be:
\begin{equation*}
A_1'\defn \braces{x\in A_1\bigg|\cT_x(A_2)<\frac{1}{10}},\quad A_2'\defn \braces{x\in A_2\bigg|\cT_x(A_1)<\frac{1}{10}}.
\end{equation*}

We can assume $\Pi(A_1')>\frac12\Pi(A_1)$ and $\Pi(A_2')>\frac12\Pi(A_2)$,  otherwise Eq.~\eqref{eq_final_partition_weakly} can be proved the same way as in the proof of Theorem \ref{th_main}. Then for any $x\in A_1'$ and $y\in A_2'$, by definition of $A_1',A_2'$ we have
\begin{equation*}
\vecnorm{\cT_x-\cT_y}{TV}\geq \cT_x(A_1)-\cT_y(A_1)=1-\cT_x(A_2)-\cT_y(A_1)\geq 1-\frac{1}{10}-\frac{1}{10}>\frac45.
\end{equation*}
So by Lemma \ref{lem_TV_control} and the symmetry between $x$ and $y$ we have 
\begin{equation*}
    \vecnorm{y-x}{G(x)}>\frac{r}{10\sqrt{n}},\quad \vecnorm{y-x}{G(y)}>\frac{r}{10\sqrt{n}},
\end{equation*}
thus according to our assumption of $G$, we have 
\begin{equation}\label{eq_weakly_cross_ratio}
\begin{aligned}
\frac{r^2}{100n}&< \crossbd d_K(x,y)^2+\euclidbd \vecnorm{y-x}{2}^2\\
&=\crossbd\brackets{e^{\dH_K(x,y)}-1}^2+\euclidbd\vecnorm{x-y}{2}^2\\
&\leq \crossbd \brackets{e^{\dn(x,y)}-1}^2+\euclidbd \frac{{\spectr}}{(\log2)^2}\dn(x,y)^2.
\end{aligned}
\end{equation}
Recall that we use the Hilbert metric $\dH_K$ instead of the cross-ratio distance $d_K$ to apply the mixed isoperimetric inequality for weakly logconcave distributions (Lemma \ref{lem_isoperimetric_weakly}). The last inequality holds since we define $\dn(x,y)=\max\braces{\frac{\log2}{\sqrt{\spectr}}\vecnorm{y-x}{2},\dH_K(x,y)}$. Consider the case that $\dn(x,y)\leq 1$, then according to the fact $e^x\leq 1+2x$ for $x\in[0,1]$, we have 
\begin{equation}\label{eq_dn_lower_bound}
\begin{aligned}
\frac{r^2}{100n}< 4\crossbd \dn(x,y)^2+\euclidbd \frac{{\spectr}}{(\log2)^2}\dn(x,y)^2.
\end{aligned}
\end{equation}
As a result, for all $x\in A_1'$ and $y\in A_2'$ we have 
\begin{equation*}
\dn (x,y)\geq \min\braces{1,\frac{r}{20\sqrt{n\parenth{\crossbd+\euclidbd\spectr}}}}\geq \frac{r}{20\sqrt{n\parenth{\crossbd+\euclidbd\spectr}}},
\end{equation*}
where the last inequality holds since we assumed $C_K\geq 1$, $n\geq1$ and we can always set the step-size $r\leq 1$. 
Now we are ready to control the ergodic flow from $A_1$ to $A_2$, using the same argument in the proof of Theorem \ref{th_main} (Eq.~\eqref{eq_log_final_ergodic}) we have 
\begin{equation}\label{eq_ergodic_flow_weakly}
\begin{aligned}
\int_{A_1}\cT_x(A_2)\pi(x)dx&\geq \frac{1}{20}\Pi(K\setminus (A_1'\cup A_2')). 
\end{aligned}
\end{equation}
Now we apply Lemma \ref{lem_isoperimetric_weakly} over the target distribution $\Pi$, and we have 
\begin{equation}\label{eq_dn_isoperi}
\Pi(K\setminus (A_1'\cup A_2'))\geq \frac{1}{6\max\braces{1,\psi_n}}\dn(A_1',A_2')\Pi(A_1')\Pi(A_2')
\end{equation}
Now insert Eq.~\eqref{eq_dn_isoperi} into Eq.~\eqref{eq_ergodic_flow_weakly}, and we have 
\begin{equation*}
\begin{aligned}
\int_{A_1}\cT_x(A_2)\pi(x)dx&\overset{(i)}\geq\frac{\Pi(A_1')\Pi(A_2')}{120\max\braces{1,\psi_n}}\cdot\frac{r}{20\sqrt{n\parenth{\crossbd+\euclidbd\spectr}}}\\
&\overset{(ii)}\geq \frac{\min\braces{\Pi(A_1),\Pi(A_2)}}{8\cdot 120\max\braces{1,\psi_n}}\cdot\frac{r}{20\sqrt{n\parenth{\crossbd+\euclidbd\spectr}}},
\end{aligned}
\end{equation*}
where inequality $(i)$ follows by the lower bound for $\dn(A_1',A_2')$ in Equation \eqref{eq_dn_lower_bound}, and inequality $(ii)$ holds because we assumed $\Pi(A_i')\geq \frac12\Pi(A_i)$ for $i=1,2$ and  $\Pi(A_1)+\Pi(A_2)=1$.

Combining the situations, we have proved an lower bound on the conductance $\Phi$ of the transition kernel $\cT$. There exists a universal constant $c$, and since the step-size $r$ is also an universal constant, $n\geq 1$ and $C_K\geq 1$ , we have
\begin{equation}
\Phi \geq  \frac{c}{\max\braces{1,\psi_n}}\frac{1}{\sqrt{n\parenth{\crossbd+\euclidbd\spectr}}}\geq \frac{c}{\psi_n\sqrt{n(\crossbd+\euclidbd \spectr)}},
\end{equation}
where the last inequality holds since $\psi_n=\Omega(1)$. Since the conductance $\barPhi$ for the lazified transition kernel $\barT$ satisfies $\barPhi\geq \frac12\Phi$. According to Lemma \ref{lem_Lovas}, there exists a universal constant $C>0$, for any error tolerance $\epsilon>0$ and $M$-warm initial distribution $\mu_0$, the mixing time satisfies:
\begin{equation*}
    T_\text{mix}(\epsilon;\mu_0)\leq \frac{8}{\Phi^2}\log\parenth{\frac{\sqrt{M}}{\epsilon}}\leq  C \psi_n^2\cdot n\parenth{\crossbd+\euclidbd \spectr}\log\parenth{\frac{\sqrt{M}}{\epsilon}}.
\end{equation*}
\end{proof}

\subsection{Violated Constraints in a High Probability Ball}\label{sec_prob_ball_intersection}
In Section \ref{sec_prob_ball_intersection}, we prove Theorem \ref{th_prob_ball_intersection}. The idea is similar to prove the conductance lower bound for the Markov chain $\cT$ in Theorem \ref{th_main}. The difference is that we needed to cut off the corner of $K$ that is too close to the remaining $(m-\cM_{R}^\delta)$ linear constraints. Since this corner has a low probability mass, it does not impact the mixing time much. We need the concept of $s$-conductance $\Phi_s$(see Eq.~\eqref{eq_conductance_def}) to enable the cutoff.

We first give the statement and the proof of Lemma \ref{lem_s_conductance}, which controls the ergodic flows from $A_1$ to $A_2$. 

\begin{lemma}\label{lem_s_conductance}
Assume our target distribution $\Pi$ with density $\pi(x)\propto \mathbf{1}_K(x)e^{-f(x)}$ with twice differentiable $f$ to be both $\alpha$-convex and $\beta$-smooth as in Eq.~\eqref{eq_distri}. Let $\cT$ denotes the transition kernel in Eq.~\eqref{eq_Transition_kernel_def} defined using the \softlog in Definition \ref{def_logarithmic}. Then there exists a step-size $r>0$ such that for any partition $K=A_1\sqcup A_2$, for any $R>0$ and $\delta>0$, we have 
\begin{equation}\label{eq_s_conductance}
\int_{A_1}\cT_u(A_2)\pi(u)du\geq \frac{r}{32000\sqrt{\brackets{\kappa+\frac{m}{n\delta^2}+\cM_{R}^\delta}n}}\min\braces{\Pi(A_1\cap \cB_{R}),\Pi(A_2\cap\cB_{R})},
\end{equation}
where the definitions of  $\cB_R$, and $\cM_{R}^\delta$ can be found in Section \ref{sec_beyond_worst_result}. 
\end{lemma}

\begin{proof}
For any $R>0$ and $\delta>0$, recall the definition of $\cM_R^\delta$, we denote $[\cM_R^\delta]$ to be the following subset of $[m]\defn \braces{1,2,\dots,m}$, 
\begin{equation*}
[\cM_R^\delta]\defn \braces{i|\,\exists\, x\in \cB_R^{\delta}\text{ such that }a_i\tp x-b_i\leq 0}.
\end{equation*}

The distance of any point $x$ to the $i$-th face $\braces{z|a_i\tp z-b_i=0}$ equals $\frac{\abss{a_i\tp x-b_i}}{\vecnorm{a_i}2}$, thus for all $x\in\cB_R$,  we have 
\begin{equation*}
\frac{a_ia_i\tp }{(a_i\tp x-b_i)^2}\preceq \frac{\vecnorm{a_i}2^2}{(a_i\tp x-b_i)^2}I_n\overset{(\text{I})}\preceq\frac{\alpha}{n\delta^2}I_n \text{ for } i \in  [m]\setminus [\cM_R^\delta], 
\end{equation*}
where inequality (I) holds since for any $x\in \cB_R$, the distance of $x$ to any face  $i\in [m]\setminus [\cM_R^\delta]$ is greater than $\delta\sqrt{\frac{n}{\alpha}}$. Otherwise, there exists $z$ such that $\vecnorm{z-x}2\leq \delta\sqrt{\frac{n}{\alpha}}$ satisfying $a_i\tp z-b_i\leq  0$ for some $i\in [m]\setminus [\cM_R^\delta]$, which implies $z\notin \cB_R^{\delta}$. Meanwhile, we also  have $z\in \cB_{R}^\delta$ because $\vecnorm{z-x^\star}2\leq \vecnorm{x-x^\star}2+\vecnorm{x-z}2\leq (R+\delta)\sqrt{\frac{n}{\alpha}}$. This is a contradiction. 

Now we can control the \softlog $G(x)$ for all $x\in K\cap \cB_R$
\begin{equation*}
\begin{aligned}
     \Glog(x)&= \beta I +\sum_{i\in [\cM_R^\delta]}\frac{a_ia_i\tp }{(a_i\tp x-b_i)^2} +\sum_{i\in [m]\setminus [\cM_R^\delta]}\frac{a_ia_i\tp }{(a_i\tp x-b_i)^2}\\
     &\preceq \brackets{\beta+\frac{\alpha(m-\cM_R^\delta)}{n\delta^2}}I+\sum_{i\in[\cM_R^\delta]}\frac{a_ia_i\tp }{(a_i\tp x-b_i)^2}.
\end{aligned}
\end{equation*}
Then we can control the $\vecnorm{y-x}{\Glog(x)}$ for all $x,y\in \cB_R\cap K$: 
\begin{equation}\label{eq_s_local_by_new}
\begin{aligned}
\vecnorm{y-x}{\Glog(x)}^2&\overset{\text{(i)}}\leq \brackets{\beta+\frac{\alpha(m-\cMR)}{n\delta^2}}\vecnorm{x-y}2^2+\cMR d_K^2(x,y)\\
&\leq  \brackets{\cMR+\frac{m}{n\delta^2(\log2)^2}+\frac{\kappa}{(\log2)^2}}d'(x,y)^2\\
&\leq  4\brackets{\cMR+\frac{m}{n\delta^2}+{\kappa}}d'(x,y)^2,
\end{aligned}
\end{equation}
where inequality $(i)$ holds due to the $\cMR$-symmetry of the logarithmic metric $H(x)\defn \sum_{i\in [\cMR]}\frac{a_ia_i\tp}{(a_i\tp x-b_i)^2}$ (see the proof of Lemma \ref{lem_logarithmic_metric_conversion} in Section \ref{sec_mixing_strongly}). $d'(x,y)$ is the combination of Euclidean and cross-ratio distances defined in Lemma \ref{lem_isoperimetric}. 

As the routine for bounding conductance of a transition kernel, we consider the subsets of $A_1,A_2$ with bad conductance. We define:
\begin{equation*}
A_1'\defn \braces{u\in A_1\cap\cB_R\bigg|\cT_u(A_2)< \frac{1}{10}},\quad A_2'\defn \braces{u\in A_2\cap \cB_R\bigg| \cT_u(A_1)< \frac{1}{10}}.    
\end{equation*}
We first consider the case  $\Pi(A_i')\leq \frac12\Pi(A_i\cap\cB_R)$ for $i=1$ or $i=2$. Since $\Pi$ is the stationary distribution of the transition kernel $\cT$, we can assume $\Pi(A_1')\leq \frac12\Pi(A_1\cap\cB_R)$ without loss of generality. we have 
\begin{equation*}
\begin{aligned}
    \int_{A_1}\cT_u(A_2)\pi(u)du&\geq \int_{A_1\setminus A_1'}\cT_u(A_2)\pi(u)du\overset{(i)}\geq \frac{1}{10}\Pi(A_1\setminus A_1')
    \geq \frac{1}{20}\Pi(A_1\cap \cB_R)\\
    &\geq \frac{1}{20}\min\braces{\Pi(A_1\cap \cB_R),\Pi(A_2\cap  \cB_R)},
\end{aligned}
\end{equation*}
where inequality $(i)$ holds because we assumed $\Pi(A_1')\leq \frac12\Pi(A_1\cap \cB_R)$. So Eq.~\eqref{eq_s_conductance} is proved where we insert $r=10^{-5}$, $n\geq 1$ and $\kappa\geq 1$.

Next we deal with the case $\Pi(A_i')>\frac12 \Pi(A_i\cap \cB_R)$ for both $i=1$ and $i=2$. For any $x\in A_1'$ and $y\in A_2'$, we have 
\begin{equation*}
\begin{aligned}
\vecnorm{\cT_x-\cT_y}{TV}&\geq \cT_x(A_1)-\cT_y(A_1)\\
&=1-\cT_x(A_2)-\cT_y(A_1)> \frac45,
\end{aligned}
\end{equation*}
According to the transition overlap we proved in Lemma \ref{lem_TV_control}, and by symmetry between $x$ and $y$, we have 
\begin{equation*}
\vecnorm{y-x}{\Glog(x)}>\frac{r}{100\sqrt{n}} \text{ and } \vecnorm{y-x}{\Glog(y)}>\frac{r}{100\sqrt{n}}.
\end{equation*}
Since we assumed $x,y\in\cB_R\cap K$, Eq.~\eqref{eq_s_local_by_new} holds, thus we have 
\begin{equation}
\begin{aligned}
d'(x,y)&\geq \frac{\vecnorm{y-x}{\Glog(x)}}{2\sqrt{\cM_R^\delta+\frac{m}{n\delta^2}+\kappa}}\geq \frac{r}{200\sqrt{n}\sqrt{\cM_R^\delta+\frac{m}{n\delta^2}+\kappa}}.
\end{aligned}    
\end{equation}
Then the LHS of Eq.~\eqref{eq_s_conductance} can be bounded by:
\begin{equation*}
\begin{aligned}
\int_{A_1}\cT_u(A_2)\pi(u)du&=\frac12 \brackets{\int_{A_1}\cT_u(A_2)\pi(u)du+\int_{A_2}\cT_u(A_1)\pi(u)du}\\
&\geq \frac12 \brackets{\int_{A_1\setminus A_1'}\cT_u(A_2)\pi(u)du+\int_{A_2\setminus A_2'}\cT_u(A_1)\pi(u)du}\\
&\geq \frac{1}{20}\Pi(\cB_R\setminus (A_1'\cap A_2')).
\end{aligned}
\end{equation*}
Now we define $\Pi_R(\cdot)$ to be the probability measure $\Pi(\cdot)$ constrained on $\cB_R\cap K$. In other words, $\Pi_R(C)\defn \frac{\Pi(C)}{\Pi(\cB_R\cap K)}$ for all Borel sets $C\subseteq K\cap \cB_R$. And we define $d'_R(x,y)$ to be the mixed metric whose cross-ratio metric is (see definition in Lemma \ref{lem_isoperimetric}) restricted on $\cB_R\cap K$:
\begin{equation*}
\begin{aligned}
d'_R(x,y)\defn \max\braces{d_{K\cap \cB_R}(x,y),\log2\sqrt{\alpha}\vecnorm{x-y}2}.
\end{aligned}
\end{equation*}
where we changed the cross-ratio distance from over $K$ to over $K\cap \cB_R$. It is easy to verify that $d_{K\cap \cB_R}(x,y)\geq d_K(x,y)$ because $K\cap \cB_R\subseteq K $, so we also have: $d'_R(x,y)\geq d'(x,y)$. Countinuing to lower bound the LHS of Eq.~\eqref{eq_s_conductance}, we apply isoperimetric inequality (Lemma \ref{lem_isoperimetric}) under the metric $d'_R(x,y)$:
\begin{equation*}
\begin{aligned}
\int_{A_1}\cT_u(A_2)\pi(u)du&\geq  \frac{1}{20}\Pi(\cB_R\setminus (A_1'\cap A_2'))= \frac1{20}\Pi(\cB_R)\Pi_R(\cB_R\setminus (A_1'\cap A_2'))\\
&\overset{(i)}\geq \frac{1}{20}\Pi(\cB_R) d_R'(A_1',A_2')\Pi_R(A_1')\Pi_R(A_2')\\
&\overset{(ii)}\geq \frac{1}{20}\Pi(\cB_R) d'(A_1',A_2')\Pi_R(A_1')\Pi_R(A_2')\\
&\overset{(iii)}\geq \frac{1}{80 \Pi(\cB_R)}  d'(A_1',A_2') \Pi(A_1\cap \cB_R)\Pi(A_2\cap \cB_R)\\
&=\frac{1}{80} d'(A_1',A_2')\Pi(\cB_R)\Pi_R(A_1\cap \cB_R)\Pi_R(A_2\cap \cB_R)\\
&\overset{(iv)}\geq \frac{1}{160} d'(A_1',A_2')\Pi(\cB_R)\min\braces{\Pi_R(A_1\cap\cB_R),\Pi_R(A_2\cap \cB_R)}\\
&= \frac{1}{160} d'(A_1',A_2')\min\braces{\Pi(A_1\cap\cB_R),\Pi(A_2\cap \cB_R)},
\end{aligned}
\end{equation*}
where inequality $(i)$ is the application of our new isoperimetric inequality (Lemma \ref{lem_isoperimetric}) over the constrained metric $d_R'(x,y)$, inequality $(ii)$ is due to our argument that $d'(x,y)\leq d_R'(x,y)$, inequality $(iii)$ holds because of the assumptions $\Pi(A_i')>\frac12\Pi(A_i\cap \cB_R)$ for $i=1,2$, and inequality $(iv)$ is due to the fact that $\Pi_R(A_1\cap\cB_R)+\Pi_R(A_2\cap \cB_R)=1$.

Now insert the upper bound of $d'(x,y)$ in Eq.~\eqref{eq_s_local_by_new}, we have 
\begin{equation*}
\int_{A_1}\cT_u(A_2)\pi(u)du\geq  \frac{r}{32000\sqrt{n}\sqrt{\cMR+\frac{m}{n\delta^2}+\kappa}}\min\braces{\Pi(A_1\cap\cB_R),\Pi(A_2\cap\cB_R)}.
\end{equation*}
\end{proof}

Now we can give the proof of Theorem \ref{th_prob_ball_intersection}, where we apply Lemma \ref{lem_s_conductance} setting $R$ to be $\rinf\parenth{\frac{\epsilon}{2M}}$, and $\delta>0$ is still an arbitrary parameter that can be tuned by the user:

\begin{proof}[proof of Theorem \ref{th_prob_ball_intersection}]
We set $\Upsilon\defn \rinf\parenth{\frac{\epsilon}{2M}}$, then by definition of $\rinf$ (see Definition \ref{def_rinf}), we have  
\begin{equation*}
\Pi(\cB_\Upsilon)\geq 1-\frac{\epsilon}{2M}.
\end{equation*}
For any partition $K=A_1\sqcup A_2$, we have $\Pi(A_i\cap \cB_\Upsilon)\geq \Pi(A_i)-\Pi(\cB_\Upsilon)$ for $i=1,2$. Apply Lemma \ref{lem_s_conductance} where we set $R\defn \Upsilon$, we have the result as:
\begin{equation*}
\int_{A_1}\cT_u(A_2)\pi(u)du\geq  \frac{r}{32000\sqrt{n}\sqrt{\cM_\Upsilon^\delta+\frac{m}{n\delta^2}+\kappa}}\min\braces{\Pi(A_1)-\frac{\epsilon}{2M},\Pi(A_2)-\frac{\epsilon}{2M}}.
\end{equation*}
We set $s\defn \frac{\epsilon}{2M}$, then the $s$-conductance of our transition kernel $\cT$ before lazification satisfies:
\begin{equation}\label{eq_s_lower}
\Phi_{s}\geq \frac{r}{32000\sqrt{n}\sqrt{\cM_\Upsilon^\delta+\frac{m}{n\delta^2}+\kappa}}.
\end{equation}

After lazification, the $s$-conductance under our new transition kernel $\tcT_x(\cdot)\defn \frac12\delta_x(\cdot)+\frac12\cT_x(\cdot)$ satisfies: $\widetilde \Phi_s=\frac12\Phi_{s}$. Using the mixing time bound for lazy Markov chains in Eq.~\eqref{eq_s_Lovas}, we obtain
\begin{equation*}
\vecnorm{\tcT^k(\mu_0)-\Pi}{TV}\leq Ms+M\parenth{1-\frac{\widetilde\Phi_s^2}{2}}^k\leq \frac{\epsilon}{2}+M\exp\parenth{-\frac{k\widetilde\Phi_s^2}{2}}.
\end{equation*}

To ensure $\vecnorm{\tcT^k(\mu_0)-\Pi}{TV}\leq \epsilon$, we only need to ensure $M\exp\parenth{-\frac{k\widetilde \Phi_s^2}{2}}\leq \frac\epsilon2$. That is, $k\geq \frac{2}{\widetilde \Phi_s^2}\log\parenth{\frac{2M}{\epsilon}}$. Insert the lower bound of $\Phi_s$ as in Eq.~\eqref{eq_s_lower} and the fact that $\widetilde \Phi_s=\frac12\Phi_s$, we only need to ensure:
\begin{equation}\label{eq_mixing_high_prob}
    k\geq \frac{8\times 32000^2}{r^2} \log\parenth{\frac{2M}{\epsilon}}\parenth{\kappa+\frac{m}{n\delta^2}+\cM_\Upsilon^\delta}n,
\end{equation}
where the step-size $r$ is also an absolute constant, and the desired mixing time is proved. To prove that Eq.~\eqref{eq_mixing_high_prob} still suffices for $\vecnorm{\tcT^k(\mu_0)-\Pi}{TV}\leq \epsilon$ to hold  if we replace $\Upsilon$ with $\widehat\Upsilon \defn \rinfbd\parenth{\frac{\epsilon}{2M}}$, we only need to prove $\cM_\Upsilon^\delta\leq\cM_{\widehat\Upsilon}^\delta$. This is easy since $\cM_R^{\delta}$ is clearly an increasing function in $R$, and we have $\Upsilon\leq \widehat\Upsilon$ (see the concentration inequality in Eq.~\eqref{eq_rinf(s)_upper_bound}).  
\end{proof}

\section{Discussions}\label{sec_disc}

This work opens several avenues for future extensions. In terms of warm initial distributions, although we proposed the uniform distribution on a certain ball in Appendix~\ref{sec_disc_warm_start}, it still has one major limitation. The warmness parameter $M$ scales exponentially with the dimension $n$, thus introducing an extra $n$-factor in the mixing times. To mitigate this, Gaussian cooling~\cite{cousins_gaussian_2018} emerges as one promising direction. This method constructs a sequence of distributions to sample from, each serving as a warm start for the next. For truncated sampling, \cite{pmlr-v247-kook24b} explores a Gaussian cooling procedure under a different framework than ours. Adapting Gaussian cooling techniques to our framework could significantly reduce the warmness dependency.

Moreover, Theorem \ref{th_main} and \ref{th_weakly} may be generalized to sample logconcave distributions truncated on convex bodies defined by nonlinear constraints. For instance, in the case of quadratic constraints,  \cite{pmlr-v247-kook24b} constructs the barrier function, whose Hessian satisfies various self-concordance properties in Definition \ref{def_sc}. By leveraging this structure, our results may yield mixing time bounds for sampling log-concave distributions truncated on ellipsoids.

Finally, our regularized local metric could inspire new mixing time analyses for higher order algorithms such as \geodesicwalk and \RHMC. 

\section*{Acknowledgement}
Both authors are partially supported by NSF CAREER Award DMS-2237322, Sloan Research Fellowship and Ralph E. Powe Junior Faculty Enhancement Awards. 

\newpage
\appendix

\section{Bounding the Transition Overlap}
In this section, we prove the upper bound of the TV-distance between transitions (Lemma \ref{lem_TV_control}). In Appendix \ref{sec_ssc_col} we prove the properties for strongly self-concordant local metrics. The proof is similar to \cite{laddha2020strong}, while we can only recover the computation in \cite{laddha2020strong} up to constant $2$. Though this does not affect the usefulness, we list our proof to be rigorous. In Appendix \ref{appendix_transit_overlap}, we prove Lemma \ref{lem_TV_control}.

\subsection{Strong Self-Concordance Properties}\label{sec_ssc_col}
In Appendix \ref{sec_ssc_col} we prove the properties for strongly self-concordant local metrics (Fact~\ref{fact_ssc_col}).

\begin{fact}[adapted from \cite{laddha2020strong}]\label{fact_ssc_col}
Assume $G$ is SSC over $K$, then for any $x,y\in K$ with $\vecnorm{x-y}{G(x)}<1$, we have
\begin{equation}\label{eq_ssc_frob}
    \vecnorm{G(x)^{-\frac{1}{2}}(G(y)-G(x))G(x)^{-\frac12}}F\leq\frac{2\vecnorm{x-y}{G(x)}}{(1-\vecnorm{x-y}{G(x)})^2},
\end{equation}
for any $\vecnorm{x-y}{G(x)}\leq \frac12$, we have
\begin{equation}\label{eq_ssc_det_col}
    \det\brackets{G(x)^{-\frac12}G(y)G(x)^{-\frac12}}\leq \exp\parenth{8\sqrt{n}\vecnorm{x-y}{G(x)}}.
\end{equation}
\end{fact}

The upper bound on Frobenius norm of $G(x)^{-\frac12}\brackets{G(y)-G(x)}G(x)^{-\frac12}$ was proved as Lemma 1.2 in \cite{laddha2020strong}. However, we are not able to recover this bound following the argument in \cite{laddha2020strong}. We followed the procedures in \cite{laddha2020strong} and found our bound (Eq.~\ref{eq_ssc_frob}) differs from  \cite{laddha2020strong} by a factor of 2. Even though this does not have a major impact over the applications of this lemma, we still list our proof for the sake of completeness. 

The upper bound of the determinant of $G(x)^{-\frac12}\brackets{G(y)-G(x)}G(x)^{-\frac12}$ (Eq.~\eqref{eq_ssc_det_col}) is an easy application of Cauchy-Schwarz inequality over Eq.~\eqref{eq_ssc_frob}. Same as in \cite{laddha2020strong}, we first introduce the weaker notion of self-concordance and a well-known lemma:

\begin{definition}[Self-Concordance]
    For convex set $K\subseteq \real^n$, we call a local metric $G:K\to\real^{n\times n}$ self-concordant if for any $x\in K$, any direction $h\in\real^n$, the derivative of $G$ along $h$ satisfies:
    \begin{equation*}
    -2\vecnorm{h}{G(x)}G(x)\preceq \deri G(x)[h]\preceq 2\vecnorm{h}{G(x)} G(x).
    \end{equation*}
\end{definition}

It is easy to see from the definition that all strongly self-concordant local metric $G$ is self-concordant. The proof of Fact~\ref{fact_ssc_col} requires the following properties of self-concordant local metrics. For the sake of completeness, we list it as Fact \ref{fact_sc_property} without proof.

\begin{fact}[Lemma 1.1 from \cite{laddha2020strong}]\label{fact_sc_property}
Given any self-concordant matrix function $G$ on $K\subseteq \real^n$,  for any $x,y\in K$ with $\vecnorm{y-x}{G(x)}<1$, we have 
\begin{equation*}
\brackets{1-\vecnorm{x-y}{G(x)}}^2G(x)\preceq G(y)\preceq {\brackets{1-\vecnorm{x-y}{G(x)}}}^{-2}G(x).
\end{equation*}
\end{fact}

Now we are prepared to give the proof of Fact~\ref{fact_ssc_col}: 

\begin{proof}[proof of Fact~\ref{fact_ssc_col}]
First, we prove Eq.~\eqref{eq_ssc_frob}. Fix $x,y\in K$, let $h\defn y-x$. For $t\in [0,1]$, define $x_t\defn x+th$, then we have 
    \begin{equation*}
    \begin{aligned}
     \vecnorm{G(x)^{-\frac12}(G(y)-G(x))G(x)^{-\frac12}}{F}
    =&\vecnorm{G(x)^{-\frac12}\brackets{\int_0^1\frac{d}{dt}G(x_t)dt}G(x)^{-\frac12}}{F}\\
    \leq &\int_0^1 \vecnorm{G(x)^{-\frac12}{\frac{d}{dt}G(x_t)}G(x)^{-\frac12}}{F}dt.  
    \end{aligned}
    \end{equation*}
Then we can further bound the integrand:
\begin{equation*}
\begin{aligned}
\vecnorm{G(x)^{-\frac12}{\frac{d}{dt}G(x_t)}G(x)^{-\frac12}}{F}^2&= \trace\braces{G(x)^{-1}{\frac{d}{dt}G(x_t)}G(x)^{-1}{\frac{d}{dt}G(x_t)}}\\
&\overset{(i)}\leq \brackets{1-\vecnorm{x_t-x}{G(x)}}^{-4}\trace\braces{G(x_t)^{-1}{\frac{d}{dt}G(x_t)}G(x_t)^{-1}{\frac{d}{dt}G(x_t)}}\\
&=\brackets{1-\vecnorm{x_t-x}{G(x)}}^{-4}\vecnorm{G(x_t)^{-\frac12}{\deri G(x_t)[h]}G(x_t)^{-\frac12}}F^2\\
&\overset{(ii)}\leq  4\brackets{1-\vecnorm{x_t-x}{G(x)}}^{-4}h\tp G(x_t)h\\
&\overset{(iii)}{\leq}4\brackets{1-\vecnorm{x_t-x}{G(x)}}^{-4} \brackets{1-\vecnorm{x_t-x}{G(x)}}^{-2}h\tp G(x)h\\
&=4\brackets{1-t\vecnorm{h}{G(x)}}^{-6} h\tp G(x)h,
\end{aligned}
\end{equation*}
where inequality $(i)$ and $(iii)$ hold because we apply Lemma \ref{fact_sc_property} at $x_t$, and inequality $(ii)$ holds due to the assumption that $G$ is strongly self-concordant. Using the upper bound of the integrand, we can bound the integral:
\begin{equation*}
\begin{aligned}
\vecnorm{G(x)^{-\frac12}(G(y)-G(x))G(x)^{-\frac12}}{F}
&\leq \int_0^1 \frac{2\vecnorm{h}{G(x)}dt}{(1-t\vecnorm{h}{G(x)})^{3}}\overset{(i)}=\frac{2\vecnorm{h}{G(x)}-\vecnorm{h}{G(x)}^2}{(1-\vecnorm{h}{G(x)})^2},
\end{aligned}
\end{equation*}
where equality $(i)$ is a direct calculation of the definite integral. As a result, Eq.~\eqref{eq_ssc_frob} is proved since $\vecnorm{h}{G(x)}^2\geq 0$. However, we notice that the definite integral implies an extra factor of $2$ compared to Lemma 1.2 in \cite{laddha2020strong}. 

Next, to prove Eq.~\eqref{eq_ssc_det_col}, we the eigenvalues of $G(x)^{-\frac12}G(y)G(x)^{-\frac12}$ to be $\lambda_i$ for $i=1,\dots,n$. Then the determinant of $G(x)^{-\frac12}G(y)G(x)^{-\frac12}$ can be written as:
\begin{equation*}
\det\brackets{G(x)^{-\frac12}G(y)G(x)^{-\frac12}}=\prod_{i=1}^n\lambda_i\leq \exp\brackets{\sum_{i=1}^n(\lambda_i-1)}.
\end{equation*}

To prove Eq.~\eqref{eq_ssc_det_col}, we only need to prove $\sum\abss{\lambda_i-1}\leq 8\sqrt{n}\vecnorm{y-x}{G(x)}$. This bound can be obtained:
\begin{equation*}
\begin{aligned}
\sum_{i=1}^n\abss{\lambda_i-1}&\leq\sqrt{n}\brackets{\sum_{i=1}^n(\lambda_i-1)^2}^{\frac12}=\sqrt{n}\vecnorm{G(x)^{-\frac12}G(y)G(x)^{-\frac12}-I}{F}\\
&\overset{(i)}\leq \frac{2\sqrt{n}\vecnorm{y-x}{G(x)}}{(1-\vecnorm{y-x}{G(x)})^2}\overset{(ii)}\leq 8\sqrt{n}\vecnorm{y-x}{G(x)},
\end{aligned}
\end{equation*}
where inequality $(i)$ is due to Eq.~\eqref{eq_ssc_frob} for strongly self-concordant $G$, and inequality $(ii)$ holds since we assumed $\vecnorm{y-x}{G(x)}\leq\frac12$. 
\end{proof}

\subsection{Bounding Transition Overlap}\label{appendix_transit_overlap}
In this section, we are prove the transition overlap in Lemma \ref{lem_TV_control}. We separate this task into three parts: Lemma \ref{lem_accept} bounds the acceptance rate to be a little less than $1/2$ globally, Lemma \ref{lem_coupling} extends the close coupling argument in \cite{andrieu_explicit_2024} and circumvent the naive use of triangle inequality $\vecnorm{\cT_x-\cT_y}{TV}\leq \vecnorm{\cP_x-\cT_x}{TV}+\vecnorm{\cP_x-\cP_y}{TV}+\vecnorm{\cP_y-\cT_y}{TV}$, and Lemma \ref{lem_TV_proposal} controls the TV-distance between two proposal distributions using SSC and Pinsker's inequality. 
Throughout the proof, we need the following two facts. 
\begin{fact}[Gaussian Concentration]\label{fact_Gauss_concen}
Assume $\xi\sim \Normal(0,I_n)$, then for any $t>0$, we have 
\begin{equation*}
\prob\parenth{\vecnorm{\xi}{2}\geq t\sqrt{n}}\leq 2\exp\parenth{-\frac{t^2}{2}}        
\end{equation*}
\end{fact}

\begin{lemma}[Acceptance Rate Control]\label{lem_accept}
Let $G$ be a SSC, ASC, LTSC matrix function defined on $K$, and we also assume $G(x)\succeq \beta I$ for $x\in K$, then there exists an absolute constant $r_0>0$, such that for all step-size $r<r_0$, we have 
\begin{equation*}
\vecnorm{\cT_x-\cP_x}{TV}\leq 0.6
\end{equation*}
\end{lemma}
\begin{proof}
Due to Markov inequality, fix any $\gamma\in (0,1)$, we have 
\begin{equation}\label{eq_accept_Markov}
\begin{aligned}
\vecnorm{\cT_x-\cP_x}{TV}&=1-\Exs_{z\sim\cP_x}\brackets{\alpha(x,z)}\\
&\leq 1-\gamma \prob_{z\sim \cP_x}\brackets{\alpha(x,z)\geq \gamma}\\
&\leq 1- \gamma \prob_{z\sim \cP_x}\brackets{\mathbf{1}_K(z) \frac{e^{-f(z)}p_z(x)}{e^{-f(x)}p_x(z)} \geq \gamma}\\
\end{aligned}
\end{equation}
For convenience, we can first assume that $r\leq 10^{-4}$. It is easy to bound the term $\mathbf{1}_K(x)$, using the fact that the ellipsoid $E(x,G(x),1)$ is contained in $K$, combined with Gaussian concentration properties. To be explicit, we have 
\begin{equation}\label{eq_accept_indicator}
\begin{aligned}
\prob\braces{\mathbf{1}_K(z)=1}&\geq \prob\braces{\vecnorm{z-x}{G(x)}<1}=   
\prob_{\xi\sim\Normal(0,I_n)}\braces{\xi\tp \xi<\frac{n}{r^2}}\\
&\overset{(i)}\geq 1-2\exp\parenth{\frac{-1}{2r^2}}\overset{(ii)}\geq 0.99
\end{aligned}
\end{equation}
where the inequality $(i)$ is a result of Fact \ref{fact_Gauss_concen}, and the inequality $(ii)$ holds since we set $r\leq 10^{-4}$. Now we bound the term $\frac{p_z(x)}{p_x(z)}$, for convenience we define $g(x)\defn \log\det G(x)$, and we have 
\begin{equation*}
\begin{aligned}
    \log\frac{p_z(x)}{p_x(z)}
    &=-\frac{n}{2r^2}(z-x)\tp (G(z)-G(x))(z-x)+\frac12(g(z)-g(x))\\
    &\overset{(i)}=\underbrace{-\frac{n}{2r^2}(z-x)\tp (G(z)-G(x))(z-x)}_{\RN{1}}+\frac12\underbrace{\nabla g(x)\tp (z-x)}_{\RN{2}}+\frac14\underbrace{ \deri^2g(x^*)[z-x,z-x]}_{\RN{3}}.
\end{aligned}
\end{equation*}
where  we used Taylor's theorem, and $x^*\in[x,z]$ denotes some midpoint for which equality $(i)$ holds. We first bound term $\RN{1}$ with high probability. By the ASC property in Definition \ref{def_sc}, there exists an $\rasc>0$ such that the following holds for any $r<\rasc$
\begin{equation}\label{eq_accept_term1}
\prob_{z\sim \Normal(x,\frac{r^2}nG(x))}\braces{\vecnorm{z-x}{G(z)}^2-\vecnorm{z-x}{G(x)}^2 \leq 2\cdot\frac{ 0.01 r^2}{n}}\geq 1-0.01
\end{equation}
Then we bound term $\RN{2}$ with high probability. We notice that term $\RN{2}$ can be bounded by the concentration of Gaussian variables. Using Fact \ref{fact_Gauss_concen} with $t=4$, then with probability $\geq 0.99$ in $z$, we have  
\begin{equation}\label{eq_accept_term2}
\begin{aligned}
\RN{2}&=\angles{G(x)^{-\frac12} \nabla g(x),G(x)^{\frac12}(z-x)}\geq -\frac{4r}{\sqrt{n}}\vecnorm{G(x)^{-\frac12}\nabla g(x)}{2}\\
&=\frac{-4r}{\sqrt{n}}\underset{\vecnorm{v}{2}=1}{\sup}\nabla g(x)\tp G(x)^{-\frac12}v\\
&=\frac{-4r}{\sqrt{n}}\underset{\vecnorm{v}{2}=1}{\sup} \trace\parenth{G(x)^{-1}\deri G(x)\brackets{G(x)^{-\frac12}v}}\\
&\geq -4r\underset{\vecnorm{v}{2}=1}{\sup} \vecnorm{G(x)^{-\frac12}\deri G(x)\brackets{G(x)^{-\frac12}v}G(x)^{-\frac12}}F\\
& \geq -8r\underset{\vecnorm{v}{2}=1}{\sup} \vecnorm{G(x)^{-\frac12}v}{G(x)}=-8r,
\end{aligned}
\end{equation}
where the last inequality follows from SSC of $G$. Now we continue to bound term $\RN{3}$. For convenience, we let $h\defn z-x$ and we have
\begin{equation*}
\begin{aligned}
    \deri^2g(x^*)[h,h]&=\trace\parenth{G(x^*)^{-1}\deri^2 G(x^*)[h,h]}-\trace\parenth{G(x^*)^{-1}\deri G(x^*)[h]G(x^*)^{-1}\deri G(x^*)[h]}\\
    &\overset{(i)}\geq -\vecnorm{h}{G(x^*)}^2-\vecnorm{G(x^*)^{-\frac12}\deri G(x^*)G(x^*)^{-\frac12}}{F}^2\\
    &\overset{(ii)}\geq -\vecnorm{h}{G(x^*)}^2-4\vecnorm{h}{G(x^*)}^2
    \overset{(iii)}\geq  -5{({1-\vecnorm{h}{G(x)}})^{-2}}\vecnorm{h}{G(x)}^2,
\end{aligned}
\end{equation*}
where inequality $(i)$ follows from LTSC of $G$, inequality $(ii)$ follows from ASC of $G$, and inequality $(iii)$ follows from self-concordant properties (see Lemma \ref{fact_sc_property}). Using Fact \ref{fact_Gauss_concen} with $t=4$, then we have $\prob_{z}\parenth{\vecnorm{h}{G(x)}>4r}\leq 0.01$, and this implies with probability $\geq 0.99$ we have 
\begin{equation}\label{eq_accept_term3}
   \RN{3}=\deri^2 g(x^*)[h,h]\overset{(i)}\geq -\cdot 5\cdot 4 \cdot 16r^2=-320r^2
\end{equation}
where inequality $(i)$ holds due to we assumed $r\leq \frac12$ and we insert $\vecnorm{h}{G(x)}>4r$. Now insert the bound of $\RN{1},\RN{2},\RN{3}$ in Eq.~\eqref{eq_accept_term1},\eqref{eq_accept_term2} and \eqref{eq_accept_term3}, together with the fact that we set $r\leq 10^{-4}$, then we have
\begin{equation}\label{eq_accept_proposal}
    \prob_{z\sim\cP_x}\braces{\log\frac{p_z(x)}{p_x(z)}\geq -0.01-4r-80r^2\geq -0.01}\geq 0.97.
\end{equation}
Now we continue to bound the term of stationary distribution
\begin{equation*}
\begin{aligned}
    \frac{e^{-f(z)}}{e^{-f(x)}}&\overset{(i)}\geq \exp\parenth{\nabla f(z)\tp (x-z)}\\
    &=\exp\brackets{(\nabla f(z)-\nabla f(x))\tp (x-z)+\nabla f(x)\tp (x-z)}\\
    &\overset{(ii)}\geq \exp\parenth{-\beta\vecnorm{z-x}{2}^2}\exp\brackets{\nabla f(x)\tp (x-z)}\\
    &\overset{(iii)}\geq \exp\parenth{-\vecnorm{z-x}{G(x)}^2}\exp\brackets{\nabla f(x)\tp (x-z)}\\
\end{aligned}
\end{equation*}
where the inequality $(i)$ follows from the convexity of $f$, the inequality $(ii)$ follows from the $\beta$-smoothness of $f$, and the inequality $(iii)$ follows from our assumption that $G(x)\succeq \beta I$. In order to control the term $\nabla f(x)\tp (z-x)$, we notice that $\nabla f(x)\tp (z-x)$ is a one dimensional Gaussian variable centered at $0$.  By the symmetric property of Gaussian distributions, we have 
\begin{equation}\label{eq_accept_stationary}
    \prob_{z\sim\cP_x}\braces{\exp\brackets{\nabla f(x)\tp (x-z)}\geq 1}\geq 0.5
\end{equation}
Now combining Eq.~\eqref{eq_accept_indicator}, \eqref{eq_accept_proposal} and \eqref{eq_accept_stationary}, we have 
\begin{equation}\label{eq_accept_final1}
    \prob_{z\sim\cP_x}\brackets{\mathbf{1}_K(z)\frac{e^{-f(z)}p_z(x)}{e^{-f(x)}p_x(z)}\geq \exp\parenth{-\vecnorm{z-x}{G(x)}^2-0.01}}\geq 0.5-0.04
\end{equation}
Using the fact that of Gaussian concentration
\begin{equation*}
\prob\braces{\vecnorm{z-x}{G(x)}^2\geq 16r^2}=\prob_{\xi\sim\Normal(0,I_n)}\braces{\vecnorm{\xi}{2}\geq 4n}\leq 2\exp\parenth{-8}\leq 0.01
\end{equation*}
This translate Eq.~\eqref{eq_accept_final1} into:
\begin{equation}
\prob_{z\sim\cP_x}\brackets{\mathbf{1}_K(z)\frac{e^{-f(z)}p_z(x)}{e^{-f(x)}p_x(z)}\geq \exp\parenth{-16r^2-0.01}}\geq 0.46-0.01
\end{equation}

As a result of Markov inequality (Eq.~\eqref{eq_accept_Markov}), we have 
\begin{equation*}
    \vecnorm{\cT_x-\cP_x}{TV}\leq 1-0.45\cdot\exp(-16r^2-0.01)\leq 0.6,
\end{equation*}
and it holds for all $r\leq \min\braces{\rasc,10^{-4}}$. 
\end{proof}

\begin{lemma}[close coupling]\label{lem_coupling}
Let $G$ be a SSC, ASC matrix function defined on $K$. Fix any $\couplingprob\in (0,1)$, there exists $r_\couplingprob>0$ such that for all $r\leq r_\couplingprob$, $x,y\in K$ and 
\begin{equation}\label{eq_close_x_y}
\max\braces{\vecnorm{y-x}{G(x)},\vecnorm{y-x}{G(y)}}\leq \frac{r}{\sqrt{n}},
\end{equation} 
we have the following bound on the TV-distance between transition distributions:
\begin{equation}\label{eq_coupling}
\vecnorm{\cT_y-\cT_x}{TV}\leq \couplingprob+\vecnorm{\cT_x-\cP_x}{TV}+\vecnorm{\cP_x-\cP_y}{TV}
\end{equation}
\end{lemma}

\begin{proof}
This lemma extends the close-coupling technique in \cite{andrieu_explicit_2024} to our asymmetric proposals. For convenience, we use $p_x(\cdot)$ to denote the probability density function of $\cP_x$. In other words, we define $p_x(z)\defn \Normal(z|x,\frac{r^2}{n}G(x)^{-1})$.

Without loss of generality, we assume $e^{-f(x)}\geq e^{-f(y)}$. 

Consider the maximal coupling of $(\cP_x,\cP_y)$, there exists a pair of random vectors $(V_x,V_y)$ such that $V_x\sim \cP_x$ and $V_y\sim \cP_y$ and $\prob\braces{V_x=V_y}=1-\vecnorm{\cP_x-\cP_y}{TV}$. Then we draw $U\sim \text{unif}(0,1)$ independent of $(V_x,V_y)$, and we define random vectors $X',Y'$ according to the following rule:
\begin{equation*}
\begin{aligned}
    X'&\defn V_x \Ind{U\leq \alpha(x,V_x)}+x\Ind{U>\alpha(x,V_x)},\\
    Y'&\defn V_y\Ind{U\leq \alpha(y,V_y)}+y\Ind{U>\alpha(y,V_y)},
\end{aligned}    
\end{equation*}
where $\alpha(x,z)$ denotes the acceptance rate $\alpha(x,z)\defn \min\braces{1,\frac{e^{-f(z)}p_z(x)}{e^{-f(x)}p_x(z)}\mathbf{1}_K(z)}$. For convenience of discussion, we define two events $\cE_1$ and $\cE_2$,
\begin{equation*}
    \cE_1\defn \braces{V_x=V_y,X'=V_x,Y'\neq V_y},\quad 
    \cE_2\defn \braces{V_x=V_y,\alpha(y,V_y)<\alpha(x,V_x)}.
\end{equation*}
It is straightforward to verify that $X'\sim \cT_x$ and $Y'\sim \cT_y$, so we have 
\begin{equation*}
\begin{aligned}
    \prob\braces{X'=Y'}&\geq \prob\braces{V_x=V_y,X'=V_x,Y'=V_y}\\
    &=\prob\braces{V_x=V_y,X'=V_x}-\prob(\cE_1)\\
    &\geq \prob\braces{V_x=V_y}+\prob\braces{V_x=X'}-1-\prob\parenth{\cE_1},
\end{aligned}
\end{equation*}
where the last inequality follows from the fact that $1\geq \prob(B_1\cup B_2)=\prob(B_1)+\prob(B_2)-\prob(B_1\cap B_2)$ for any event $B_1,B_2$. Then we insert the coupling equation $\prob\braces{V_x=V_y}=1-\vecnorm{\cP_x-\cP_y}{TV}$, and the easy-to-verify fact $\prob\braces{V_x=X'}=1-\vecnorm{\cP_x-\cT_x}{TV}$, we have
\begin{equation*}
\begin{aligned}
    1-\prob\braces{X'=Y'}&\leq \vecnorm{\cP_x-\cP_y}{TV}+\vecnorm{\cP_x-\cT_x}{TV}+\prob(\cE_1)\\
    \vecnorm{\cT_x-\cT_y}{TV}&\leq \vecnorm{\cP_x-\cP_y}{TV}+\vecnorm{\cP_x-\cT_x}{TV}+\prob(\cE_1),
\end{aligned}
\end{equation*}
where the last inequality follows by the coupling inequality $\vecnorm{\cT_x-\cT_y}{TV}\leq \prob\braces{X'\neq Y'}$ when $X'\sim \cT_x$ and $Y'\sim\cT_y$. In order to prove Eq.~\eqref{eq_coupling}, we only need to control $\prob(\cE_1)$ to be smaller than $\couplingprob$,
\begin{equation*}
\begin{aligned}
    \prob(\cE_1)&\overset{(i)}\leq\prob\braces{V_x=V_y,\alpha(y,V_y)<U\leq \alpha(x,V_x)} \\
    &=\Exs_{(V_x,V_y,U)}\brackets{\Ind{V_x=V_y}\cdot \Ind{\alpha(y,V_y)<U\leq \alpha(x,V_x)}}\\
    &\overset{(ii)}=\Exs_{(V_x,V_y)}\braces{\Ind{V_x=V_y,\alpha(y,V_y)<\alpha(x,V_x)}\times \brackets{\alpha(x,V_x)-\alpha(y,V_y)}}
\end{aligned}
\end{equation*}
where inequality $(i)$ holds since $Y'\neq V_y\Rightarrow U>\alpha(y,V_y)$, and $X'=V_x\Rightarrow U\leq \alpha(x,V_x)$ except on a set with probability $0$ (the proposal satisfies $V_x=x$), because our proposal has no point mass. Equality $(ii)$ follows from our assumption that $U$ is independent of $(V_x,V_y)$. Then we continue to control the difference in acceptance rates at $x$ and $y$. For convenience, we denote $\acceptbeforemin{x}{z}$ to be the following ratio,
\begin{equation*}
    \acceptbeforemin{x}{z}\defn \frac{e^{-f(z)}p_z(x)}{e^{-f(x)}p_x(z)}.
\end{equation*}
Then we have
\begin{equation*}
\begin{aligned}
    \prob(\cE_1)&\leq \Exs_{(V_x,V_y)}\brackets{\mathbb{I}_{\cE_2}\cdot \parenth{\alpha(x,V_x)-\alpha(y,V_y)}}\\
    &\overset{(i)}= \Exs_{(V_x,V_y)}\brackets{\mathbb{I}_{\cE_2}\cdot \parenth{\min\braces{1,\acceptbeforemin{x}{V_x}\mathbf{1}_K(V_x)}-{\acceptbeforemin{y}{V_y}\mathbf{1}_K(V_y)}}}\\
    &\overset{(ii)}\leq \Exs_{(V_x,V_y)}\brackets{\mathbb{I}_{\cE_2} \min\braces{1,{\acceptbeforemin{x}{V_x}\mathbf{1}_K(V_x)-\acceptbeforemin{y}{V_y}\mathbf{1}_K(V_y)}}}\\
    &\leq \underbrace{\Exs_{(V_x,V_y)}\brackets{\mathbb{I}_{\cE_2\cap \braces{V_x,V_y\in K}} \min\braces{1,{\acceptbeforemin{x}{V_x}-\acceptbeforemin{y}{V_y}}}}}_{\text{I}}+\underbrace{\prob\braces{V_x\notin K}+\prob\braces{V_y\notin K}}_{\text{II}},
\end{aligned}    
\end{equation*}
where equality $(i)$ follows since in the event $\cE_2$ we have $\alpha(y,V_y)<\alpha(x,V_x)\leq 1$. Inequality $(ii)$ holds since for all $a,b\geq 0$ we have $\min\braces{1,b}-a\leq \min(1,b-a)$. We proceed by controlling term I and term II separately. Term II can be controlled easily by Gaussian concentration as in Eq.~\eqref{eq_accept_indicator} in Lemma \ref{lem_accept}:
\begin{equation*}
\RN{2}\leq 2\prob_{\xi\sim\Normal(0,I_n)}\braces{\xi\tp\xi>\frac{n}{r^2}}\leq 4\exp\parenth{-\frac{1}{2r^2}}\overset{(i)}\leq \frac{\gamma}{2},
\end{equation*}
where we only need to set $r\leq \parenth{{2\log(\frac{8}{\gamma})}}^{-1/2}$ to ensure inequality $(i)$ holds. Now we proceed to control term $\RN{1}$:

\begin{equation}\label{eq_coupling1}
\begin{aligned}
\text{I}&\overset{(i)}\leq   \Exs_{(V_x,V_y)}\brackets{\mathbb{I}_{\cE_2\cap \braces{V_x,V_y\in K}} \min\braces{1,{\frac{e^{-f(V_y)}}{e^{-f(y)}} \parenth{\frac{p_{V_x}(x)}{p_x(V_x)}-\frac{p_{V_y}(y)}{p_y(V_y)}}}}}\\
&\leq \Exs_{(V_x,V_y)}\brackets{\mathbb{I}_{\cE_2\cap \braces{V_x,V_y\in K}} \min\braces{1,{\frac{e^{-f(V_y)}p_{V_y}(y)}{e^{-f(y)}p_y(V_y)} \abss{\frac{p_{V_x}(x)p_y(V_y)}{p_x(V_x)p_{V_y}(y)}-1}}}}\\
&\overset{(ii)}\leq \Exs_{(V_x,V_y)}\brackets{\mathbb{I}_{\cE_2\cap \braces{V_x,V_y\in K}} \min\braces{1, \abss{\frac{p_{V_x}(x)p_y(V_y)}{p_x(V_x)p_{V_y}(y)}-1}}}\\
\end{aligned}
\end{equation}
where inequality $(i)$ holds since $V_x=V_y$ in the event $\cE_2$ and we assumed $e^{-f(x)}\geq e^{-f(y)}$. Inequality $(ii)$ holds because on the event $\cE_2\cap \braces{V_x,V_y\in K}$ we have
\begin{equation*}
    \acceptbeforemin{y}{V_y}=\acceptbeforemin{y}{V_y}\mathbf{1}_K(y)=\alpha(y,V_y)<1.
\end{equation*}
For convenience, we let the notation $\Gamma(x,z)$ denote the following expression:
\begin{equation}
\distdiff{x}{z}=-\frac{n}{2r^2}\parenth{z-x}\tp\parenth{G(z)-G(x)}\parenth{z-x}.
\end{equation}
Now we continue the inequality in Eq.~\eqref{eq_coupling1}: 
\begin{equation*}
\begin{aligned}
    \text{I}&\leq \Exs_{(V_x,V_y)}\brackets{\mathbb{I}_{\cE_2\cap \braces{V_x,V_y\in K}} \min\braces{\abss{\frac{p_{V_x}(x)p_y(V_y)}{p_x(V_x)p_{V_y}(y)}-1},1}}\\
    &\overset{(i)}=\Exs_{(V_x,V_y)}\brackets{\mathbb{I}_{\cE_2\cap \braces{V_x,V_y\in K}} \min\braces{\abss{\sqrt{\frac{\det{ G(y)}}{\det G(x)}} 
    \frac{\exp\brackets{\distdiff{x}{V_x}}}{\exp\brackets{\distdiff{y}{V_y}}}
    -1},1}}\\
    &\leq \Exs_{(V_x,V_y)}\min\braces{\abss{\sqrt{\frac{\det{ G(y)}}{\det G(x)}} 
    \frac{\exp\brackets{\distdiff{x}{V_x}}}{\exp\brackets{\distdiff{y}{V_y}}}
    -1},1}\\
    &\overset{(ii)}\leq \prob\braces{\abss{\sqrt\frac{\det{ G(y)}}{\det G(x)} 
    \frac{\exp\brackets{\distdiff{x}{V_x}}}{\exp\brackets{\distdiff{y}{V_y}}}-1}\geq e^{3c}-1}+(e^{3c}-1)\\
    &\overset{(iii)}\leq \prob\braces{\abss{\distdiff{y}{V_y}}\geq c}+\prob\braces{\abss{\distdiff{x}{V_x}}\geq c}+(e^{3c}-1)
\end{aligned} 
\end{equation*}
where inequality $(i)$ holds since we insert the proposal densities and we notice $V_x=V_y$ in the event $\cE_2$. Inequality $(ii)$ holds for any constant $c\geq 0$, and we would determine $c$ later. Since we have $\sqrt{\det (G(y)G(x)^{-1})}\in [e^{-4r},e^{4r}]$ due to SSC property in Fact \ref{fact_Gauss_concen}  and the closeness of $x$ and $y$ in Eq.~\eqref{eq_close_x_y}, and we can always set $r\leq c/4$ so that inequality $(iii)$ holds. 

Now we set $c\defn \min\braces{\frac13\log(1+\frac{\gamma}{6}),\frac{\gamma}{6}}$, using ASC property in Definition \ref{def_sc}, then there exists $r_c>0$ such that for any step-size $r\leq r_c$, we have 
\begin{equation}
\RN{1}\leq c+c+e^{3c}-1\leq \frac{\gamma}{6}+\frac{\gamma}{6}+\frac{\gamma}{6}= \frac{\gamma}{2}
\end{equation}
The result is proved if we combine term $\RN{1}$ and $\RN{2}$ for all $r<\min\braces{\frac{c}{4},r_c,\parenth{{2\log(\frac{8}{\gamma})}}^{-1/2}}$. 
\end{proof}

\begin{lemma}[Proposal Overlap]\label{lem_TV_proposal}
Assume $r\leq \frac1{16}$ to be the step-size, and $x,y\in K$ such that $\vecnorm{y-x}{G(x)}\leq \frac{r}{10\sqrt{n}}$.
Recall that we use $\cP_x$, $\cP_y$ to denote the Gaussian proposal distributions at $x,y\in K$. In other words, $\cP_x\defn \Normal(x,\frac{r^2}{n}G(x)^{-1})$, $\cP_y\defn \Normal(x,\frac{r^2}{n}G(y)^{-1})$, then we have the following bound:
\begin{equation*}
    \vecnorm{\cP_x-\cP_y}{TV}\leq \sqrt{\frac{1}{100}+24r}
\end{equation*}
\end{lemma}
\begin{proof}
By Pinsker's inequality $\vecnorm{\cP_x-\cP_y}{TV}\leq \sqrt{2\kldiv{\cP_y}{\cP_x}}$. The KL-divergence between two Gaussian distributions $P=\Normal(\mu_1,\Sigma_1)$ and $Q=\Normal(\mu_2,\Sigma_2)$ can be computed analytically:
\begin{equation}
    \kldiv{Q}{P}=\frac12\braces{\underbrace{\trace(\Sigma_1^{-1}\Sigma_2)-n}_{\text{I}}+\underbrace{\log\det \parenth{\Sigma_1 \Sigma_2^{-1}}}_{\text{II}}+\underbrace{(\mu_1-\mu_2)\tp \Sigma_1^{-1}(\mu_1-\mu_2)}_{\text{III}}}.
\end{equation}
We plug in $P\defn \Normal(x,\frac{r^2}{n}G(x)^{-1})$ and $Q\defn \Normal(y,\frac{r^2}{n}G(y)^{-1})$, then bound the terms I, II, and III separately. We first bound the determinant ratio I:
\begin{equation*}
\text{I}=\log\det\parenth{G(y)G(x)^{-1}}\overset{(i)}\leq  8\sqrt{n}\vecnorm{y-x}{G(x)}\leq 8r,
\end{equation*}
where inequality $(i)$ follows from the properties of SSC in Fact \ref{fact_ssc_col}. Then we bound the trace term II, we denote the eigenvalues of $G(x)^{-\frac12}G(y)G(x)^{-\frac12}$ to be $\lambda_1,\ldots,\lambda_n$, so according to Fact \ref{fact_ssc_col} we have  
\begin{equation}\label{eq_eigen_control}
    \sum_{i=1}^n \abss{\lambda_i-1}\leq \sqrt{n}\vecnorm{G(x)^{-\frac12}G(y)G(x)^{-\frac12}-I}{F}\overset{(i)}\leq \sqrt{n}\cdot \frac{2r/\sqrt{n}}{\parenth{1-\frac12}^2}=8r,
\end{equation}
where inequality $(i)$ holds because we assumed $r\leq \frac12$. Since we further assumed $r<\frac{1}{16}$, thus Eq.~\eqref{eq_eigen_control} further implies that $\lambda_i\geq \frac12$ for all index $i\in [n]$. Thus we have
\begin{equation*}
\begin{aligned}
    \text{II}=-n+\trace\parenth{G(x)G(y)^{-1}}&=\sum_{i=1}^n \parenth{\frac{1}{\lambda_i}-1}
    \leq \sum_{i=1}^n\abss{\frac{\lambda_i-1}{\lambda_i}}\leq 2\sum_{i=1}^n\abss{\lambda_i-1}\leq 16r.
\end{aligned}
\end{equation*}
Finally, term $\RN{3}$ can be bounded as
\begin{equation*}
    \RN{3}=\frac{n}{r^2}(y-x)\tp {G(x)}(y-x)=\frac{n}{r^2}\vecnorm{y-x}{G(x)}^2 \leq \frac{n}{r^2}\cdot \frac{r^2}{100n}=\frac{1}{100}.
\end{equation*}
\end{proof}

With the preparation of Lemma \ref{lem_accept}, \ref{lem_coupling} and \ref{lem_TV_proposal}, we now give the proof of Lemma \ref{lem_TV_control}. 
\begin{proof}[proof of Lemma \ref{lem_TV_control}]
We first notice that $\vecnorm{y-x}{G(y)}\leq \frac{\vecnorm{y-x}{G(x)}}{{1-\vecnorm{y-x}{G(x)}}}$ by Fact \ref{fact_sc_property}, and since we can set $r\leq \frac12$, we have 
\begin{equation*}
    \max\braces{\vecnorm{y-x}{G(x)},\vecnorm{y-x}{G(y)}}\leq 2\vecnorm{y-x}{G(x)}\leq \frac{r}{\sqrt{n}},
\end{equation*}
so we can use Lemma \ref{lem_accept} and Lemma \ref{lem_coupling} with $\gamma=0.01$, then there exists $r_0, r_\gamma>0$ such that for all $r<\min\braces{10^{-4},r_0, r_\gamma}$ we have 
\begin{equation*}
\begin{aligned}
\vecnorm{\cT_x-\cT_y}{TV}&\leq 0.01+\vecnorm{\cT_x-\cP_x}{TV}+\vecnorm{\cP_x-\cP_y}{TV}\\
&\leq 0.01+ 0.6+\vecnorm{\cP_x-\cP_y}{TV}\\
&\leq 0.61+ \sqrt{\frac{1}{100}+ 24r}\leq \frac{4}{5}. 
\end{aligned}
\end{equation*}

\end{proof}

\section{Properties of Soft-Threshold  Metric}\label{appendix_logarithmic}

In this section, we briefly verify the SSC, LTSC and ASC of \softlog as defined in Definition \ref{def_logarithmic}, and proving the correctness of Lemma \ref{lem_logarithmic_sc}. 

For the SSC, it is well-known that the \Hlogmetric $H(x)=A_x\tp A_x$ is SSC when $H(x)$ is invertible. Adding a regularization term $\lambda I_n$ is only helping us: 
\begin{equation*}
\vecnorm{G(x)^{-\frac12}\deri G(x)[h]G(x)^{-\frac12}}{F}
\end{equation*} is becoming smaller since $G(x)\succeq H(x)$ and $\deri H(x)[h]=\deri G(x)[h]$. The rigorous version of this instinct is shown in Lemma \ref{lem_matrix_inequality}. 

Our proof used the fact that the \Hlogmetric $H(x)=A_x\tp A_x$ is strongly self-concordant if $H(x)$ is invertible  \cite{laddha2020strong}, and also the fact that $H(x)$ is SLTSC when $H(x)$ is invertible \cite{pmlr-v247-kook24b}. This is summarized in Lemma \ref{lem_log_SLTSC_SASC}. However, in our case, we do not restrict $H(x)$ to be invertible since we allow $K$ to be unbounded and $m<n$. In order to prove $G(x)=H(x)+\lambda I$ is SSC and LTSC, we would use a limit argument: we add artificial constraints so that $H(x)\defn A_xA_x\tp$ is invertible so we can apply Lemma \ref{lem_log_SLTSC_SASC} and the additivity of LTSC as in Fact \ref{fact_addition}, thus $G(x)=H(x)+\lambda I$ is also SSC and LTSC, then we take the limit and remove all artificial constraints.

\begin{lemma}\label{lem_matrix_inequality}
    If $A, B, C\in\real^{n\times n}$ are a symmetric matrices, and if $0\preceq A\preceq C$, then we have 
    \begin{equation*}
    \vecnorm{A^{\frac12}BA^{\frac12}}{F}\leq \vecnorm{C^{\frac12}BC^{\frac12}}{F}.
    \end{equation*}
\end{lemma}

\begin{proof}
Note that  $\trace\braces{P^{\frac12}BPBP^{\frac12}}=\vecnorm{P^{\frac12}BP^{\frac12}}{F}^2$ for any symmetric PSD matrix $P$. So we only need to prove $\trace\braces{A^{\frac12}BABA^{\frac12}}\leq \trace\braces{C^{\frac12}BCBC^{\frac12}}$. This is easy to verify:
\begin{equation*}
\begin{aligned}
\trace\braces{A^{\frac12}BABA^{\frac12}}\leq\trace\braces{A^{\frac12}BCBA^{\frac12}}
=\trace\braces{C^{\frac12}BABC^{\frac12}}
\leq \trace\braces{C^{\frac12}BCBC^{\frac12}}.
\end{aligned}
\end{equation*}
\end{proof}

\begin{lemma}[Lemma $4.1$ \cite{laddha2020strong}, Lemma E.1 \cite{pmlr-v247-kook24b}]\label{lem_log_SLTSC_SASC}
Let $K=\braces{x|Ax>b}$ be a convex polytope in $\real^n$ and the logarithmic metric $H(x)\defn A_x\tp A_x$ is invertible for all $x\in K$. Then we have $H$ is SSC and SLTSC on $K$.
\end{lemma}

The remaining task in this section is to verify \softlog $G$ defined in Definition \ref{def_logarithmic} is ASC. This is proved by concentration of Gaussian polynomials, which also appears in \cite{sachdeva2016mixing}, and the intuition is that adding a regularization term $\lambda I$ only makes the Gaussian concentration tighter. Lemma \ref{lem_Gauss_concentration} is a general bound for Gaussian polynomials, and Lemma \ref{lem_Gauss_Exs_bound} is upper bounding the expectation of specific Gaussian polynomials that appears in our proof. The condition $\sum_{i=1}^m b_ib_i\tp=I$ in the original Lemma \ref{lem_Gauss_Exs_bound} is changed to $\sum_{i=1}^m b_ib_i\tp\preceq I$ due to our regularization, and the proof also changes. To be rigorous, we list the proof of Lemma \ref{lem_Gauss_Exs_bound} here.

\begin{lemma}[Theorem 6.7 from \cite{janson_gaussian_1997}]\label{lem_Gauss_concentration}
Let $P$ be a degree $q$ polynomial over $\real^n$, and $\xi\sim \Normal(0,I_n)$. Then for any $t\geq \parenth{2e}^{\frac{q}{2}}$ we have
\begin{equation}
\prob\brackets{\abss{P(\xi)}\geq t\parenth{\Exs P(\xi)^2}^{1/2}}\leq \exp\parenth{-\frac{q}{2e}t^{2/q}}        
\end{equation}
\end{lemma}

\begin{lemma}[Adapted from \textbf{Fact 10} in \cite{sachdeva2016mixing}]\label{lem_Gauss_Exs_bound}
Suppose $\xi\sim\Normal(0,I_n)$ and $b_i\in\real^n$ are vectors for $i\in [m]$ such that $\sum_{i=1}^m b_ib_i\tp \preceq I_n$, then we have the following bounds:
\begin{equation*}
\Exs\braces{\bigg[{\sum_{i=1}^m(b_i\tp \xi)^3}\bigg]^2}\leq 15n, \quad 
\Exs\braces{\bigg[\sum_{i=1}^m(b_i\tp \xi)^4\bigg]^2}\leq 105n^2,
\end{equation*}    
\end{lemma}
\begin{proof}
Same as in \textbf{Fact 10} in \cite{sachdeva2016mixing}, we get 
\begin{equation*}
\Exs\braces{\bigg[{\sum_{i=1}^m(b_i\tp \xi)^3}\bigg]^2}=9\sum_{i,j=1}^m(b_i\tp b_i)(b_j\tp b_j)(b_i\tp b_j)
+6\sum_{i,j=1}^m (b_i\tp b_j)^3,
\end{equation*}
Following the notations in \cite{sachdeva2016mixing}, we set $B$ to be the $m\times n$ matrix with its $i$-th row being $b_i\tp$, and $w\in\real^m$ be such that $w_i=b_i\tp b_i$. The first term is simplified to:
\begin{equation*}
\sum_{i,j=1}^m(b_i\tp b_i)(b_j\tp b_j)(b_i\tp b_j)=w\tp BB\tp w,    
\end{equation*}
since we assumed $\sum_{i\in[m]} b_ib_i\tp \preceq I_n$, thus $B\tp B\preceq I_n$. Since all non-zero eigenvalues of $BB\tp$ are the same with $B\tp B$, thus $BB\tp\preceq I_m$. As a result, $w\tp BB\tp w\leq w\tp w$. The remaining arguments in \cite{sachdeva2016mixing} go on smoothly, which we omit here.  
\end{proof}

With all the preparations, we now give the proof of Lemma \ref{lem_logarithmic_sc}.  

\begin{proof}[Proof of Lemma \ref{lem_logarithmic_sc}]

We define the two following local metrics for any $\gamma>0$. Fix $x\in K$, we define:
\begin{equation*}
\begin{aligned}
    H^{(\gamma)}(x)&\defn \sum_{i=1}^m \frac{a_ia_i\tp}{(a_i\tp x-b_i)^2}+\sum_{j=1}^n \frac{e_je_j\tp}{(e_j\tp x-\gamma)^2},\\
    G^{(\gamma)}(x)&\defn H^{(\gamma)}(x) +\lambda I_n,
\end{aligned}
\end{equation*}
where $e_j$ is the unit vector in the $j$-th direction for $j\in [n]$. We add the term $\sum_{j}\frac{e_je_j\tp}{(e_j\tp x-\gamma)^2}$ to ensure that $H^{(\gamma)}(x)$ is invertible, so we can use the strong self-concordance of \Hlogmetric. It is also clear that fix $x\in K$,  $\underset{\gamma\to \infty}{\lim} H^{(\gamma)}(x)=H(x)$ and $\underset{\gamma\to \infty}{\lim} 
 G^{(\gamma)}(x)=G(x)$. So we have 
\begin{equation*}
\begin{aligned}
   \vecnorm{\Ggamma^{-\frac12} \deri\Ggamma [h]\Ggamma^{-\frac12}}{F}&=\vecnorm{\Ggamma^{-\frac12} \deri\Hgamma [h]\Ggamma^{-\frac12}}{F}\\
   &\overset{(i)}\leq \vecnorm{\Hgamma^{-\frac12} \deri\Hgamma [h]\Hgamma^{-\frac12}}{F}\\
   &\overset{(ii)}{\leq} 2\vecnorm{h}{\Hgamma}\leq 2\vecnorm{h}{\Ggamma},
\end{aligned}
\end{equation*}
where inequality $(i)$ holds due to $\Ggamma \succeq \Hgamma$ and Lemma \ref{lem_matrix_inequality}, and inequality $(ii)$ holds because $\Hgamma$ is a \Hlogmetric, so is strongly self-concordant. It is clear that for any $x\in K$, $\Ggamma\to G(x)$ as $\gamma \to \infty$, so we take limits on both sides, and we have 
\begin{equation*}
 \vecnorm{G(x)^{-\frac12} \deri G(x) [h]G(x)^{-\frac12}}{F}\leq 2\vecnorm{h}{G(x)},
\end{equation*}
so we just proved SSC of $G$. Using the same technique, we can prove $G$ is LTSC. We know from Lemma \ref{lem_log_SLTSC_SASC} that $H^{(\gamma)}$ is SLTSC for every $\gamma>0$, then according to Fact \ref{fact_addition}, $G^{(\gamma)}\defn H^{(\gamma)}+\lambda I$ is LTSC since $\lambda I$ is clearly SLTSC. So fix any $x\in K$ and $h\in\real^n$ we have 
\begin{equation*}
    \trace\braces{\Ggamma\deri^2 \Ggamma[h,h]}\geq -\vecnorm{h}{\Ggamma}^2,
\end{equation*}
and take the limit $\gamma\to 0$ on both sides, we conclude that $G$ is LTSC. 

The proof follows similar steps in \cite{sachdeva2016mixing}, except that we define the \softlog in Definition  \ref{def_logarithmic}, so we allow the polytope $K$ to be unbounded, thus $A$ is not necessarily full-rank and $m$ could be smaller than $n$. 

We set $\hat{a}_i=\frac{G(x)^{-\frac12}a_i}{a_i\tp x- b_i}$, and noticing that $z=x+\frac{r}{\sqrt{n}}G(x)^{-\frac12} \xi$, we define term I to be:
\begin{equation*}
\begin{aligned}
\RN{1}&\defn \frac{n}{r^2}\parenth{\vecnorm{z-x}{G(z)}^2- \vecnorm{z-x}{G(x)}^2}\\
&={\xi\tp G(x)^{-\frac12}\brackets{G(z)-G(x)}G(x)^{-\frac12}\xi}\\
&=\xi\tp G(x)^{-\frac12}\braces{\sum_{i=1}^m\frac{a_ia_i\tp}{(a_i\tp x-b_i)^2}\brackets{\frac{1}{(1+\frac{r}{\sqrt{n}}\hat{a}_i\tp \xi)^2}-1}}G(x)^{-\frac12}\xi\\
&=\sum_{i=1}^m (\hat{a}_i\tp \xi)^2 \brackets{\frac{1}{(1+\frac{r}{\sqrt{n}}\hat{a}_i\tp \xi)^2}-1}\\
&=\frac{r^2}{n}\sum_{i=1}^m (\hat{a}_i\tp \xi)^4 \brackets{\frac{2}{1+\frac
r{\sqrt{n}}\hat{a_i}\tp \xi}+\frac{1}{(1+\frac
r{\sqrt{n}}\hat{a_i}\tp \xi)^2}}-2\frac{r}{\sqrt{n}}\sum_{i=1}^m (\hat{a}_i\tp \xi)^3. 
\end{aligned}
\end{equation*}
We notice that $\sum_{i=1}^m \hat{a}_i\hat{a}_i\tp =G(x)^{-\frac12}H(x)G(x)^{-\frac12}\preceq I_n$, thus we can apply Lemma \ref{lem_Gauss_Exs_bound}. Let $P_1(\xi)=\sum_{i=1}^m (\hat{a}_i\tp \xi)^3$, and $P_2(\xi)=\sum_{i=1}^m(\hat{a}_i\tp \xi)^4$, then we define the following event:
\begin{equation*}
    \cE_0\defn\braces{\xi\bigg|\abss{P_1(\xi)}\geq t\sqrt{15n},\quad \abss{P_2(\xi)}\geq t\sqrt{105n^2}}.
\end{equation*}
Then we can apply the concentration inequality of Gaussian polynomials Lemma \ref{lem_Gauss_concentration}, for any $t>\parenth{2e}^2$, we have 
\begin{equation}
\begin{aligned}
\prob_{\xi}(\cE_0)\leq \exp\parenth{-\frac{3}{2e}t^{\frac23}}+\exp\parenth{-\frac{4}{2e}t^{\frac12}}\leq \epsilon,
\end{aligned}
\end{equation}
where the inequality follows since we set $t=\max\braces{(2e)^2,\parenth{\frac{e}{2}\log(\frac{2}{\epsilon})}^2, \parenth{\frac{2e}{3}\log\parenth{\frac2\epsilon}}^{\frac32}}$. Moreover, for any $\xi\in\cE_0$, we have for any integers $i\in[m]$:
\begin{equation*}
    \abss{\frac{r}{\sqrt{n}}\hat{a}_i\tp \xi}\leq \frac{r}{\sqrt{n}}\brackets{\sum_{j=1}^m(\hat{a}_j\tp \xi)^4}^{1/4}\leq 
    \frac{r}{\sqrt{n}}\parenth{t\sqrt{105n^2}}^{\frac{1}{4}}\leq rt^{\frac14}\cdot (105)^{\frac18}<\frac12
\end{equation*}
The last inequality holds since we can set $r\leq \frac1{2t^{\frac14}(105)^{\frac18}}$.
In conclusion, with probability greater than $1-\epsilon$, we have 
\begin{equation*}
\begin{aligned}
\abss{\RN{1}}&\leq \frac{r^2}{n}\sum_{i=1}^m (\hat{a}_i\tp \xi)^4 \abss{\frac{2}{1+\frac
r{\sqrt{n}}\hat{a_i}\tp \xi}+\frac{1}{(1+\frac
r{\sqrt{n}}\hat{a_i}\tp \xi)^2}}+2\frac{r}{\sqrt{n}}\abss{\sum_{i=1}^m (\hat{a}_i\tp \xi)^3}\\
&\leq \frac{r^2}n \sum_{i=1}^m 8\cdot (\hat{a}_i\tp \xi)^4+2\frac{r}{\sqrt{n}}\abss{\sum_{i=1}^m(\hat{a}_i\tp\xi)^3}\\
&\leq 8\frac{r^2}n t\sqrt{105n^2}+2\frac{r}{\sqrt n}t\sqrt{15n}\leq 8\sqrt{105}tr^2+2\sqrt{15}tr \leq 2\epsilon,
\end{aligned}
\end{equation*}
where the last inequality holds as long as we set $r\leq \min\braces{1,\frac{1\epsilon}{4\sqrt{105}+\sqrt{15}}\cdot\frac1t}$.
\end{proof}

\section{Warm Initialization \& Per-step Complexity}
In this section, we talk about the computational complexity for each step of Markov transition, and we also design a feasible warm start.

\subsection{Algebraic Complexity of Each Iteration}\label{appendix_per_step_complexity}
For per-step complexity, we are mainly interested in the algebraic complexity of each step. In other words, we assume we can do exact addition, subtraction, multiplication, division over $\real$. Since we need to compute the decomposition $G(x)=Q\tp Q$, we also assume we can compute the exact square root $\sqrt{x}$ for all $x\in\real^+$. In our model, each of the five arithmetic operations $\{+,-,\times,\div,\sqrt{\quad}\}$ has a unit cost $O(1)$.

We first list a simple lemma, arguing that drawing from a uniform ellipsoid $E(x,G(x),r)$ can be reduced to drawing from the unit ball and computing the decomposition $G(x)=Q\tp Q$.

\begin{lemma}\label{lem_transformed_density}
Assume $\xi\sim \Normal(0,I_n)$ is drawn from the standard Gaussian distribution, fix any invertible matrix $Q\in\real^{n\times n}$ satisfying $G(x)\defn Q\tp Q$, the new random vector $Z$ defined by $Z\defn x+\frac{r}{\sqrt{n}}Q^{-1}\xi$ satisfies $Z\sim \Normal(x,\frac{r^2}{n}G(x)^{-1})$. 
\end{lemma}

For each iteration in Algorithm \ref{algo_main}, given current state $x$, using the result in Lemma \ref{lem_transformed_density}, we need to do the two following steps:
\begin{enumerate}
    \item draw $\xi\sim \Normal(0,I_n) $, compute $z=x+\frac{r}{\sqrt{n}}Q^{-1}\xi$, where $Q\in\real^{n\times n}$ is any invertible matrix satisfying $G(x)=Q\tp Q$.
    \item given $z$, compute $\det G(z)$.
\end{enumerate}

It costs $O(n)$ to draw $\xi$ since each component can be drawn i.i.d from one-dimensional standard Gaussian distribution. It is worth mentioning that when analyzing the mixing times of our Markov Chains, we just set $Q\defn G(x)^{\frac12}$, this is a legitimate assignment due to the uniqueness of the square root of $G(x)$. i.e., for any PSD \& symmetric matrix $C$, there exists a unique PSD \& symmetric matrix $B$ such that $B^2=C$.  However, using only $\{+,-,\times,\div,\sqrt{\quad}\}$ over (nonnegative) real numbers, we may not compute $G(x)^{\frac12}$ exactly. Because this involves computing the eigenvalues of $G(x)$ and the corresponding eigenvectors, and further needs us to exactly solve a polynomial equation of order $n$ over $\real$, which can not be done using these basic arithmetic operations.

We may write $G(x)\defn c A_x\tp W_x A_x+\lambda I$, where $c>0$ is some scalar irrelevant to $x$ and $W_x\in\real^{m\times m}$ is a diagonal matrix changing with $x$. For the \softlog $G$ as defined in Definition \ref{def_logarithmic}, $c\equiv 1$ and $W_x\equiv I_m$. For the \reglewis $G$ in Definition \ref{def_Lewis}, $c\defn c_1\sqrt{n}(\log m)^{c_2}$ and $W_x$ is the ridge-Lewis weights as defined in Eq.~\eqref{eq_Lewis_weights_def}. For now we assume that $w_x$ is known exactly for each $x$, thus $c$ and $W_x$ are known exactly (Later we would discuss a high-precision solver of $w_x$). 

We first discuss how to compute $G(x)$ efficiently. We can compute $S_x\defn \Diag(Ax-b)$ in $O(mn)$ arithmetic operations, then we can compute $A_x\defn S_x^{-1}A$ in $O(mn)$ arithmetic operations since $S_x$ is diagonal. We then attempt to compute $A_x\tp A_x$ using fast matrix multiplication, if $m\leq n$, $A_x\tp A_x$ can be computed in $O(n^\omega)$ operations by filling $A_x$ into a $n\times n$ matrix with $0$-entries. Otherwise $A_x\tp A_x$ can be computed in $O(mn^{\omega-1})$ by partitioning $A_x$ into $\floors{\frac{m}{n}}$ square matrices. Now we get $G(x)$, and we came across $O(\max\braces{m,n}n^{\omega-1})$ arithmetic operations. 

Given $G(x)$ as an $n\times n$ matrix, we need to compute an invertible matrix $Q$ such that $G(x)=Q\tp Q$. Since $G(x)$ is symmetric,  we first compute an invertible matrix $V\in\real^{n\times n}$ such that $\Lambda=VG(x)V\tp$ for some diagonal matrix $\Lambda$, and this can be done in $O(n^\omega)$ (see Chapter 16.8 about \textbf{orthogonal basis transform} in \cite{burgisser_algebraic_1997}). Then $Q\defn \sqrt{\Lambda}\parenth{V\tp}^{-1}$ is the desired invertible matrix, where $Q$ is invertible since $G(x)$ is positive definite, thus $\Lambda\succ 0$. Finally, $Q$ can be computed in $O(n^\omega)$, because computing the square root $\sqrt{\Lambda}$ takes $O(n)$ square-root operations, and computing the inverse of $V\tp \in\real^{n\times n}$ takes $O(n^\omega)$ operations (see Chapter 16.4 about \textbf{matrix inversion} in \cite{burgisser_algebraic_1997}). 

Finally, given $G(x)\in\real^{n\times n}$, we can compute its determinant $\det G(x)$ in $O(n^\omega)$ operations (see Chapter 16.4 about \textbf{determinant} in \cite{burgisser_algebraic_1997}). Combining all these arithmetic operations together, given that $W_x$ and $E_x$ is known exactly, which is the case for \softlog $G$, the per-step arithmetic cost is $O\parenth{\max\braces{m,n}n^{\omega-1}}$. 

For the \reglewis, \cite{2022Fazel-high-precision} proposed a quasi-Newton algorithm to compute an $\epsilon$-approximation of the Lewis weights $w_x$ in Eq.\eqref{eq_Lewis_weights_def} using $\textbf{polylog}\parenth{\frac{1}{\epsilon}}$ steps of gradient descent, where each descent involves computing leverage scores that costs $O(mn^{\omega-1})$. So the per-step arithmetic cost for regularized Dikin walk using \reglewis is $\Ot(\max\braces{m,n}n^{\omega-1})$ if we ignore logarithmic factors.

\subsection{Uniform Ball as Warm Initialization}\label{sec_disc_warm_start}
In this section, we discuss how to compute a ball $\ball(x_0,r_0)$ such that its uniform distribution is a suitable warm start for the truncated  distribution defined in Eq.~\eqref{eq_distri}. In addition to the assumptions in Eq.~\eqref{eq_distri}, we further assume that $K$ is bounded in a ball of radius $\outr$, and $K$ contains a ball of radius $\inr$. We list the warmness bound as Lemma \ref{lem_warmness}.

\begin{lemma}\label{lem_warmness}
Let $\Pi$ be a  distribution on $\real^n$ with density $\pi(x)\propto \mathbf{1}_K(x) e^{-f(x)}$, where $K$ is a polytope with $m$ linear constraints, and $f$ is $\alpha$-convex and $\beta$-smooth with condition number $\kappa = \beta/\alpha$, as in Eq.~\eqref{eq_distri}.

If we further assume that there exist two balls  such that $\ball(x_1,\inr)\subseteq K \subseteq \ball(x_2,\outr)$, then there exists a ball $\ball(x_0,r_0)$ such that its uniform distribution is a $M$-warm with respect to $\Pi$, where $M$ satisfies
\begin{equation*}
\log M \leq 1 +n\log\frac{3\outr}{\inr}+n\cdot\max\braces{\frac12\log\parenth{{\beta\outr^2}},\log\parenth{{2\beta\outr\vecnorm{x^\dag-x^\star}2}}}.
\end{equation*}
The radius $r_0$ and center $x_0$ of the ball can be computed by 
\begin{equation*}
    r_0\defn \frac{r_1\inr}{\vecnorm{x_1-x^\dag}2+\inr}\,,\quad x_0\defn x^\dag+ \frac{r_1}{\inr+\vecnorm{x_1-x^\dag}2}(x_1-x^\dag),
\end{equation*}
where $x^\dag\defn \arg\min_{K} f(x)$, $x^\star\defn \arg\min_{\real^n}f(x)$, and $r_1\defn \min\braces{\frac1{\sqrt{\beta}},\frac{1}{2\beta\vecnorm{x^\dag-x^\star}{2}}}$.
\end{lemma}
\begin{proof}
 For convenience, we use $B$ to denote $\ball(x_0,r_0)$ in this section. Our idea is to make $B$ close to the mode $x^\dag={\arg\min}_K f(x)$ within the polytope. We first determine a radius $r_1$, such that inside $\ball(x^\dag,r_1)$, the function $e^{-f}$ shrinks no more than a constant factor $C\geq 1$. In other words, for any $x\in \ball(x^\dag,r_1)$, we require:

\begin{equation*}
\exp\brackets{-f(x)}\geq\frac1C \exp\brackets{-f(x^\dag)},  
\end{equation*}
for ease of computation, we choose $C$ to be $e$ here. We then require $B$ to be contained in $\ball(x^\dag,r_1)\cap K$. Under this requirement, the $M$-warmness of the uniform distribution $\mu_0$ over $B\defn\ball(x_0,r_0)$ can be bounded by:
\begin{equation}\label{eq_warmness_bound1}
\begin{aligned}
\frac{\mu_0(U)}{\Pi(U)}&\leq \frac{\mu_0(U)}{\Pi(U\cap B)}=\frac{\vol(U\cap B)}{\vol(B)}\cdot \frac{\int_K e^{-f(z)}dz}{\int_{K\cap U\cap B}e^{-f(z)}dz}\\
&\overset{(i)}=\frac{\vol(U\cap B)}{\vol(B)}\cdot \frac{\int_K e^{-f(z)}dz}{\int_{U\cap B}e^{-f(z)}dz}\\
&\overset{(ii)}\leq\frac{\vol(U\cap B)}{\vol(B)}\cdot \frac{e\vol(K)}{\int_{U\cap B}e^{-f(z)}dz}\frac{\int_{U\cap B}e^{-f(z)}dz}{\vol(U\cap B)}\\
&\leq \frac{e\vol(K)}{\vol(B)}\leq \frac{e\outr^n}{r_0^n},
\end{aligned}
\end{equation}
where equality $(i)$ holds since we assumed $B\subseteq K$ and inequality $(ii)$ results from our assumption $e^{-f(x)}\geq \frac{1}{e} \cdot e^{-f(x^\dag)}$ for $x\in B \subseteq \ball(x^\dag,r_1)$, thus the relationship also holds for the average over $U\cap B$:
\begin{equation*}
\frac{\int_{U\cap B}e^{-f(z)dz}}{\vol(U\cap B)}\geq \frac{1}{e} \cdot e^{-f(x^\dag)}\geq \frac1e\frac{\int_K e^{-f(z)}dz}{\vol(K)}.
\end{equation*}
We now try to determine $r_1$ so that inside the ball $\ball(x^\dag,r_1)$, the function $e^{-f}$ shrinks less than a factor of $e$, and this translates to $f$ increases less than $1$. Due to the $\beta$-smoothness of the function $f$, the change in $f$ can be controlled by:
\begin{equation}\label{eq_increase_of_f}
f(y)-f(x^\dag)\leq \vecnorm{\nabla f(x^\dag)}2\vecnorm{y-x^\dag}2+\frac{\beta}{2}\vecnorm{y-x^\dag}2^2,
\end{equation}
thus we further control the gradient at $x^\dag$ by mean value theorem, for some $t\in (0,1)$, we have 

\begin{equation*}
    \begin{aligned}
    \vecnorm{\nabla f(x^\dag)}2&\overset{(i)}=\vecnorm{\nabla f(x^\dag)-\nabla f(x^\star)}2\leq \vecnorm{\nabla^2 f(x^\star+t(x^\dag-x^\star))}2\vecnorm{x^\star-x^\dag}2\leq \beta \vecnorm{x^\star-x^\dag}2,
    \end{aligned}
\end{equation*}
where equality $(i)$ follows from $\nabla f=0$ at the global mode $x^\star$. As a result, we can set $r_1$  to be:
\begin{equation}
r_1\defn \min\braces{\sqrt{\frac{1}{\beta}},\frac{1}{2\beta\vecnorm{x^\dag-x^\star}2}},    
\end{equation}
insert $r_1$ into Eq.~\eqref{eq_increase_of_f}, we easily check for any $y\in \ball(x^\dag,r_1)$, we have 
\begin{equation*}
  f(y)-f(x^\dag)\leq \beta \vecnorm{x^\star-x^\dag}2r_1+\frac{\beta}{2}r_1^2\leq 1,  
\end{equation*}
It is worth noting when the mode within the polytope $K$ coincides with the global mode, we have $\vecnorm{x^\dag-x^\star}2=0$, thus we have  $\frac{1}{2\beta\vecnorm{x^\dag-x^\star}2}=\infty$ and $r_1=\sqrt{\frac{1}{\beta}}$.  $f(y)-f(x^\dag)\leq 1$ for $y\in \ball(x^\dag,r_1)$ still holds because $\nabla f(x^\dag)=0$ in Eq.~\eqref{eq_increase_of_f}. 

After getting the ball $\ball(x^\dag,r_1)$, we require our  initial distribution $B=\ball(x_0,r_0)$ to be inside $\ball(x^\dag,r_1)$ so that the function $f$ only increases a constant. Moreover, we also want to make sure $\ball(x_0,r_0)$ is inside the polytope. This can be ensured by our assumption that $K$ contains a ball of radius $\inr$. 

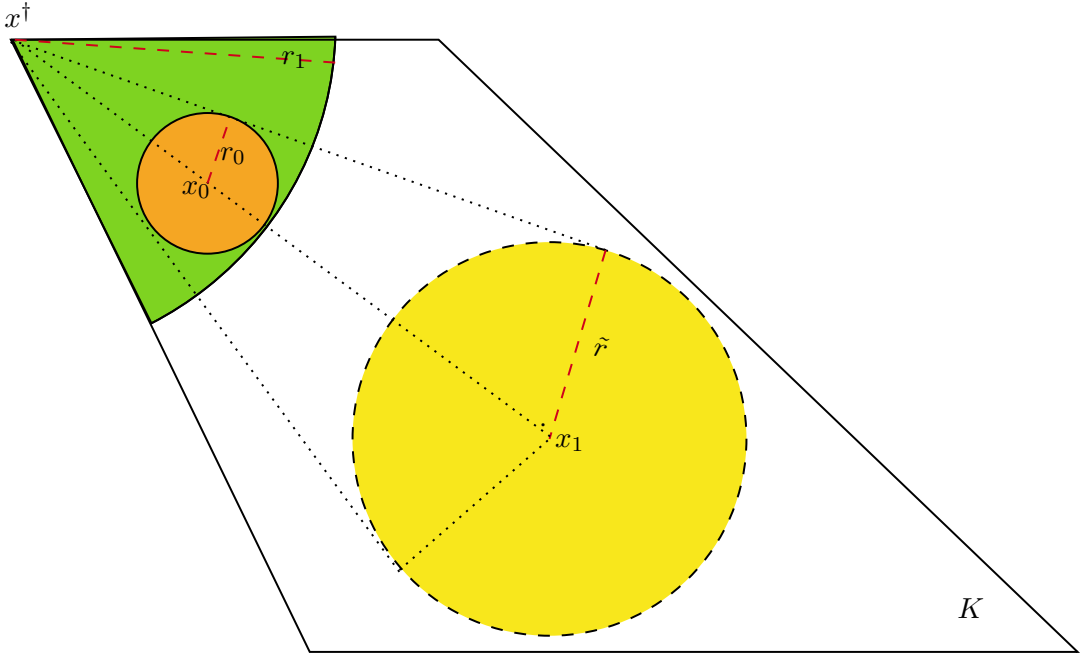
\begin{figure}
\tikzset{every picture/.style={line width=0.75pt}} 
\begin{tikzpicture}[x=0.75pt,y=0.75pt,yscale=-1,xscale=1]
\draw  [draw opacity=0.2][fill={rgb, 255:red, 126; green, 211; blue, 33 }  ,fill opacity=0.24 ] (181.66,57.46) .. controls (180.11,104.8) and (157.73,151.92) .. (116.56,184.3) .. controls (107.75,191.23) and (98.49,197.16) .. (88.92,202.12) -- (18,59) -- cycle ; \draw   (181.66,57.46) .. controls (180.11,104.8) and (157.73,151.92) .. (116.56,184.3) .. controls (107.75,191.23) and (98.49,197.16) .. (88.92,202.12) ;  
\draw  [fill={rgb, 255:red, 245; green, 166; blue, 35 }  ,fill opacity=1 ] (81.67,131.5) .. controls (81.67,111.89) and (97.56,96) .. (117.17,96) .. controls (136.77,96) and (152.67,111.89) .. (152.67,131.5) .. controls (152.67,151.11) and (136.77,167) .. (117.17,167) .. controls (97.56,167) and (81.67,151.11) .. (81.67,131.5) -- cycle ;
\draw   (168.74,368) -- (18.84,59) -- (233.75,59) -- (556,368) -- cycle ;
\draw  [fill={rgb, 255:red, 248; green, 231; blue, 28 }  ,fill opacity=0.66 ][dash pattern={on 4.5pt off 4.5pt}] (190.33,260.51) .. controls (190.33,205.65) and (234.8,161.17) .. (289.66,161.17) .. controls (344.53,161.17) and (389,205.65) .. (389,260.51) .. controls (389,315.37) and (344.53,359.84) .. (289.66,359.84) .. controls (234.8,359.84) and (190.33,315.37) .. (190.33,260.51) -- cycle ;
\draw  [dash pattern={on 0.84pt off 2.51pt}]  (18,60) -- (214,327.67) ;
\draw  [dash pattern={on 0.84pt off 2.51pt}]  (18,59) -- (318,165.17) ;
\draw [color={rgb, 255:red, 208; green, 2; blue, 27 }  ,draw opacity=1 ] [dash pattern={on 4.5pt off 4.5pt}]  (318,165.17) -- (289.66,260.51) ;
\draw  [dash pattern={on 0.84pt off 2.51pt}]  (214,326.67) -- (289.66,260.51) ;
\draw  [dash pattern={on 0.84pt off 2.51pt}]  (20,59) -- (291.66,260.51) ;
\draw [color={rgb, 255:red, 208; green, 2; blue, 27 }  ,draw opacity=1 ] [dash pattern={on 4.5pt off 4.5pt}]  (117.17,131.5) -- (128.2,99.4) ;
\draw [color={rgb, 255:red, 208; green, 2; blue, 27 }  ,draw opacity=1 ] [dash pattern={on 4.5pt off 4.5pt}]  (20,59) -- (182.2,70.6) ;

\draw (283,248.9) node [anchor=north west][inner sep=0.75pt]    {$\cdot $};
\draw (312.3,206.64) node [anchor=north west][inner sep=0.75pt]  [rotate=-12.91]  {$\tilde{r}$};
\draw (13.33,38.4) node [anchor=north west][inner sep=0.75pt]    {$x^{\dagger}$};
\draw (291,256.73) node [anchor=north west][inner sep=0.75pt]    {$x_{1}$};
\draw (102.67,129.07) node [anchor=north west][inner sep=0.75pt]    {$x_{0}$};
\draw (122,111.73) node [anchor=north west][inner sep=0.75pt]    {$r_{0}$};
\draw (153,63.4) node [anchor=north west][inner sep=0.75pt]    {$r_{1}$};
\draw (494,339.4) node [anchor=north west][inner sep=0.75pt]    {$K$};

\end{tikzpicture}
    \caption{An example of warm-start $\ball(x_0,r_0)$ for $K\subseteq \real^2$: Here the mode within the polytope $x^\dag\defn \arg\min_K f(x)$ coincides the upper-left vertex of $K$. We need to ensure $\ball(x_0,r_0)\subseteq \ball(x^\dag,r_1)$ and $\ball(x_0,r_0)\subseteq K$, where the first condition reduces to $\vecnorm{x_0-x^\dag}2+r_0\leq r_1$, and the later is ensured by $\ball(x_0,r_0)$ being inside the convex hull of $\{x^\dag\}\cup\ball(x_1,\inr)$, and $r_0$ can be easily computed by similarity of cones.
    }
    \label{fig_warm_start}
\end{figure}

These two conditions can be ensured by the following two requirements by simple geometry relations (also see Figure \ref{fig_warm_start} for illustration):
\begin{equation*}
\begin{aligned}
&\vecnorm{x_0-x^\dag}2+r_0\leq r_1 \text{ , and }\frac{r_0}{\vecnorm{x^\dag-x_0}2}=\frac{\inr}{\vecnorm{x_1-x^\dag}2}.
\end{aligned}
\end{equation*}
where the first inequality ensures that $\ball(x_0,r_0)\subseteq \ball (x^\dag,r_1)$, and the second inequality ensures that $\ball(x_0,r_0)$ is included in the convex hull of ${x^\dag}\cup \ball(x_1,\inr)$, thus $\ball(x_0,r_0)\subseteq K$. From these two inequalities, it is easy to get the following largest $r_0$ and the corresponding center $x_0$:
\begin{equation*}
r_0\defn \frac{r_1\inr}{\vecnorm{x_1-x^\dag}2+\inr}\,,\quad x_0\defn x^\dag+ \frac{r_1}{\inr+\vecnorm{x_1-x^\dag}2}(x_1-x^\dag).
\end{equation*}
The $M$-warmness of our initial distribution (over $\ball(x_0,r_0)$) can be bounded by inserting $r_1$ into Eq.~\eqref{eq_warmness_bound1}:
\begin{equation*}
\begin{aligned}
M&\leq \frac{e\outr^n}{r_0^n}= \frac{e\outr^n\parenth{\vecnorm{x_1-x^\dag}2+\inr}^n}{r_1^n\inr^n}
\overset{(i)}\leq  \frac{e\outr^n\parenth{3\outr}^n}{r_1^n\inr^n}\\
&=  e\parenth{\frac{{3\outr}}{\inr}}^n\max\braces{\parenth{\beta\outr^2}^{n/2},\parenth{{2\beta\vecnorm{x^\dag-x^\star}2\outr}}^n},
\end{aligned}
\end{equation*}
where inequality $(i)$ holds since we $x_1,x^\dag\in K$ and $\text{diam}(K)\leq 2\outr$. Since our sampling algorithm has high accuracy, the mixing time depends on $\log M$:
\begin{equation*}
\log M = 1 +n\log\frac{3\outr}{\inr}+n\cdot\max\braces{\frac12\log\parenth{{\beta\outr^2}},\log\parenth{{2\beta\outr\vecnorm{x^\dag-x^\star}2}}}.
\end{equation*}
\end{proof}

\section{Bounding \texorpdfstring{$\bar\nu$}{nu}-symmetric metrics by cross-ratio distance}\label{sec_cross_ratio_by_local}
In this section, we provide the proof of Fact \ref{fact_ssc_col} for the sake of rigor and completeness. 
\begin{proof}[Proof of Fact~\ref{fact_cross_ratio_by_local}]

The proof is similar to steps in \cite{laddha2020strong}, but we notice that the assumption $\vecnorm{p-x}2\leq\vecnorm{y-q}2$ in \cite{laddha2020strong} may not hold, and when the other side hold, we can only derive the conclusion of local metric at $y$ instead of at $x$. Moreover, since we are dealing with possibly unbounded polytopes, we also consider the case when $p$ or $q$ is at infinity. 

Throughout the proof, we can assume that both $\vecnorm{x-y}{H(x)}$ and $\vecnorm{x-y}{H(y)}$ are positive. Otherwise, we have at least one of them to be $0$, then the RHS of Eq.~\eqref{eq_cross_ratio_by_local} is $0$ so the inequality becomes trivial.

Assume we have the chord $[p,q]$ induced by the line $\widebar{xy}$ in the order $p,x,y,q$.

\begin{itemize}
\item First, we consider the case when both $p,q$ are bounded. 
When $\vecnorm{p-x}2\leq \vecnorm{y-q}2$, then we have $\vecnorm{p-x}2\leq \vecnorm{x-q}2$, thus $x+x-p\in K$, so we have  $p\in K\cap (2x-K)$. According to weak $\bar\nu-$symmetry of $H$, we have 
\begin{equation*}
\vecnorm{p-x}{H(x)}\leq \sqrt{\bar\nu}.
\end{equation*}
So we have  
\begin{equation*}
\begin{aligned}
d_K(x,y)&=\frac{\vecnorm{x-y}2\vecnorm{p-q}2}{\vecnorm{p-x}2\vecnorm{q-y}2}\geq \frac{\vecnorm{x-y}2}{\vecnorm{p-x}2}\\
&=\frac{\vecnorm{x-y}{H(x)}}{\vecnorm{p-x}{H(x)}}\geq \frac{\vecnorm{x-y}{H(x)}}{\sqrt{\bar\nu}}.
\end{aligned}
\end{equation*}
Next we consider the other case $\vecnorm{y-q}2<\vecnorm{p-x}2$, then following the same derivation, we have 

\begin{equation*}
\begin{aligned}
d_K(x,y)&=\frac{\vecnorm{x-y}2\vecnorm{p-q}2}{\vecnorm{p-x}2\vecnorm{q-y}2}\geq \frac{\vecnorm{x-y}2}{\vecnorm{q-y}2}\\
&=\frac{\vecnorm{x-y}{H(y)}}{\vecnorm{q-y}{H(y)}}\geq \frac{\vecnorm{x-y}{H(y)}}{\sqrt{\bar\nu}}.
\end{aligned}
\end{equation*}
Combining the two results, and we come to the conclusion that in any cases, we have 
\begin{equation*}
d_K(x,y)\geq \frac1{\sqrt{\bar\nu}}{\min\braces{\vecnorm{x-y}{H(x)},\vecnorm{x-y}{H(y)}}}.
\end{equation*}
\item
Second, we consider the case when one and only one of $\{p,q\}$ is at infinity. Without loss of generality, we assume $q=\infty$, then we deduce that $x+t(y-x)\in K$ for all $t\geq 0$. Since $x-p$ and $y-x$ are in the same direction, thus $x+(x-p)\in K$, so we have $p\in K\cap (2x-K)$.

According to the definition of $\bar\nu-$symmetry, $\vecnorm{p-x}{H(x)}\leq \sqrt{\bar\nu}$, then we have 
\begin{equation*}
\begin{aligned}
 d_K(x,y)&=\frac{\vecnorm{x-y}2}{\vecnorm{p-x}2}=\frac{\vecnorm{x-y}{H(x)}}{\vecnorm{p-x}{H(x)}}\geq \frac{\vecnorm{x-y}{H(x)}}{\sqrt{\bar\nu}}\\
&\geq \frac{1}{\sqrt{\bar\nu}}\min\braces{\vecnorm{x-y}{H(x)},\vecnorm{x-y}{H(y)}}.
\end{aligned}
\end{equation*}
\item 
Finally, we consider the case that both $p,q$ are at infinity, then $d_K(x,y)=0$ by the definition of cross-ratio distance (Definition \ref{def_cross_ratio_unbounded}). In order to prove Eq.~\eqref{eq_cross_ratio_by_local}, it is adequate to prove $\vecnorm{x-y}{H(x)}=0$. 

Since both $p,q$ are at infinity, it implies that for all $t\in\real$,  $tx+(1-t)y\in K$. By inserting $t\defn 2-l$ for $l\in\real$, it is easy to check that $tx+(1-t)y\in K\cap(2x-K)$ for $t\in\real$. Then due to our assumption of $\bar\nu$-symmetry of $H$, for all $t\in\real$ we have 
\begin{equation*}
   \abss{t-1}\cdot\vecnorm{y-x}{H(x)}= \vecnorm{tx+(1-t)y-x}{H(x)}\leq \sqrt{\bar\nu}.
\end{equation*}
Due to the arbitrariness of $t$, we let $\abss{t-1}\to +\infty$, since $\bar\nu$ is finite, so we have  $\vecnorm{y-x}{H(x)}=0$. 
\end{itemize}
\end{proof}

\section{Isoperimetry for Weakly Logconcave Measures}

In this section, we complete the omitted pieces needed for extending the new isoperimetry to weakly logconcave measures in Section \ref{sec_new_iso_weakly}. In Appendix \ref{appendix_reduction_to_isotropy}, we show why we can assume the distribution is isotropic by an affine transformation. In Appendix \ref{appendix_homeomorphism}, we show that the combined metric of Euclidean and Hilbert metric actually induces Euclidean topology (Lemma \ref{lem_homeomorphism}). 

\subsection{Reduction to Isotropic Distributions}\label{appendix_reduction_to_isotropy}
Let $\mu_\pi,\Sigma_\pi$ denote the mean and the covariance matrix of the log-concave probability distribution $\Pi$ defined in Equation \eqref{eq_distri}. Now suppose the isoperimetric inequality in Lemma \ref{lem_isoperimetric_weakly} holds for isotropic log-concave distributions ($\Sigma_\pi=I_n$), we show that this implies Lemma \ref{lem_isoperimetric_weakly} for general covariance matrix $\Sigma_\pi$. 

Since we assume $K$ is open, thus $\Sigma_\pi$ is an invertible matrix. We define the following bijective affine map $\cA: \real^n\to\real^n$ by $\cA(x)\defn \Sigma_\pi^{-\frac12}(x-\mu_\pi)$. Now we let $\Pi^{\cA}$ denotes induced measure of $\Pi$ under the mapping $\cA$. In other words, for any Borel set $B$ in $K$, we define
\begin{equation}
    \Pi^{\cA}(B)\defn \Pi(\cA^{-1}(B)).
\end{equation}
It is easy to verify that the induced measure is  logconcave with covariance matrix $I_n$. So for any measurable decomposition $K=S_1\sqcup S_2\sqcup S_3$, use the isoperimetric inequality for isotropic and logconcave distributions, we have 
\begin{equation}\label{eq_log_concave_weakly_isotropic}
    \Pi^{\cA}\parenth{\cA(S_3)}\geq  c_n \cdot\dn_\cA(\cA(S_1),\cA(S_2))\min\braces{ \Pi^{\cA}(\cA(S_1)), \Pi^{\cA}(\cA(S_2))}\\
\end{equation}
where $c_n$ denotes the constant $c_n\defn \parenth{6\max\braces{1,\psi_n}}^{-1}$, and $\dn_\cA$ is the mixed metric defined by:
\begin{equation*}
    \dn_\cA(x,y)\defn \max\braces{\parenth{\log{2}}\vecnorm{y-x}2, \dH_{\cA(K)}(x,y)}.
\end{equation*}
Thus Eq.~\eqref{eq_log_concave_weakly_isotropic} translates to
\begin{equation}
\begin{aligned}
\Pi(S_3)&\geq \underset{(x,y)\in S_1\times S_2}\inf c_n\cdot\dn_\cA\parenth{\cA x,\cA y}\min\braces{\Pi(S_1),\Pi(S_2)}\\
&\overset{(i)}= \underset{(x,y)\in S_1\times S_2}\inf c_n\cdot\max\braces{(\log2)\vecnorm{y-x}{\Sigma_{\pi}^{-1}},\dH_{\cA(K)}(\cA x,\cA y)}\min\braces{\Pi(S_1),\Pi(S_2)}\\
&\overset{(ii)}=\underset{(x,y)\in S_1\times S_2} \inf c_n\cdot\max\braces{(\log2)\vecnorm{y-x}{\Sigma_{\pi}^{-1}},\dH_K(x,y)}\min\braces{\Pi(S_1),\Pi(S_2)}\\
&\overset{(iii)}\geq  \underset{(x,y)\in S_1\times S_2}\inf c_n\cdot\max\braces{(\log2)\frac{\vecnorm{y-x}{2}}{\sqrt{\spectr}},\dH_K(x,y)}\min\braces{\Pi(S_1),\Pi(S_2)},
\end{aligned}
\end{equation}
where inequality $(i)$ holds since we have $\vecnorm{\cA x-\cA y}{2}=\vecnorm{y-x}{\Sigma_\pi^{-1}}$ by definition of $\cA$, and inequality $(ii)$ holds due to the affine invariance of Hilbert metric, and inequality $(iii)$ holds since $\Sigma_{\pi}\preceq \spectr I_n$.

\subsection{The Mixed Metric induces Euclidean Topology}\label{appendix_homeomorphism}
\begin{proof}[Proof of Lemma \ref{lem_homeomorphism}]

Fix any $x,y\in K$, it is clear that  $\dn(x,y)\geq 0$. If $\dn(x,y)=0$, then $\vecnorm{y-x}{2}=0$, so we have $x=y$. It is also clear that $\dn(x,y)=\dn(y,x)$ by definition. The remaining condition to ensure $\dn$ to be a metric is 
the triangle inequality. Fix $x,y,z\in K$, we first prove the triangle inequality for Hilbert metric $\dH_K$, which is well-known for bounded open set $K$, see for example \cite{niblo_hiberts_1993}. For unbounded $K$, let $\ball_r$ denotes the ball in $\real^n$ centered around $0$ with radius $r$, then we can prove it by a simple limit argument:
\begin{equation*}
\begin{aligned}
\dH_K(x,y)&=\log\parenth{1+d_K(x,y)}\overset{(i)}=\lim_{r\to+\infty} \log\parenth{1+d_{K\cap\ball_r}(x,y)}\\
&= \lim_{r\to+\infty} \dH_{K\cap \ball_r}(x,y)
\overset{(ii)}\leq \lim_{r\to+\infty} \brackets{\dH_{K\cap \ball_r}(x,z)+ \dH_{K\cap \ball_r}(z,y)}\\
&=\dH_K(x,z)+\dH_K(z,y)
\end{aligned}
\end{equation*}
where equality $(i)$ holds due to our definition of cross-ratio distance for unbounded convex sets (see Definition \ref{def_cross_ratio_unbounded}), inequality $(ii)$ holds by applying the triangle inequality for bounded open sets. Thus the triangle inequality for $\dn$ can be proved:
\begin{equation*}
\begin{aligned}
   \dn(x,y)&\leq \max\braces{\vecnorm{x-z}{2}+\vecnorm{z-y}{2},\dH_K(x,z)+\dH_K(y,z)}\\
   &\leq \max\braces{\vecnorm{x-z}{2},\dH_K(x,z)}+\max\braces{\vecnorm{y-z}{2},\dH_K(y,z)}\\
   &=\dn(x,z)+\dn(z,y).
\end{aligned}
\end{equation*}
To prove that $\tilde{d}$ induces the same topology as the Euclidean distance $\vecnorm{\cdot}{2}$, we notice that the identity map $x\mapsto x$ from 
$(K,\dn)$ to $(K,\vecnorm{\cdot}{2})$ is clearly continuous, and we proceed to prove its inverse is also continuous.

Fix $x\in K$, then there exists $r_x>0$ such that $\ball(x,r_x)\subseteq K$ by the openness of $K$. Without loss of generality, we can always assume $y$ such that $\vecnorm{y-x}{2}<r_x$ Assume $\widebar{xy}$ intersects $\partial K$ in the order $p,x,y,q$. Then we have $\vecnorm{p-x}{2}\geq r_x$ and $\vecnorm{y-q}{2}\geq r_x-\vecnorm{y-x}{2}$. So we have 
\begin{equation}\label{eq_homeomorphism}
\begin{aligned}
    \dn(x,y)\leq \vecnorm{y-x}{2}+\log\brackets{1+{\frac{\vecnorm{y-x}{2}}{r_x}+\frac{\vecnorm{y-x}{2}}{r_x-\vecnorm{y-x}{2}}+\frac{\vecnorm{y-x}{2}^2}{r_x(r_x-\vecnorm{y-x}{2})}}}.
\end{aligned}
\end{equation}
Treat $r_x>0$ as a fixed constant, then it is clear that when $\vecnorm{y-x}{2}\to 0$, the RHS of Eq.~\eqref{eq_homeomorphism} converges to $0$.
As a result, we proved the identity map is continuous in both direction, and we established homeomorphism between $(K,\dn)$ and $(K,\vecnorm{\cdot}{2})$. 
\end{proof}


\bibliographystyle{alpha}
\bibliography{ref}

\newcommand{\etalchar}[1]{$^{#1}$}
\begin{thebibliography}{LKPL{\etalchar{+}}11}

\bibitem[AA06]{arellanovalle_unification_2006}
Reinaldo~B. Arellano‐Valle and Adelchi Azzalini.
\newblock On the unification of families of skew‐normal distributions.
\newblock {\em Scandinavian Journal of Statistics}, 33(3):561--574, September 2006.

\bibitem[AC93]{albert_bayesian_1993}
James~H. Albert and Siddhartha Chib.
\newblock Bayesian analysis of binary and polychotomous response data.
\newblock {\em Journal of the American Statistical Association}, 88(422):669--679, June 1993.

\bibitem[AD15]{afshar_reflection_2015}
Hadi~Mohasel Afshar and Justin Domke.
\newblock Reflection, refraction, and {Hamiltonian} {Monte} {Carlo}.
\newblock In {\em Proceedings of the 28th {International} {Conference} on {Neural} {Information} {Processing} {Systems} - {Volume} 2}, {NIPS}'15, pages 3007--3015, Cambridge, MA, USA, 2015. MIT Press.
\newblock event-place: Montreal, Canada.

\bibitem[AFDZ23]{anceschi2023bayesian}
Niccol{\`o} Anceschi, Augusto Fasano, Daniele Durante, and Giacomo Zanella.
\newblock Bayesian conjugacy in probit, tobit, multinomial probit and extensions: {A} review and new results.
\newblock {\em Journal of the American Statistical Association}, 118(542):1451--1469, April 2023.

\bibitem[ALPW24]{andrieu_explicit_2024}
Christophe Andrieu, Anthony Lee, Sam Power, and Andi~Q. Wang.
\newblock Explicit convergence bounds for {Metropolis} {Markov} chains: {Isoperimetry}, spectral gaps and profiles.
\newblock {\em The Annals of Applied Probability}, 34(4), August 2024.

\bibitem[BCS97]{burgisser_algebraic_1997}
Peter B{\"u}rgisser, Michael Clausen, and Mohammad~Amin Shokrollahi.
\newblock {\em Algebraic {Complexity} {Theory}}, volume 315 of {\em Grundlehren der mathematischen {Wissenschaften}}.
\newblock Springer Berlin Heidelberg, Berlin, Heidelberg, 1997.

\bibitem[BDMP17]{brosse_sampling_2017}
Nicolas Brosse, Alain Durmus, {\'E}ric Moulines, and Marcelo Pereyra.
\newblock Sampling from a log-concave distribution with compact support with proximal {Langevin} {Monte} {Carlo}.
\newblock In Satyen Kale and Ohad Shamir, editors, {\em Proceedings of the 2017 {Conference} on {Learning} {Theory}}, volume~65 of {\em Proceedings of {Machine} {Learning} {Research}}, pages 319--342. PMLR, July 2017.

\bibitem[BEL18]{bubeck_sampling_2018}
S{\'e}bastien Bubeck, Ronen Eldan, and Joseph Lehec.
\newblock Sampling from a log-concave distribution with projected {Langevin} {Monte} {Carlo}.
\newblock {\em Discrete \& Computational Geometry}, 59(4):757--783, June 2018.

\bibitem[BH97]{bobkov_isoperimetric_1997}
S.~G. Bobkov and C.~Houdr{\'e}.
\newblock Isoperimetric constants for product probability measures.
\newblock {\em The Annals of Probability}, 25(1), January 1997.

\bibitem[BKM17]{blei2017variational}
David~M Blei, Alp Kucukelbir, and Jon~D McAuliffe.
\newblock Variational inference: A review for statisticians.
\newblock {\em Journal of the American statistical Association}, 112(518):859--877, 2017.

\bibitem[Bo15]{bubeck_convex_2015}
S{\'e}bastien Bubeck and {others}.
\newblock Convex optimization: {Algorithms} and complexity.
\newblock {\em Foundations and Trends{\textregistered} in Machine Learning}, 8(3-4):231--357, 2015.
\newblock Publisher: Now Publishers, Inc.

\bibitem[Bus82]{buser1982note}
Peter Buser.
\newblock A note on the isoperimetric constant.
\newblock In {\em Annales scientifiques de l'{\'E}cole normale sup{\'e}rieure}, volume~15, pages 213--230, 1982.

\bibitem[CDWY18]{chen2018fast}
Yuansi Chen, Raaz Dwivedi, Martin~J Wainwright, and Bin Yu.
\newblock Fast {MCMC} sampling algorithms on polytopes.
\newblock {\em The Journal of Machine Learning Research}, 19(1):2146--2231, 2018.

\bibitem[CDWY20]{chen_fast_2020}
Yuansi Chen, Raaz Dwivedi, Martin~J. Wainwright, and Bin Yu.
\newblock Fast mixing of {Metropolized} {Hamiltonian} {Monte} {Carlo}: {Benefits} of multi-step gradients.
\newblock {\em Journal of Machine Learning Research}, 21(92):1--72, 2020.

\bibitem[CE22]{chen2022hitandrun}
Yuansi Chen and Ronen Eldan.
\newblock Hit-and-run mixing via localization schemes, 2022.

\bibitem[CFPT23]{chalkis_truncated_2023}
Apostolos Chalkis, Vissarion Fisikopoulos, Marios Papachristou, and Elias Tsigaridas.
\newblock Truncated log-concave sampling for convex bodies with reflective {Hamiltonian} {Monte} {Carlo}.
\newblock {\em ACM Transactions on Mathematical Software}, 49(2):1--25, June 2023.

\bibitem[Che70]{cheeger1970lower}
Jeff Cheeger.
\newblock A lower bound for the smallest eigenvalue of the laplacian.
\newblock {\em Problems in analysis}, 625(195-199):110, 1970.

\bibitem[Che21]{chen_almost_2021}
Yuansi Chen.
\newblock An almost constant lower bound of the isoperimetric coefficient in the {KLS} conjecture.
\newblock {\em Geometric and Functional Analysis}, 31(1):34--61, February 2021.

\bibitem[CV18]{cousins_gaussian_2018}
Ben Cousins and Santosh Vempala.
\newblock Gaussian cooling and ${O}^*(n^3)$ algorithms for volume and {Gaussian} volume.
\newblock {\em SIAM Journal on Computing}, 47(3):1237--1273, January 2018.

\bibitem[DCWY19]{dwivedi2019log}
Raaz Dwivedi, Yuansi Chen, Martin~J. Wainwright, and Bin Yu.
\newblock Log-concave sampling: Metropolis-hastings algorithms are fast.
\newblock {\em Journal of Machine Learning Research}, 20(183):1--42, 2019.

\bibitem[DLH93]{niblo_hiberts_1993}
Pierre De~La~Harpe.
\newblock On {Hibert}'s {metric} for {simplices}.
\newblock In Graham~A. Niblo and Martin~A. Roller, editors, {\em Geometric {Group} {Theory}}, pages 97--119. Cambridge University Press, 1 edition, July 1993.

\bibitem[Dur19]{durante_conjugate_2019}
Daniele Durante.
\newblock Conjugate {Bayes} for probit regression via unified skew-normal distributions.
\newblock {\em Biometrika}, 106(4):765--779, December 2019.

\bibitem[Eld13]{eldan_thin_2013}
Ronen Eldan.
\newblock Thin shell implies spectral gap up to polylog via a stochastic localization scheme.
\newblock {\em Geometric and Functional Analysis}, 23(2):532--569, April 2013.

\bibitem[FD22]{fasano_class_2022}
Augusto Fasano and Daniele Durante.
\newblock A class of conjugate priors for multinomial probit models which includes the multivariate normal one.
\newblock {\em Journal of Machine Learning Research}, 23(30):1--26, 2022.

\bibitem[FLPS22]{2022Fazel-high-precision}
Maryam Fazel, Yin~Tat Lee, Swati Padmanabhan, and Aaron Sidford.
\newblock Computing lewis weights to high precision.
\newblock In {\em Proceedings of the 2022 Annual ACM-SIAM Symposium on Discrete Algorithms (SODA)}, pages 2723--2742, 2022.

\bibitem[FRDP21]{fasano_closed-form_2021}
Augusto Fasano, Giovanni Rebaudo, Daniele Durante, and Sonia Petrone.
\newblock A closed-form filter for binary time series.
\newblock {\em Statistics and Computing}, 31(4):47, July 2021.

\bibitem[Gew96]{geweke_bayesian_1996}
John~F. Geweke.
\newblock Bayesian inference for linear models subject to linear inequality constraints.
\newblock In Jack~C. Lee, Wesley~O. Johnson, and Arnold Zellner, editors, {\em Modelling and {Prediction} {Honoring} {Seymour} {Geisser}}, pages 248--263. Springer New York, New York, NY, 1996.

\bibitem[GG22]{ghosal_bayesian_2022}
Rahul Ghosal and Sujit~K. Ghosh.
\newblock Bayesian inference for generalized linear model with linear inequality constraints.
\newblock {\em Computational Statistics \& Data Analysis}, 166:107335, February 2022.

\bibitem[GKV24]{gatmiry_sampling_2024}
Khashayar Gatmiry, Jonathan Kelner, and Santosh~S. Vempala.
\newblock Sampling polytopes with {Riemannian} {HMC}: {Faster} mixing via the {Lewis} weights barrier.
\newblock In {\em Proceedings of {Thirty} {Seventh} {Conference} on {Learning} {Theory}}, volume 247 of {\em Proceedings of {Machine} {Learning} {Research}}, pages 1796--1881. PMLR, July 2024.

\bibitem[GSL92]{gelfand_bayesian_1992}
Alan~E. Gelfand, Adrian F.~M. Smith, and Tai-Ming Lee.
\newblock Bayesian analysis of constrained parameter and truncated data problems using {Gibbs} sampling.
\newblock {\em Journal of the American Statistical Association}, 87(418):523--532, June 1992.

\bibitem[HAP{\etalchar{+}}19]{heirendt_creation_2019}
Laurent Heirendt, Sylvain Arreckx, Thomas Pfau, Sebasti{\'a}n~N. Mendoza, Anne Richelle, Almut Heinken, Hulda~S. Haraldsd{\'o}ttir, Jacek Wachowiak, Sarah~M. Keating, Vanja Vlasov, Stefania Magnusd{\'o}ttir, Chiam~Yu Ng, German Preciat, Alise {\v Z}agare, Siu H.~J. Chan, Maike~K. Aurich, Catherine~M. Clancy, Jennifer Modamio, John~T. Sauls, Alberto Noronha, Aarash Bordbar, Benjamin Cousins, Diana~C. El~Assal, Luis~V. Valcarcel, I{\~n}igo Apaolaza, Susan Ghaderi, Masoud Ahookhosh, Marouen Ben~Guebila, Andrejs Kostromins, Nicolas Sompairac, Hoai~M. Le, Ding Ma, Yuekai Sun, Lin Wang, James~T. Yurkovich, Miguel A.~P. Oliveira, Phan~T. Vuong, Lemmer~P. El~Assal, Inna Kuperstein, Andrei Zinovyev, H.~Scott Hinton, William~A. Bryant, Francisco~J. Arag{\'o}n~Artacho, Francisco~J. Planes, Egils Stalidzans, Alejandro Maass, Santosh Vempala, Michael Hucka, Michael~A. Saunders, Costas~D. Maranas, Nathan~E. Lewis, Thomas Sauter, Bernhard~{\O}. Palsson, Ines Thiele, and Ronan M.~T. Fleming.
\newblock Creation and analysis of biochemical constraint-based models using the {COBRA} {Toolbox} v.3.0.
\newblock {\em Nature Protocols}, 14(3):639--702, March 2019.

\bibitem[HCT{\etalchar{+}}17]{haraldsdottir_chrr_2017}
Hulda~S Haraldsd{\'o}ttir, Ben Cousins, Ines Thiele, Ronan~M.T Fleming, and Santosh Vempala.
\newblock {CHRR}: coordinate hit-and-run with rounding for uniform sampling of constraint-based models.
\newblock {\em Bioinformatics}, 33(11):1741--1743, June 2017.

\bibitem[Jan97]{janson_gaussian_1997}
Svante Janson.
\newblock {\em Gaussian {Hilbert} {Spaces}}.
\newblock Cambridge University Press, 1 edition, June 1997.

\bibitem[JLLV21]{jia_reducing_2021}
He~Jia, Aditi Laddha, Yin~Tat Lee, and Santosh Vempala.
\newblock Reducing isotropy and volume to {KLS}: an ${O}^*(n^3\psi^2)$ volume algorithm.
\newblock In {\em Proceedings of the 53rd {Annual} {ACM} {SIGACT} {Symposium} on {Theory} of {Computing}}, pages 961--974, Virtual Italy, June 2021. ACM.

\bibitem[JLV22]{jambulapati_slightly_2022}
Arun Jambulapati, Yin~Tat Lee, and Santosh~S. Vempala.
\newblock A slightly improved bound for the {KLS} constant, October 2022.
\newblock arXiv:2208.11644 [cs, math].

\bibitem[JSPD19]{johndrow_mcmc_2019}
James~E. Johndrow, Aaron Smith, Natesh Pillai, and David~B. Dunson.
\newblock {MCMC} for imbalanced categorical data.
\newblock {\em Journal of the American Statistical Association}, 114(527):1394--1403, July 2019.

\bibitem[KL22]{klartag_bourgains_2022}
Bo'az Klartag and Joseph Lehec.
\newblock Bourgain's slicing problem and {KLS} isoperimetry up to polylog, April 2022.
\newblock arXiv:2203.15551 [math].

\bibitem[Kla23]{klartag_logarithmic_2023}
Bo'az Klartag.
\newblock Logarithmic bounds for isoperimetry and slices of convex sets, June 2023.
\newblock arXiv:2303.14938 [math].

\bibitem[KLS95]{kannan1995isoperimetric}
Ravi Kannan, L{\'a}szl{\'o} Lov{\'a}sz, and Mikl{\'o}s Simonovits.
\newblock Isoperimetric problems for convex bodies and a localization lemma.
\newblock {\em Discrete \& Computational Geometry}, 13:541--559, 1995.

\bibitem[KLS97]{Kannan1997RandomWA}
Ravi Kannan, L{\'a}szl{\'o}~Mikl{\'o}s Lov{\'a}sz, and Mikl{\'o}s Simonovits.
\newblock Random walks and an ${O}^*(n^5)$ volume algorithm for convex bodies.
\newblock {\em Random Struct. Algorithms}, 11:1--50, 1997.

\bibitem[KN12]{kannan_random_2012}
Ravindran Kannan and Hariharan Narayanan.
\newblock Random walks on polytopes and an affine interior point method for linear programming.
\newblock {\em Mathematics of Operations Research}, 37(1):1--20, February 2012.

\bibitem[KV24]{pmlr-v247-kook24b}
Yunbum Kook and Santosh~S. Vempala.
\newblock Gaussian cooling and {D}ikin walks: {T}he interior-point method for logconcave sampling.
\newblock In Shipra Agrawal and Aaron Roth, editors, {\em Proceedings of Thirty Seventh Conference on Learning Theory}, volume 247 of {\em Proceedings of Machine Learning Research}, pages 3137--3240. PMLR, 30 Jun--03 Jul 2024.

\bibitem[Led04]{ledoux2004spectral}
Michel Ledoux.
\newblock Spectral gap, logarithmic sobolev constant, and geometric bounds.
\newblock {\em Surveys in differential geometry}, 9(1):219--240, 2004.

\bibitem[LKPL{\etalchar{+}}11]{lesage_new_2011}
James~P. LeSage, R.~Kelley~Pace, Nina Lam, Richard Campanella, and Xingjian Liu.
\newblock New {Orleans} business recovery in the aftermath of hurricane {Katrina}.
\newblock {\em Journal of the Royal Statistical Society Series A: Statistics in Society}, 174(4):1007--1027, October 2011.

\bibitem[LLV20]{laddha2020strong}
Aditi Laddha, Yin~Tat Lee, and Santosh Vempala.
\newblock Strong self-concordance and sampling.
\newblock In {\em Proceedings of the 52nd annual ACM SIGACT symposium on theory of computing}, pages 1212--1222, 2020.

\bibitem[LNP12]{lewis_constraining_2012}
Nathan~E. Lewis, Harish Nagarajan, and Bernhard~O. Palsson.
\newblock Constraining the metabolic genotype--phenotype relationship using a phylogeny of in silico methods.
\newblock {\em Nature Reviews Microbiology}, 10(4):291--305, April 2012.

\bibitem[Lov99]{lovasz_hit-and-run_1999}
L{\'a}szl{\'o} Lov{\'a}sz.
\newblock Hit-and-run mixes fast.
\newblock {\em Mathematical Programming}, 86(3):443--461, December 1999.

\bibitem[LS93]{lovasz1993random}
L{\'a}szl{\'o} Lov{\'a}sz and Mikl{\'o}s Simonovits.
\newblock Random walks in a convex body and an improved volume algorithm.
\newblock {\em Random structures \& algorithms}, 4(4):359--412, 1993.

\bibitem[LS19]{lee2019solving}
Yin~Tat Lee and Aaron Sidford.
\newblock Solving linear programs with sqrt (rank) linear system solves.
\newblock {\em arXiv preprint arXiv:1910.08033}, 2019.

\bibitem[LST21]{lee_structured_2021}
Yin~Tat Lee, Ruoqi Shen, and Kevin Tian.
\newblock Structured logconcave sampling with a restricted {Gaussian} oracle.
\newblock In Mikhail Belkin and Samory Kpotufe, editors, {\em Proceedings of {Thirty} {Fourth} {Conference} on {Learning} {Theory}}, volume 134 of {\em Proceedings of {Machine} {Learning} {Research}}, pages 2993--3050. PMLR, August 2021.

\bibitem[LV06a]{lovasz_fast_2006}
Laszlo Lovasz and Santosh Vempala.
\newblock Fast algorithms for logconcave functions: Sampling, rounding, integration and optimization.
\newblock In {\em 2006 47th {Annual} {IEEE} {Symposium} on {Foundations} of {Computer} {Science} ({FOCS}'06)}, pages 57--68, Berkeley, CA, USA, 2006. IEEE.

\bibitem[LV06b]{Lov06hrcorner}
L\'{a}szl\'{o} Lov\'{a}sz and Santosh Vempala.
\newblock Hit-and-run from a corner.
\newblock {\em SIAM Journal on Computing}, 35(4):985--1005, 2006.

\bibitem[LV06c]{lovasz_simulated_2006}
L{\'a}szl{\'o} Lov{\'a}sz and Santosh Vempala.
\newblock Simulated annealing in convex bodies and an ${O}^* (n^4)$ volume algorithm.
\newblock {\em Journal of Computer and System Sciences}, 72(2):392--417, March 2006.

\bibitem[LV07]{lovasz2007geometry}
L{\'a}szl{\'o} Lov{\'a}sz and Santosh Vempala.
\newblock The geometry of logconcave functions and sampling algorithms.
\newblock {\em Random Structures \& Algorithms}, 30(3):307--358, 2007.

\bibitem[LV17]{lee_geodesic_2017}
Yin~Tat Lee and Santosh~S. Vempala.
\newblock Geodesic walks in polytopes.
\newblock In {\em Proceedings of the 49th {Annual} {ACM} {SIGACT} {Symposium} on {Theory} of {Computing}}, {STOC} 2017, pages 927--940, New York, NY, USA, 2017. Association for Computing Machinery.

\bibitem[LV18]{lee_RHMC_2018}
Yin~Tat Lee and Santosh~S. Vempala.
\newblock Convergence rate of {Riemannian} {Hamiltonian} {Monte} {Carlo} and faster polytope volume computation.
\newblock In {\em Proceedings of the 50th {Annual} {ACM} {SIGACT} {Symposium} on {Theory} of {Computing}}, {STOC} 2018, pages 1115--1121, New York, NY, USA, 2018. Association for Computing Machinery.

\bibitem[LV19]{lee_eldans_2019}
Yin~Tat Lee and Santosh~S. Vempala.
\newblock Eldan's stochastic localization and the {KLS} conjecture: Isoperimetry, concentration and mixing, January 2019.
\newblock arXiv:1612.01507 [cs, math].

\bibitem[MFWB22]{mou_efficient_2022}
Wenlong Mou, Nicolas Flammarion, Martin~J. Wainwright, and Peter~L. Bartlett.
\newblock An efficient sampling algorithm for non-smooth composite potentials.
\newblock {\em Journal of Machine Learning Research}, 23(233):1--50, 2022.

\bibitem[Min13]{minka2013expectation}
Thomas~P Minka.
\newblock Expectation propagation for approximate bayesian inference.
\newblock {\em arXiv preprint arXiv:1301.2294}, 2013.

\bibitem[MV22]{mangoubi2022faster}
Oren Mangoubi and Nisheeth Vishnoi.
\newblock Sampling from log-concave distributions with infinity-distance guarantees.
\newblock In S.~Koyejo, S.~Mohamed, A.~Agarwal, D.~Belgrave, K.~Cho, and A.~Oh, editors, {\em Advances in Neural Information Processing Systems}, volume~35, pages 12633--12646. Curran Associates, Inc., 2022.

\bibitem[MV23]{NEURIPS2023_mangoubi}
Oren Mangoubi and Nisheeth~K. Vishnoi.
\newblock Sampling from structured log-concave distributions via a soft-threshold {Dikin} walk.
\newblock In A.~Oh, T.~Naumann, A.~Globerson, K.~Saenko, M.~Hardt, and S.~Levine, editors, {\em Advances in Neural Information Processing Systems}, volume~36, pages 31908--31942. Curran Associates, Inc., 2023.

\bibitem[ND04]{neelon_bayesian_2004}
Brian Neelon and David~B. Dunson.
\newblock Bayesian isotonic regression and trend analysis.
\newblock {\em Biometrics}, 60(2):398--406, June 2004.

\bibitem[NR17]{narayanan2017efficient}
Hariharan Narayanan and Alexer Rakhlin.
\newblock Efficient sampling from time-varying log-concave distributions.
\newblock {\em Journal of Machine Learning Research}, 18(112):1--29, 2017.

\bibitem[PP14]{pakman_exact_2014}
Ari Pakman and Liam Paninski.
\newblock Exact {Hamiltonian} {Monte} {Carlo} for truncated multivariate {Gaussians}.
\newblock {\em Journal of Computational and Graphical Statistics}, 23(2):518--542, April 2014.

\bibitem[QH19]{qin_convergence_2019}
Qian Qin and James~P. Hobert.
\newblock Convergence complexity analysis of {Albert} and {Chib}'s algorithm for {Bayesian} probit regression.
\newblock {\em The Annals of Statistics}, 47(4), August 2019.

\bibitem[SN16]{saa_ll-achrb_2016}
Pedro~A. Saa and Lars~K. Nielsen.
\newblock ll-{ACHRB}: a scalable algorithm for sampling the feasible solution space of metabolic networks.
\newblock {\em Bioinformatics}, 32(15):2330--2337, August 2016.

\bibitem[SV16]{sachdeva2016mixing}
Sushant Sachdeva and Nisheeth~K Vishnoi.
\newblock The mixing time of the {Dikin} walk in a polytope---a simple proof.
\newblock {\em Operations Research Letters}, 44(5):630--634, 2016.

\bibitem[TTFT08]{tian_efficient_2008}
Guo-Liang Tian, Man-Lai Tang, Hong-Bin Fang, and Ming Tan.
\newblock Efficient methods for estimating constrained parameters with applications to regularized (lasso) logistic regression.
\newblock {\em Computational Statistics \& Data Analysis}, 52(7):3528--3542, March 2008.

\bibitem[Vem05]{Vempala2005Survey}
Santosh Vempala.
\newblock Geometric random walks: a survey.
\newblock {\em Combinatorial and Computational Geometry MSRI Publications Volume}, 52, 01 2005.

\bibitem[WFGP04]{wiback_monte_2004}
Sharon~J. Wiback, Iman Famili, Harvey~J. Greenberg, and Bernhard~{\O}. Palsson.
\newblock Monte {Carlo} sampling can be used to determine the size and shape of the steady-state flux space.
\newblock {\em Journal of Theoretical Biology}, 228(4):437--447, June 2004.

\end{thebibliography}
\end{document}